\documentclass[smallcondensed,envcountsect,envcountsame]{svjour3}
\usepackage{includes}
\RequirePackage{fix-cm}

\renewcommand{\vec}{\mathaccent"017E }
\DeclareMathSymbol{\Omega}{\mathord}{operators}{"0A}
\DeclareMathSymbol{\Gamma}{\mathord}{operators}{"00}
\DeclareMathSymbol{\Sigma}{\mathord}{operators}{"06}
\DeclareMathSymbol{\Delta}{\mathord}{operators}{"01}

\smartqed
\spnewtheorem{notation}[theorem]{Notation}{\bfseries}{\rmfamily}
\spnewtheorem{fact}[theorem]{Fact}{\bfseries}{\rmfamily}
\renewcommand{\figurename}{Figure}

\begin{document}
\journalname{Journal of Automated Reasoning}
\date{Received: n/a}
\title{A Formal \C{} Memory Model for Separation Logic}
\author{Robbert Krebbers}
\institute{Robbert Krebbers \at
	ICIS, Radboud University Nijmegen, The Netherlands\\
	Aarhus University, Denmark \\
	\email{mail@robbertkrebbers.nl}}

\maketitle{}

\begin{abstract}
\parfillskip=2em plus 1fill
The core of a formal semantics of an imperative programming language is a
memory model that describes the behavior of operations on the memory.
Defining a memory model that matches the description of \C{} in the
\Celeven{} standard is challenging because \C{} allows both \emph{high-level}
(by means of typed expressions) and \emph{low-level} (by means of bit
manipulation) memory accesses.
The \Celeven{} standard has restricted the interaction between these two levels
to make more effective compiler optimizations possible, on the expense of
making the memory model complicated.

We describe a formal memory model of the (non-concurrent part of the)
\Celeven{} standard that incorporates these restrictions, and at the same time
describes low-level memory operations.
This formal memory model includes a rich permission model to make it
usable in separation logic and supports reasoning about program
transformations.
The memory model and essential properties of it have been fully
formalized using the \Coq{} proof assistant.

\keywords{
ISO \Celeven{} Standard \and
\C{} Verification \and
Memory Models \and
Separation Logic \and
Interactive Theorem Proving \and
\Coq{}}
\end{abstract}

\section{Introduction}
\label{section:introduction}

A memory model is the core of a semantics of an imperative programming
language.
It models the memory states and describes the behavior of memory operations.
The main operations described by a \C{} memory model are:

\begin{itemize}
\item Reading a value at a given address.
\item Storing a value at a given address.
\item Allocating a new object to hold a local variable or storage obtained via
	\lstinline|malloc|.
\item Deallocating a previously allocated object.
\end{itemize}

Formalizing the \Celeven{} memory model in a faithful way is challenging
because \C{} features both \emph{low-level} and \emph{high-level} data access.
Low-level data access involves unstructured and untyped byte representations
whereas high-level data access involves typed abstract values such as arrays,
structs and unions.

This duality makes the memory model of \C{} more complicated than the memory
model of nearly any other programming language.
For example, more mathematically oriented languages such as \Java{}
and \ML{} feature only high-level data access, in which case the memory can
be modeled in a relatively simple and structured way, whereas assembly languages
feature only low-level data access, in which case the memory can be modeled
as an array of bits.

The situation becomes more complicated as the \Celeven{} standard allows
compilers to perform optimizations based on a high-level view of
data access that are inconsistent with the traditional low-level view
of data access.
This complication has lead to numerous ambiguities in the standard text
related to aliasing, uninitialized memory, end-of-array pointers and
type-punning that cause problems for \C{} code when compiled with widely used
compilers.
See for example the message~\cite{mac:01} on the standard committee's mailing
list, Defect Reports \#236, \#260, and \#451~\cite{iso:09}, and the various
examples in this paper.

\paragraph{Contribution.}
This paper describes the \CHtwoO{} memory model, which is part of the
the \CHtwoO{} project~\cite{kre:wie:11,kre:wie:13,%
kre:13:cpp,kre:14:popl,kre:ler:wie:14,kre:14:vstte,kre:wie:15,kre:15:phd}.
\CHtwoO{} provides an operational, executable and axiomatic
semantics in \Coq{} for a large part of the non-concurrent fragment of \C{},
based on the official description of \C{} as given by the \Celeven{}
standard~\cite{iso:12}.

The key features of the \CHtwoO{} memory model are as follows:

\begin{itemize}
\item \textbf{Close to \Celeven{}.}
	\CHtwoO{} is faithful to the \Celeven{} standard in order to be
	compiler independent.
	When one proves something about a given program with respect to
	\CHtwoO, it should behave that way with \emph{any} \Celeven{}
	compliant compiler (possibly restricted to certain implementation defined
	choices).
\item \textbf{Static type system.}
	Given that \C{} is a statically typed language, \CHtwoO{} does not only
	capture the dynamic semantics of \Celeven{} but also its type system.
	We have established properties such as	type preservation of the memory
	operations.
\item \textbf{Proof infrastructure.}
	All parts of the \CHtwoO{} memory model and semantics have been
	formalized in \Coq{} (without axioms).
	This is essential for its application to program verification in
	proof assistants.
	Also, considering the significant size of \CHtwoO{} and its memory model, proving
	metatheoretical properties of the language would have been intractable
	without the support of a proof assistant.

	Despite our choice of using \Coq, we believe that nearly all parts of
	\CHtwoO{} could be formalized in any proof assistant based on higher-order
	logic.
\item \textbf{Executable.}
	To obtain more confidence in the accuracy of \CHtwoO{} with respect
	to \Celeven{}, the \CHtwoO{} memory model is executable.
	An executable memory model allows us to test the \CHtwoO{} semantics on
	example programs and to	compare the	behavior with that of widely used
	compilers~\cite{kre:wie:15,kre:15:phd}.
\item \textbf{Separation logic.}
	In order to reason about concrete \C{} programs, one needs a program logic.
	To that end, the \CHtwoO{} memory model incorporates a complex permission
	model suitable for separation logic.
	This permission system, as well as the memory model itself, forms a
	separation algebra.
\item \textbf{Memory refinements.}
	\CHtwoO{} has an expressive notion of memory refinements that relates memory
	states.
	All memory operations are proven invariant under this notion.
	Memory refinements form a general way to validate many common-sense
	properties of the memory model in a formal way.
	They also open the door to reasoning about program transformations, which
	is useful if one were to use the memory model as part of a verified
	compiler front-end.
\end{itemize}

This paper is an extended version of previously published conference papers at
CPP~\cite{kre:13:cpp} and VSTTE~\cite{kre:14:vstte}.
In the time following these two publications, the memory model has been extended
significantly and been integrated into an operational, executable and
axiomatic semantics~\cite{kre:wie:15,kre:15:phd}.
The memory model now supports more features, various improvements to the
definitions have been carried out, and more properties have been formally
proven as part of the \Coq{} development.

Parts of this paper also appear in the author's PhD thesis~\cite{kre:15:phd},
which describes the entire \CHtwoO{} project including its operational,
executable and axiomatic semantics, and metatheoretical results about these.

\paragraph{Problem.}
The \Celeven{} standard gives compilers a lot of freedom in what behaviors a
program may have~\cite[3.4]{iso:12}.
It uses the following notions of under-specification:

\begin{itemize}
\item \emph{Unspecified behavior}:
  two or more behaviors are allowed.
  For example: the execution order of expressions.
  The choice may vary for each use of the construct.
\item \emph{Implementation defined behavior}:
  like unspecified behavior, but the compiler has to document its choice.
  For example: size and endianness of integers.
\item \emph{Undefined behavior:}
  the standard imposes no requirements at all,
  the program is even allowed to crash.
  For example: dereferencing a \lstinline|NULL| pointer, or signed integer
  overflow.
\end{itemize}

\noindent
Under-specification is used extensively to make \C{} portable, and to allow
compilers to generate fast code.
For example, when dereferencing a pointer, no code has to be generated to check
whether the pointer is valid or not.
If the pointer is invalid (\lstinline|NULL| or a dangling pointer), the
compiled program may do something arbitrary instead of having to exit with a
\lstinline|NullPointerException| as in \Java.
Since the \CHtwoO{} semantics intends to be a formal version of the \Celeven{}
standard, it has to capture the behavior of \emph{any} \C{} compiler, and thus
has to take \emph{all} under-specification seriously (even if that makes the
semantics complex).

Modeling under-specification in a formal semantics is folklore:
unspecified behavior corresponds to non-determinism,
implementation defined behavior corresponds to parametrization, and
undefined behavior corresponds to a program having no semantics.
However, the extensive amount of underspecification in the \Celeven{}
standard~\cite[Annex J]{iso:12}, and especially that with respect to the memory
model, makes the situation challenging.
We will give a motivating example of subtle underspecification in the
introduction of this paper.
Section~\ref{section:challenges} provides a more extensive overview.

\paragraph{Motivating example.}
A drawback for efficient compilation of programming languages with pointers
is \emph{aliasing}.
Aliasing describes a situation in which multiple pointers refer to the same
object.
In the following example the pointers \lstinline|p| and \lstinline|q| are
said to be \emph{aliased}.

\begin{lstlisting}
int x; int *p = &x, *q = &x;
\end{lstlisting}

The problem of aliased pointers is that writes through one pointer may effect the
result of reading through the other pointer.
The presence of aliased pointers therefore often disallows one to change the
order of instructions.
For example, consider:

\begin{lstlisting}
int f(int *p, int *q) {
  int z = *q; *p = 10; return z;
}
\end{lstlisting}

When \lstinline|f| is called with pointers \lstinline|p|
and \lstinline|q| that are aliased, the assignment to \lstinline|*p| also
affects \lstinline|*q|.
As a result, one cannot transform the function body of \lstinline|f|
into the shorter \lstinline|*p = 10; return (*q);|.
The shorter function will return 10 in case \lstinline|p| and \lstinline|q|
are aliased, whereas the original \lstinline|f| will always return the original value of
\lstinline|*q|.

Unlike this example, there are many situations in which pointers
can be assumed \emph{not} to alias.
It is essential for an optimizing compiler to determine where aliasing cannot
occur, and use this information to generate faster code.
The technique of determining whether pointers can alias or not is called 
\emph{alias analysis}.

In \emph{type-based alias analysis}, type information is used to determine
whether pointers can alias or not.
Consider the following example:

\begin{lstlisting}
short g(int *p, short *q) {
  short z = *q; *p = 10; return z;
}
\end{lstlisting}

Here, a compiler is allowed to assume that \lstinline|p| and \lstinline|q|
are not aliased because they point to objects of different types.
The compiler is therefore allowed to transform the function body of
\lstinline|g| into the shorter \lstinline|*p = 10; return (*q);|.

The peculiar thing is that the \C{} type system does not statically enforce
the property that pointers to objects of different types are not aliased.
A union type can be used to create aliased pointers to different types:

\begin{lstlisting}
union int_or_short { int x; short y; } u = { .y = 3 };
int *p = &u.x;   // p points to the x variant of u
short *q = &u.y; // q points to the y variant of u
return g(p, q);  // g is called with aliased pointers p and q
\end{lstlisting}

\label{page:effective_types}
The above program is valid according to the rules of the \Celeven{} type system,
but has undefined behavior during execution of \lstinline|g|.
This is caused by the standard's notion of \emph{effective
types}~\cite[6.5p6-7]{iso:12} (also called \emph{strict-aliasing restrictions})
that assigns undefined behavior to incorrect usage of aliased pointers to
different types.

We will inline part of the function body of \lstinline|g| to indicate the
incorrect usage of aliased pointers during the execution of the example.

\begin{lstlisting}
union int_or_short { int x; short y; } u = { .y = 3 };
int *p = &u.x;   // p points to the x variant of u
short *q = &u.y; // q points to the y variant of u
@// g(p, q) is called, the body is inlined@
short z = *q;    // u has variant y and is accessed through y -> OK
*p = 10;         // u has variant y and is accessed through x -> bad
\end{lstlisting}

The assignment \lstinline|*p = 10| violates the rules for effective types.
The memory area where \lstinline|p| points to contains a union whose variant
is \lstinline|y| of type \lstinline|short|, but is accessed through a pointer
to variant \lstinline|x| of type \lstinline|int|.
This causes undefined behavior.

Effective types form a clear tension between the low-level and high-level way
of data access in \C.
The low-level representation of the memory is inherently untyped and
unstructured and therefore does not contain any information about
variants of unions.
However, the standard treats the memory as if it were typed.

\paragraph{Approach.}
Most existing \C{} formalizations (most notably Norrish~\cite{nor:99},
Leroy \etal~\cite{ler:09:cacm,ler:app:bla:ste:12} and Ellison and
Ro\c{s}u~\cite{ell:ros:12}) use an unstructured untyped memory model
where each object in the formal memory model consists of an array of bytes.
These formalizations therefore cannot assign undefined behavior to
violations of the rules for effective types, among other things.

In order to formalize the interaction between low-level and high-level data
access, and in particular effective types, we represent the formal memory state
as a forest of well-typed trees whose structure corresponds to the structure of
data types in \C.
The leaves of these trees consist of bits to capture low-level aspects of the
language.

The key concepts of our memory model are as follows:

\begin{itemize}
\item \emph{Memory trees} (Section~\ref{section:ctrees}) are used to
	represent each object in memory.
	They are abstract trees whose structure corresponds to the shape of \C{}
	data types.
	The memory tree of \lstinline|struct S { short x, *r; } s = { 33, &s.x }|
	might be (the precise shape and the bit representations are implementation
	defined):
	
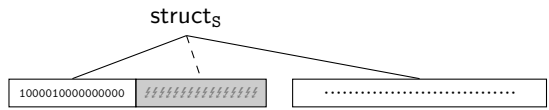
\begin{equation*}
	\begin{tikzpicture}[
		node distance=-\pgflinewidth,
		every node/.style={shape=rectangle,minimum height=1.1em,minimum width=5em}
	]
	\node (S) {\(\MStructSym {\mathtt S}\)};
	\node[draw] at (-1.5,-1) (x) {\scalebox{.50}{\texttt{1000010000000000}}};
	\node[draw,fill=gray,fill opacity=0.4,text opacity=1,right=of x] (pad)
		{\scalebox{.60}{\ntimes {16} \BIndet}};
	\node[draw,minimum width=10em,right=1em of pad] (r)
		{\scalebox{0.85}{\ntimes {32} {\(\cdot\)}}};
	\draw (S.south) -- (x.north);
	\draw[dashed] (S.south) -- (pad.north);
	\draw (S.south) -- (r.north);
	\end{tikzpicture}
	\end{equation*}

	The leaves of memory trees contain permission annotated bits
	(Section~\ref{section:bits}).
	Bits are represented symbolically: the integer value \lstinline|33|
	is represented as its binary representation \lstinline|1000010000000000|,
	the padding bytes as symbolic \emph{indeterminate} bits \(\BIndet\)
	(whose actual value should not be used), and the pointer \lstinline|&s.x| as
	a sequence of symbolic \emph{pointer bits}.
	
	The \emph{memory} itself is a forest of memory trees.
	Memory trees are explicit about type information (in particular the variants
	of unions) and thus give rise to a natural formalization of effective types.
\item \emph{Pointers} (Section~\ref{section:pointers}) are formalized using
	paths through memory trees.
	Since we represent pointers as paths, the formal representation contains
	detailed information about how each pointer has been obtained (in particular
	which variants of unions were used).
	A detailed formal representation of pointers is essential to describe
	effective types.
\item \emph{Abstract values} (Definition~\ref{section:values}) are trees
	whose structure is similar to memory trees, but have base values
	(mathematical integers and pointers) on their leaves.
	The abstract value of
	\lstinline|struct S { short x, *r; } s = { 33, &s.x }| is:
	
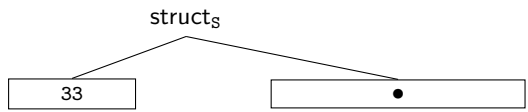
\begin{equation*}
	\begin{tikzpicture}[
		node distance=-\pgflinewidth,
		every node/.style={shape=rectangle,minimum height=1.1em,minimum width=5em}
	]
	\node (S) {\(\VStructSym {\mathtt S}\)};
	\node[draw] at (-1.5,-1) (x) {\small\texttt{33}};
	\node[draw,minimum width=10em,right=5.3em of x] (r) {\(\bullet\)};
	\draw (S.south) -- (x.north);
	\draw (S.south) -- (r.north);
	\end{tikzpicture}
	\end{equation*}

	Abstract values hide internal details of the memory such as permissions,
	padding and	object representations.
	They are therefore used in the external interface of the memory model
	and throughout the operational semantics.
\end{itemize}

Memory trees, abstract values and bits with permissions can be converted into
each other.
These conversions are used to define operations internal to the memory model.
However, none of these conversions are bijective because different information
is materialized in these three data types:

\begin{center}
\begin{tabular}{l|c|c|c}
& Abstract values & Memory trees & Bits with permissions \\ \hline\hline
Permissions &  & \checkmark & \checkmark \\
Padding & & always \(\BIndet\) & \checkmark \\
Variants of union  & \checkmark & \checkmark & \\
Mathematical values & \checkmark & &
\end{tabular}
\end{center}

This table indicates that abstract values and sequences of bits are complementary.
Memory trees are a middle ground, and therefore suitable to describe both the
low-level and high-level aspects of the \C{} memory.

\paragraph{Outline.}
This work presents an executable mathematically precise version of a large part
of the (non-concurrent) \C{} memory model.
\begin{itemize}
\item Section~\ref{section:challenges} describes some challenges that a
	\Celeven{} memory model should address; these include end-of-array pointers,
	byte-level operations, indeterminate memory,	and type-punning.	
\item Section~\ref{section:types} describes the types of \C.
	Our formal development is parametrized by an abstract interface to
	characterize implementation defined behavior.
\item Section~\ref{section:permissions} describes the permission model using
	a variant of separation algebras that is suitable for	formalization in
	\Coq.
	The permission model is built compositionally from simple separation
	algebras. 
\item Section~\ref{section:memory} contains the main part of this paper,
	it describes a memory model that can accurately deal with the challenges
	posed in Section~\ref{section:challenges}.
\item Section~\ref{section:formal_proofs} demonstrates that our memory model is
	suitable for formal proofs.
	We prove that the standard's notion	of effective types has the	desired effect
	of allowing type-based alias analysis (Section~\ref{section:aliasing}),
	we present a method to reason compositionally
	about memory transformations (Section~\ref{section:refinements}), and
	prove that the memory model has a separation algebra structure
	(Section~\ref{section:memory_separation}).
\item Section~\ref{section:coq} describes the \Coq{} formalization: all proofs
	about our memory model have been fully formalized using \Coq.
\end{itemize}

\noindent
As this paper describes a large formalization effort, we often just give
representative parts of definitions.
The interested reader can find all details online as part of our \Coq{}
formalization at:
\begin{center}
 \url{http://robbertkrebbers.nl/research/ch2o/}.
\end{center}

\section{Notations}
\label{section:notations}

This section introduces some common mathematical notions and notations that
will be used throughout this paper.

\begin{definition}
We let \(\nat\) denote the type of \emph{natural numbers} (including \(0\)),
let \(\Z\) denote the type of \emph{integers}, and let \(\Q\) denote the type of
\emph{rational numbers}.
We let \(i \divides j\) denote that \(i \in \nat\) is a \emph{divisor} of
\(j \in \nat\).
\end{definition}

\begin{definition}
We let \(\Prop\) denote the type of \emph{propositions}, and let \(\bool\)
denote the type of \emph{Booleans} whose elements are \(\true\) and \(\false\).
Most propositions we consider have a corresponding Boolean-valud decision
function.
In \Coq{} we use type classes to keep track of these correspondences, but in
this paper we leave these correspondences implicit.
\end{definition}

\begin{definition}
We let \(\option A\) denote the \emph{option type over \(A\)}, whose elements
are inductively defined as either \(\None\) or \(\Some x\) for some \(x \in A\).
We implicitly lift operations to operate on the option type, and often omit
cases of definitions that yield \(\None\).
This is formally described using the \emph{option monad} in the \Coq{}
formalization.
\end{definition}

\begin{definition}
A \emph{partial function \(f\) from \(A\) to \(B\)} is a function
\(f : A \to \option B\).
\end{definition}

\begin{definition}
A partial function \(f\) is called \emph{a finite partial function} or a
\emph{finite map} if its \emph{domain}
\(\mdom f \defined \{ x \separator \exists y\in B\wsdot f\,x = y\} \) is finite.
The type of finite partial functions is denoted as \(\map A B\).
The operation \mbox{\(\minsert x y f\)} yields \(f\) with the value \(y\)
for argument \(x\).
\end{definition}

\begin{definition}
We let \(A \times B\) denote \emph{the product of types \(A\) and \(B\).}
Given a \emph{pair} \(\pair x y \in A \times B\), we let
\(\fst {\pair x y} \defined x\) and \(\snd {\pair x y} \defined y\)
denote the \emph{first} and \emph{second projection} of \(\pair x y\).
\end{definition}

\begin{definition}
We let \(\lst A\) denote the \emph{list type over \(A\)}, whose elements are
inductively defined as either \(\nil\) or \(\cons x {\vec x}\) for some
\(x \in A\) and \(\vec x \in \lst A\).
We let \(x_i \in A\) denote the \(i\)th element of a list
\(\vec x \in \lst A\) (we count from 0).
Lists are sometimes denoted as \(\listlit {x_0,\dotsc,x_{n-1}} \in \lst A\)
for \(x_0,\dotsc,x_{n-1} \in A\).

We use the following operations on lists:
\begin{itemize}
\item We often implicitly lift a function \(f : A_0 \to \dotsb \to A_n\)
	point-wise to the function \(f : \lst {A_0} \to \dotsb \to \lst {A_n}\).
	The resulting list is truncated to the length of the smallest input
	list in case \(n > 1\).
\item We often implicitly lift a predicate \(P : A_0 \to A_{n-1} \to \Prop\)
	to the predicate
	\(P : \lst {A_0} \to \dotsb \to \lst {A_{n-1}} \to \Prop\) that guarantees
	that \(P\) holds for all (pairs of) elements of the list(s).
	The lifted predicate requires all lists to have the same length in case
	\(n > 1\).
\item We let \(\length {\vec x} \in \nat\) denote the length of
	\(\vec x \in \lst A\).
\item We let \(\sublist i j {\vec x} \in \lst A\) denote the sublist
	\(x_i \dotsc x_{j-1}\) of \(\vec x \in \lst A\).
\item We let \(\replicate n x \in \lst A\) denote the list consisting of
	\(n\) times	\(x \in A\).
\item We let \(\sublistpad i j y {\vec x} \in \lst A\) denote the sublist
	\(x_i \dotsc x_{j-1}\) of \(\vec x \in \lst A\) which is padded with
	\(y \in A\) in case \(\vec x\) is too short.
\item Given lists \(\vec x \in \lst A\) and \(\vec y \in \lst B\) with
	\(\length {\vec x} = \length {\vec y}\), we let
	\(\vv {xy} \in \lst{(A \times B)}\) denote the point-wise pairing of
	\(\vec x\) and \(\vec y\).
\end{itemize}
\end{definition}

\section{Challenges}
\label{section:challenges}

This section illustrates a number of subtle forms of underspecification in \C{}
by means of example programs, their bizarre behaviors exhibited by widely
used \C{} compilers, and their treatment in \CHtwoO.
Many of these examples involve delicacies due to the interaction between the
following two ways of accessing data:

\begin{itemize}
\item In a \emph{high-level} way using arrays, structs and unions.
\item In a \emph{low-level} way using unstructured and untyped byte
	representations.
\end{itemize}

The main problem is that compilers use a high-level view of data access
to perform optimizations whereas both programmers and traditional memory models
expect data access to behave in a concrete low-level way.

\subsection{Byte-level operations and object representations}
\label{section:challenge:byte_level}
Apart from \emph{high-level} access to objects in memory by means of typed
expressions, \C{} also allows \emph{low-level} access by means of byte-wise
manipulation.
Each object of type \(\tau\) can be interpreted as an \lstinline|unsigned char|
array of length \(\mathtt{sizeof(\tau)}\), which is called the \emph{object
representation}~\cite[6.2.6.1p4]{iso:12}.
Let us consider:

\begin{lstlisting}
struct S { short x; short *r; } s1 = { 10, &s1.x };
unsigned char *p = (unsigned char*)&s1;
\end{lstlisting}

On 32-bits computing architectures such as \xeightysix{} (with
\lstinline|_Alignof(short*)|\(=4\)), the object representation of \lstinline|s1|
might be:
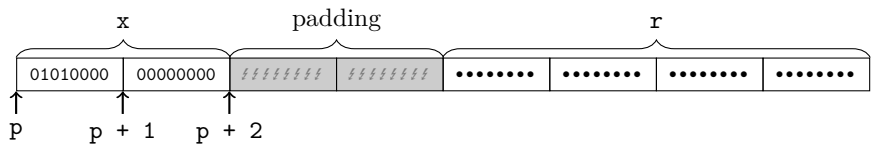
\begin{equation*}
\begin{tikzpicture}[
	start chain=going right,
	node distance=-\pgflinewidth,
	every node/.style={shape=rectangle,minimum height=1.3em, minimum width=4.2em},
	brace/.style={decorate,decoration={brace,amplitude=0.7em}},
	lab/.style={pos=0.5,anchor=south,yshift=0.7em}
]
\node[draw,on chain] (x1) {\scalebox{.75}{\texttt{01010000}}};
\node[draw,on chain] (x2) {\scalebox{.75}{\texttt{00000000}}};
\node[draw,on chain,fill=gray,fill opacity=0.4,text opacity=1] (pad1)
	{\scalebox{.65}{\ntimes 8 {\BIndet\,}}};
\node[draw,on chain,fill=gray,fill opacity=0.4,text opacity=1] (pad2)
	{\scalebox{.65}{\ntimes 8 {\BIndet\,}}};
\node[draw,on chain] (r1) {\scalebox{.60}{\ntimes 8 {\(\bullet\)\,}}};
\node[draw,on chain] {\scalebox{.60}{\ntimes 8 {\(\bullet\)\,}}};
\node[draw,on chain] {\scalebox{.60}{\ntimes 8 {\(\bullet\)\,}}};
\node[draw,on chain] (r2) {\scalebox{.60}{\ntimes 8 {\(\bullet\)\,}}};
\draw[brace] (x1.north west) -- (x2.north east) node [lab] {\lstinline|x|};
\draw[brace] (pad1.north west) -- (pad2.north east) node [lab] {padding};
\draw[brace] (r1.north west) -- (r2.north east)	node [lab] {\lstinline|r|};
\draw[thick,<-] (x1.south west) -- node[below,pos=1]{\lstinline|p|} +(0,-0.3);
\draw[thick,<-] (x2.south west) -- node[below,pos=1]{\lstinline|p + 1|} +(0,-0.3);
\draw[thick,<-] (pad1.south west) -- node[below,pos=1]{\lstinline|p + 2|} +(0,-0.3);
\end{tikzpicture}
\end{equation*}

The above object representation contains a hole due to \emph{alignment} of
objects.
The bytes belonging such holes are called \emph{padding bytes}.

Alignment is the way objects are arranged in memory.
In modern computing architectures, accesses to addresses that are a multiple of
a word sized chunk (often a multiple of 4 bytes on a 32-bits computing
architecture) are significantly faster due to the way the processor interacts
with the memory.
For that reason, the \Celeven{} standard has put restrictions on the addresses
at which objects may be allocated~\cite[6.2.8]{iso:12}.
For each type \(\tau\), there is an implementation defined integer constant
\lstinline|_Alignof($\tau$)|, and objects of type \(\tau\) are required to be
allocated at addresses that are a multiple of that constant.
In case \lstinline|_Alignof(short*)|\(=4\), there are thus two bytes of padding
in between the fields of \lstinline|struct S|.

An object can be copied by copying its object representation.
For example, the struct \lstinline|s1| can be copied to \lstinline|s2| as
follows:

\begin{lstlisting}
struct S { short x; short *r; } s1 = { 10, &s1.x };
struct S s2;
for (size_t i = 0; i < sizeof(struct S); i++)
  ((unsigned char*)&s2)[i] = ((unsigned char*)&s1)[i];
\end{lstlisting}

In the above code, \lstinline|size_t| is an unsigned integer type, which is able
to hold the results of the \lstinline|sizeof| operator~\cite[7.19p2]{iso:12}.

Manipulation of object representations of structs also involves access to
padding bytes, which are not part of the high-level representation.
In particular, in the example the padding bytes are also being
copied.
The problematic part is that padding bytes have indeterminate values, whereas
in general, reading an indeterminate value has undefined behavior 
(for example, reading from an uninitialized \lstinline|int| variable is
undefined).
The \Celeven{} standard provides an exception for
\lstinline|unsigned char|~\cite[6.2.6.1p5]{iso:12}, and the above example thus
has defined behavior.

Our memory model uses a symbolic
representation of bits (Definition~\ref{definition:bits}) to distinguish
determinate and indeterminate memory.
This way, we can precisely keep track of the situations in which access to
indeterminate memory is permitted.

\subsection{Padding of structs and unions}
\label{section:challenge:padding}

The following excerpt from the \Celeven{} standard points out another challenge
with respect to padding bytes~\cite[6.2.6.1p6]{iso:12}:

\begin{quoting}
When a value is stored in an object of structure or union type,
including in a member object, the bytes of the object representation that
correspond to any padding bytes take unspecified values.
\end{quoting}

\noindent Let us illustrate this difficulty by an example:

\begin{lstlisting}
struct S { char x; char y; char z; };
void f(struct S *p) { p->x = 0; p->y = 0; p->z = 0; }
\end{lstlisting}

On architectures with \lstinline|sizeof(struct S) = 4|, objects of type
\lstinline|struct S| have one byte of padding.
The object representation may be as follows:
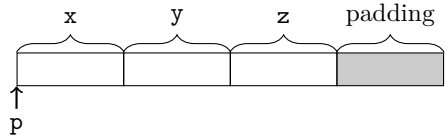
\begin{equation*}
\begin{tikzpicture}[
	start chain=going right,
	node distance=-\pgflinewidth,
	every node/.style={shape=rectangle,minimum height=1.3em, minimum width=4.2em},
	brace/.style={decorate,decoration={brace,amplitude=0.7em}},
	lab/.style={pos=0.5,anchor=south,yshift=0.7em}
]
\node[draw,on chain] (x) {};
\node[draw,on chain] (y) {};
\node[draw,on chain] (z) {};
\node[draw,on chain,fill=gray,fill opacity=0.4] (pad) {};
\draw[brace] (x.north west) -- (x.north east) node [lab] {\lstinline|x|};
\draw[brace] (y.north west) -- (y.north east) node [lab] {\lstinline|y|};
\draw[brace] (z.north west) -- (z.north east) node [lab] {\lstinline|z|};
\draw[brace] (pad.north west) -- (pad.north east) node [lab] {padding};
\draw[thick,<-] (x.south west) -- node[below,pos=1]{\lstinline|p|} +(0,-0.3);
\end{tikzpicture}
\end{equation*}

Instead of compiling the function \lstinline|f| to three store
instructions for each field of the struct, the \Celeven{} standard allows a
compiler to use a single instruction to store zeros to the entire struct.
This will of course affect the padding byte.
Consider:

\begin{lstlisting}
struct S s = { 1, 1, 1 };
((unsigned char*)&s)[3] = 10;
f(&s);
return ((unsigned char*)&s)[3];
\end{lstlisting}

Now, the assignments to fields of \lstinline|s| by the function \lstinline|f|
affect also the padding bytes of \lstinline|s|,
including the one \lstinline|((unsigned char*)&s)[3]| that we have assigned to.
As a consequence, the returned value is unspecified.

From a high-level perspective this behavior makes sense.
Padding bytes are not part of the abstract value of a struct, so their actual
value should not matter.
However, from a low-level perspective it is peculiar.
An assignment to a specific field of a struct affects the object representation of
parts not assigned to.

None of the currently existing \C{} formalizations describes this
behavior correctly.
In our tree based memory model we enforce that padding bytes always have an
indeterminate value, and in turn we have the desired behavior implicitly.
Note that if the function call \lstinline|f(&s)| would have been removed, the
behavior of the example program remains unchanged in \CHtwoO.

\subsection{Type-punning}
\label{section:challenge:type_punning}
Despite the rules for effective types, it is under certain conditions
nonetheless allowed to access a union through another variant than the current
one.
Accessing a union through another variant is called \emph{type-punning}.
For example:

\begin{lstlisting}
union int_or_short { int x; short y; } u = { .x = 3 };
printf("%d\n", u.y);
\end{lstlisting}

This code will reinterpret the bit representation of the \lstinline|int| value
\lstinline|3| of \lstinline|u.x| as a value of type \lstinline|short|.
The reinterpreted value that is printed is implementation defined (on
architectures where \lstinline|short|s do not have trap values).

Since \Celeven{} is ambiguous about the exact conditions under which
type-punning is allowed\footnote{The term \emph{type-punning} merely appears in
a footnote\cite[footnote 95]{iso:12}.
There is however the related \emph{common initial sequence}
rule~\cite[6.5.2.3]{iso:12}, for which the
\Celeven{} standard uses the notion of \emph{visible}.
This notion is not clearly defined either.}, we follow the
interpretation by the \GCC{} documentation~\cite{gnu:11}:

\begin{quoting}
Type-punning is allowed, provided the memory is accessed through the union type.
\end{quoting}

According to this interpretation the above program indeed has implementation
defined behavior because the variant \lstinline|y| is accessed via the
expression \lstinline|u.y| that involves the variable \lstinline|u| of the
corresponding union type.

However, according to this interpretation, type-punning via a pointer
to a specific variant of a union type yields undefined behavior.
This is in agreement with the rules for effective types.
For example, the following program has undefined behavior.

\begin{lstlisting}
union int_or_short { int x; short y; } u = { .x = 3 };
short *p = &u.y;
printf("%d\n", *p);
\end{lstlisting}

We formalize the interpretation of \Celeven{} by \GCC{} by decorating pointers
and l-values to subobjects with annotations (Definition~\ref{definition:pointers}).
When a pointer to a variant of a union is stored in memory, or used as
the argument of a function, the annotations are changed to ensure that
type-punning no longer has defined behavior via that pointer.
In Section~\ref{section:aliasing} we formally establish that this approach
is correct by showing that a compiler can perform type-based alias analysis
(Theorem~\ref{theorem:aliasing} on page~\pageref{theorem:aliasing}).

\subsection{Indeterminate memory and pointers}
\label{section:challenge:indeterminate_pointers}
A pointer value becomes indeterminate when the object it points to
has reached the end of its lifetime~\cite[6.2.4]{iso:12} (it has gone out
of scope, or has been deallocated).
Dereferencing an indeterminate pointer has of course undefined behavior because
it no longer points to an actual value.
However, not many people are aware that using an indeterminate pointer in
pointer arithmetic and pointer comparisons also yields undefined behavior.
Consider:

\begin{lstlisting}
int *p = malloc(sizeof(int)); assert (p != NULL);
free(p);
int *q = malloc(sizeof(int)); assert (q != NULL);
if (p == q) { // undefined, p is indeterminate due to the free
  *q = 10;
  *p = 14;
  printf("%d\n", *q); // p and q alias, expected to print 14
}
\end{lstlisting}

In this code \lstinline|malloc(sizeof(int))| yields a pointer to a newly
allocated memory area that may hold an integer, or yields a \cNULL{}
pointer in case no memory is available.
The function \lstinline|free| deallocates memory allocated by
\lstinline|malloc|.
In the example we assert that both calls to \lstinline|malloc| succeed.

After execution of the second call to \lstinline|malloc| it may happen that
the memory area of the first call to \lstinline|malloc| is reused:
we have used \lstinline|free| to deallocate it after all.
This would lead to the following situation in memory:
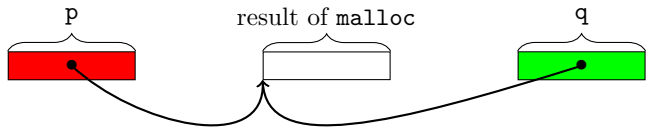
\begin{equation*}
\begin{tikzpicture}[
	every node/.style={shape=rectangle,minimum height=1.1em,minimum width=5em},
	brace/.style={decorate,decoration={brace,amplitude=0.7em}},
	lab/.style={pos=0.5,anchor=south}
]
\node[fill=red,draw] (p) {\(\bullet\)};
\draw[brace] (p.north west) -- (p.north east)
	node [lab,yshift=0.7em] {\lstinline|p|};
\node[draw,right=5em of p] (malloc) {};
\draw[brace] (malloc.north west) -- (malloc.north east)
	node [lab,yshift=0.7em] {result of \lstinline|malloc|};
\node[fill=green,draw,right=5em of malloc] (q) {\(\bullet\)};
\draw[brace] (q.north west) -- (q.north east)
	node [lab,yshift=0.7em] {\lstinline|q|};
\draw[thick,->] (p.center) ..
	controls ($(p.south east) + (0,-0.6)$)
	and ($(malloc.south west) + (0.05,-1)$) .. (malloc.south west);
\draw[thick,->] (q.center) ..
	controls ($(malloc.south east) + (0,-0.6)$)
	and ($(malloc.south west) + (0.05,-1)$) .. (malloc.south west);
\end{tikzpicture}
\vspace{-0.7em}
\end{equation*}

Both \GCC{} (version 4.9.2) or \clang{} (version 3.5.0) use the fact that
\lstinline|p| and \lstinline|q| are obtained via different calls to
\lstinline|malloc| as a license to assume
that \lstinline|p| and \lstinline|q| do not alias.
As a result, the value 10 of \lstinline|*q| is inlined, and the program prints
the value 10 instead of the naively expected value 14.

The situation becomes more subtle because when the object a pointer points to
has been deallocated, not just the argument of \lstinline|free| becomes
indeterminate, but also all other copies of that pointer.
This is therefore yet another example where high-level representations interact
subtly with their low-level counterparts.

In our memory model we represent pointer values symbolically
(Definition~\ref{definition:pointers}), and keep track of memory areas
that have been previously deallocated.
The behavior of operations like \lstinline|==| depends on the memory state,
which allows us to accurately capture the described undefined behaviors.

\subsection{End-of-array pointers}
\label{section:challenge:end_of_array}
The way the \Celeven{} standard deals with pointer equality is subtle.
Consider the following excerpt~\cite[6.5.9p6]{iso:12}:

\begin{quoting}
Two pointers compare equal if and only if  [...]  or one is a pointer to one
past the end of one array object and the other is a pointer to the start of
a different array object that happens to immediately follow the first array
object in the address space.
\end{quoting}

End-of-array pointers are peculiar because they cannot be dereferenced, they do
not point to any value after all.
Nonetheless, end-of-array are commonly used when looping through arrays.

\begin{lstlisting}
int a[4] = { 0, 1, 2, 3 };
int *p = a;
while (p < a + 4) { *p += 1; p += 1; }
\end{lstlisting}

The pointer \lstinline|p| initially refers to the first element of the array
\lstinline|a|.
The value \lstinline|p| points to, as well as \lstinline|p| itself, is being
increased as long as \lstinline|p| is before the end-of-array pointer
\lstinline|a + 4|.
This code thus increases the values of the array \lstinline|a|.
The initial state of the memory is displayed below:

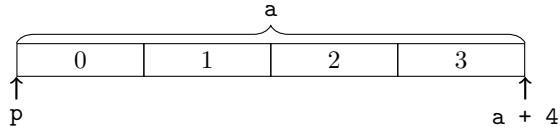
\begin{equation*}
\begin{tikzpicture}[
	start chain=going right,
	node distance=-\pgflinewidth,
	every node/.style={shape=rectangle,minimum height=1.1em,minimum width=5em},
	brace/.style={decorate,decoration={brace,amplitude=0.7em}},
	lab/.style={pos=0.5,anchor=south}
]
\node[draw,on chain] (a) {0};
\node[draw,on chain] {1};
\node[draw,on chain] (a2) {2};
\node[draw,on chain] (a3) {3};
\draw[brace] (a.north west) -- (a3.north east)
	node [lab,yshift=0.7em] {\lstinline|a|};
\draw[thick,<-] (a.south west) -- node[below,pos=1]{\lstinline|p|} +(0,-0.3);
\draw[thick,<-] (a3.south east) -- node[below,pos=1]{\lstinline|a + 4|} +(0,-0.3);
\end{tikzpicture}
\vspace{-0.2em}
\end{equation*}

End-of-array pointers can also be used in a way where the result of a comparison
is not well-defined.
In the example below, the \lstinline|printf| is executed only
if \lstinline|x| and \lstinline|y| are allocated adjacently in the address
space (typically the stack).

\begin{lstlisting}
int x, y;
if (&x + 1 == &y) printf("x and y are allocated adjacently\n");
\end{lstlisting}

Based on the aforementioned excerpt of the \Celeven{}
standard~\cite[6.5.9p6]{iso:12}, one would naively say that the
value of \lstinline|&x + 1 == &y| is uniquely determined by the way
\lstinline|x| and \lstinline|y| are allocated in the address space.
However, the \GCC{} implementers disagree\footnote{See
\url{https://gcc.gnu.org/bugzilla/show_bug.cgi?id=61502}}.
They claim that Defect Report \#260~\cite{iso:09} allows them to take the
\emph{derivation of a pointer value} into account.

In the example, the pointers \lstinline|&x + 1| and \lstinline|&y| are
derived from unrelated objects (the local variables \lstinline|x| and
\lstinline|y|).
As a result, the \GCC{} developers claim that \lstinline|&x + 1| and
\lstinline|&y| may compare unequal albeit being allocated adjacently.
Consider:

\begin{lstlisting}
int compare(int *p, int *q) {
  // some code to confuse the optimizer
  return p == q;
}
int main() {
  int x, y;
  if (&x + 1 == &y) printf("x and y are adjacent\n");
  if (compare(&x + 1, &y)) printf("x and y are still adjacent\n");
}
\end{lstlisting}

When compiled with \GCC{} (version 4.9.2), we have observed that the string
\texttt{x and y are still adjacent} is being printed, wheras
\texttt{x and y are adjacent} is not being printed.
This means that the value of \lstinline|&x + 1 == &y| is not consistent among
different occurrences of the comparison.

Due to these discrepancies we assign undefined behavior to questionable uses of
end-of-array pointers while assigning the correct defined behavior to pointer
comparisons involving end-of-array pointers when looping through arrays (such
as in the first example above).
Our treatment is similar to our extension of \CompCert~\cite{kre:ler:wie:14}.

\subsection{Sequence point violations and non-determinism}
\label{section:challenge:sequence_point}
Instead of having to follow a specific execution order,
the execution order of expressions is unspecified in \C{}.
This is a common cause of portability problems because a
compiler may use an arbitrary execution order for each expression, and each
time that expression is executed.
Hence, to ensure correctness of a \C{} program with respect to an arbitrary
compiler, one has to verify that \emph{each} possible execution order is free of
undefined behavior and gives the correct result.

In order to make more effective optimizations possible (for example, delaying of
side-effects and interleaving), the \C{} standard does not allow an object to
be modified more than once during the execution of an expression.
If an object is modified more than once, the program has undefined behavior.
We call this requirement the \emph{sequence point restriction}.
Note that this is not a static restriction, but a restriction on valid
executions of the program.
Let us consider an example:

\begin{lstlisting}
int x, y = (x = 3) + (x = 4);
printf("%d %d\n", x, y);
\end{lstlisting}

By considering all possible execution orders, one would naively expect this
program to print \lstinline|4 7| or \lstinline|3 7|, depending on
whether the assignment \lstinline|x = 3| or \lstinline|x = 4| is executed
first.
However, \lstinline|x| is modified twice within the same
expression, and thus both execution orders have undefined behavior.
The program is thereby allowed to exhibit any behavior.
Indeed, when compiled with \lstinline|gcc -O2| (version 4.9.2), the compiled
program prints \lstinline|4 8|, which does not correspond to any of the
execution orders.

Our approach to non-determinism and sequence
points is inspired by Norrish~\cite{nor:98} and Ellison and
Ro\c{s}u~\cite{ell:ros:12}.
Each bit in memory carries a permission
(Definition~\ref{definition:permissions}) that is set to a special \emph{locked}
permission when a store has been performed.
The memory model prohibits any access (read or store) to objects with
locked permissions.
At the next sequence point, the permissions of locked objects are changed back
into their original permission, making future accesses possible again.

It is important to note that we do not have non-determinism in the memory
model itself, and have set up the memory model in such a way that all
non-determinism is on the level of the small-step operational semantics.

\section{Types in \C}
\label{section:types}

This section describes the types used in the \CHtwoO{} memory model.
We support integer, pointer, function pointer, array, struct,
union and void types.
More complicated types such as enum types and typedefs are defined by
translation~\cite{kre:wie:15,kre:15:phd}.

This section furthermore describes an abstract interface, called an
\emph{implementation environment}, that describes properties such as size and
endianness of integers, and the layout of structs and unions.
The entire \CHtwoO{} memory model and semantics 
will be parametrized by an implementation environment.

\subsection{Integer representations}
\label{section:integers}

This section describes the part of implementation environments corresponding to
integer types and the encoding of integer values as bits.
Integer types consist of a \emph{rank} (\(\charrank\), \(\shortrank\),
\(\intrank\) \ldots) and a \emph{signedness} (\(\Signed\) or \(\Unsigned\)).
The set of available ranks as well as many of their properties are
implementation defined.
We therefore abstract over the ranks in the definition of integer types.

\begin{definition}
\label{definition:int_types}
\emph{Integer signedness} and \emph{integer types} over ranks \(k \in K\) are
inductively defined as:
\begin{flalign*}
\syntax{signedness} \\
\syntax{inttype}
\end{flalign*}
The projections are called \(\rankSym : \inttype \to K\) and
\(\signSym : \inttype \to \signedness\).
\end{definition}

\begin{definition}
\label{definition:int_coding_spec}
An \emph{integer coding environment with ranks \(K\)} consists of a total
order \((K,\subseteq)\) of \emph{integer ranks} having at least the following
ranks:
\[
	\charrank \subset \shortrank \subset \intrank \subset \longrank
	\subset \longlongrank \qquad\text{and}\qquad \sizerank.
\]
It moreover has the following functions:
\begin{align*}
\charbits :{}& \natge 8 & \intendianizeSym :{}& K \to \lst\bool \to \lst\bool \\
\charsignedness :{}& \signedness &
	\intdeendianizeSym :{}& K \to \lst\bool \to \lst\bool \\
\ranksizeSym :{}& K \to \natgt 0
\end{align*}
Here, \(\intendianizeSym\;k\) and \(\intdeendianizeSym\;k\) should be inverses,
\(\intendianizeSym\;k\) should be a permutation, \(\ranksizeSym\) should
be (non-strictly) monotone, and \(\ranksize \charrank = 1\).
\end{definition}

\begin{definition}
\label{definition:int_typed}
The judgment \(\inttyped x {\itype\tau}\) describes that \emph{\(x \in \Z\) has
integer type \(\itype\tau\)}.
\begin{gather*}
\AXC{\(-2 ^ {\charbits * \ranksize k - 1} \le x
	< 2 ^ {\charbits * \ranksize k - 1}\)}
\UIC{\(\inttyped x {\IntType \Signed k}\)}
\DP\quad
\AXC{\(0 \le x < 2 ^ {\charbits * \ranksize k}\)}
\UIC{\(\inttyped x {\IntType \Unsigned k}\)}
\DP
\end{gather*}
\end{definition}

The rank \(\charrank\) is the rank of the smallest integer type,
whose unsigned variant corresponds to bytes that constitute object
representations (see Section~\ref{section:challenge:byte_level}).
Its bit size is \(\charbits\) (called \lstinline|CHAR_BIT| in the standard
library header files~\cite[5.2.4.2.1]{iso:12}), and its signedness
\(\charsignedness\) is implementation defined~\cite[6.2.5p15]{iso:12}.

The rank \(\sizerank\) is the rank of the integer types \lstinline|size_t| and
\lstinline|ptrdiff_t|, which are defined in the standard library header
files~\cite[7.19p2]{iso:12}.
The type \lstinline|ptrdiff_t| is a signed integer type used to represent the
result of subtracting two pointers, and the type \lstinline|size_t| is an
unsigned integer type used to represent sizes of types.

An integer coding environment can have an arbitrary number of integer ranks
apart from the standard ones
\(\charrank\), \(\shortrank\), \(\intrank\), \(\longrank\), \(\longlongrank\),
and \(\sizerank\).
This way, additional integer types like those describe in~\cite[7.20]{iso:12}
can easily be included.

The function \(\ranksizeSym\) gives the byte size of an integer of a given rank.
Since we require \(\ranksizeSym\) to be monotone rather than strictly monotone,
integer types with different ranks can have the same
size~\cite[6.3.1.1p1]{iso:12}.
For example, on many implementations \lstinline|int| and \lstinline|long| have
the same size, but are in fact different.

The \Celeven{} standard allows implementations to use either sign-magnitude, 1's
complement or 2's complement signed integers representations.
It moreover allows integer representations to contain padding or parity
bits~\cite[6.2.6.2]{iso:12}.
However, since all current machine architectures use 2's complement
representations, this is more of a historic artifact.
Current machine architectures use 2's complement representations because these
do not suffer from positive and negative zeros and thus enjoy unique
representations of the same integer.
Hence, \CHtwoO{} restricts itself to implementations that use
2's complement signed integers representations.

Integer representations in \CHtwoO{} can solely differ with respect to
endianness (the order of the bits).
The function \(\intendianizeSym\) takes a list of bits in little endian order
and permutes them accordingly. 
We allow \(\intendianizeSym\) to yield an arbitrary permutation and thus we not
just support big- and little-endian, but also mixed-endian variants.

\begin{definition}
\label{definition:int_encode}
Given an integer type \(\itype\tau\), the \emph{integer encoding functions}
\(\inttobitsSym {\itype\tau} : \Z \to \lst\bool\) and
\(\intofbitsSym {\itype\tau} : \lst\bool \to \Z\) are defined as follows:
\begin{flalign*}
\inttobits {\IntType {si} k} x \defined{}&
	\intendianize k {(\textnormal{\(x\) as little endian 2's complement})}\\
\intofbits {\IntType {si} k} {\vec \beta} \defined{}&
	\textnormal{of little endian 2's complement
	\((\intdeendianize k {\vec\beta})\)}
\end{flalign*}
\end{definition}

\begin{lemma}
The integer encoding functions are inverses.
That means:
\begin{enumerate}
\item We have \(\intofbits {\itype\tau} {\inttobits {\itype\tau} x} = x\) and
	\(\length {\inttobits {\itype\tau} x} = \ranksize {\itype\tau}\) provided
	that \(\inttyped x {\itype\tau}\).
\item We have 
	\(\inttobits {\itype\tau} {\intofbits {\itype\tau} {\vec\beta}} = \vec\beta\)
	and \(\inttyped {\intofbits {\itype\tau} {\vec\beta}} {\itype\tau}\)
	provided that \(\length{\vec\beta} = \ranksize {\itype\tau}\).
\end{enumerate}
\end{lemma}

\subsection{Definition of types}
\label{section:types_definition}

We support integer, pointer, function pointer, array,
struct, union and void types.
The translation that we have described in~\cite{kre:wie:15,kre:15:phd}
translates more complicated types, such as typedefs and enums, into these
simplified types.
This translation also alleviates other simplifications
of our simplified definition of types, such as the use of unnamed struct and union
fields.
Floating point types and type qualifiers like \lstinline|const| and
\lstinline|volatile| are not supported.

All definitions in this section are implicitly parametrized by an integer
coding environment with ranks \(K\) (Definition~\ref{definition:int_coding_spec}).

\begin{definition}
\emph{Tags \(t \in \stag\)} (sometimes called \emph{struct/union names}) and
\emph{function names \(f \in \funname\)} are represented as strings.
\end{definition}

\begin{definition}
\label{definition:types}
\emph{Types} consist of \emph{point-to types}, \emph{base types} and
\emph{full types}.
These are inductively defined as:
\begin{flalign*}
\syntax{ptrtype} \\
\syntax{basetype} \\
\syntax{type}
\end{flalign*}
\end{definition}

The three mutual inductive parts of types correspond to the different components
of the memory model.
Addresses and pointers have point-to types
(Definitions~\ref{definition:addr_typed} and~\ref{definition:ptr_typed}),
base values have base types (Definition~\ref{definition:base_val_typed}), and
memory trees and values have full types (Definitions~\ref{definition:ctree_typed}
and~\ref{definition:val_typed}).

The void type of \C{} is
used for two entirely unrelated purposes: \lstinline|void| is used for functions
without return type and \lstinline|void*| is used for pointers to objects of
unspecified type.
In \CHtwoO{} this distinction is explicit in the syntax of types.
The type \(\TVoid\) is used for function without return value.
Like the mathematical \emph{unit} type it has one value called \(\VVoid\)
(Definition~\ref{definition:base_values}).
The type \(\TPtr \TAny\) is used for pointers to objects of unspecified type.

Unlike more modern programming languages \C{} does not provide first class functions.
Instead, \C{} provides function pointers which are just addresses of
executable code in memory instead of closures.
Function pointers can be used in a way similar to ordinary pointers: they can
be used as arguments and return value of functions, they can be part of
structs, unions and arrays, \etc{}

The \C{} language sometimes allows function types to be used as shorthands for
function pointers, for example:

\begin{lstlisting}
void sort(int *p, int len, int compare(int,int));
\end{lstlisting}

The third argument of \lstinline|sort| is a shorthand for
\lstinline|int (*compare)(int,int)|
and is thus in fact a function pointer instead of a function.
We only have function pointer types, and the third argument of the type of the
function \lstinline|sort| thus contains an additional \(\TPtrSym\):
\[
\TFun {\listlit{
	\TBase {\TPtr {(\TType {\TBase {\TSignedInt}})}},\
	\TBase \TSignedInt,\
	\TPtr {(\TFun {\TBase \TSignedInt} {\TBase \TSignedInt})}
}} {\TBase \TVoid}.
\]

Struct and union types consist of just a name, and do
not contain the types of their fields.
An environment is used to assign fields to structs and unions, and to assign
argument and return types to function names.

\begin{definition}
\label{definition:env}
\emph{Type environments} are defined as:
\begin{flalign*}
\syntax{env}.
\end{flalign*}
The functions \(\compounddomSym : \env \to \finset \stag\) and
\(\fundomSym : \env \to \finset \funname\) yield the declared structs and
unions, respectively the declared functions.
We implicitly treat environments as functions
\(\map \stag {\lst\type}\) and
\(\map \funname {(\lst\type \times \type)}\) that
correspond to underlying finite partial functions.
\end{definition}

Struct and union names on the one hand, and function names on the other, have
their own name space in accordance with the \Celeven{}
standard~\cite[6.2.3p1]{iso:12}.

\begin{notation}
We often write an environment as a mixed sequence of struct and
union declarations \(t : \vec \tau\), and function declarations
\(f : \pair {\vec\tau} \tau\).
This is possible because environments are finite.
\end{notation}

Since we represent the fields of structs and unions as lists,
fields are nameless.
For example, the \C{} type \lstinline|struct S1 { int x; struct S1 *p; }| is
translated into the environment \(\mathtt{S1} : \listlit{
\TBase {\TInt {\IntType \Signed \intrank}},
\TBase {\TPtr {\TType {\TStruct {\mathtt{S1}}}}}}\).

Although structs and unions are semantically very different (products versus
sums, respectively), environments do not keep track of whether a tag has been
used for a struct or a union type.
Structs and union types with the same tag are thus allowed.
The translator in~\cite{kre:wie:15,kre:15:phd} forbids the same name
being used to declare both a struct and union type.

Although our mutual inductive syntax of types already forbids many
incorrect types such as functions returning functions (instead of function
pointers), still some ill-formed types such as \lstinline|int[0]| are
syntactically valid.
Also, we have to ensure that cyclic structs and unions are only allowed when
the recursive definition is guarded through pointers.
Guardedness by pointers ensures that the sizes of types are finite and
statically known.
Consider the following types:

\begin{lstlisting}
struct list1 { int hd; struct list1 tl; };     /* illegal */
struct list2 { int hd; struct list2 *p_tl; };  /* legal */
\end{lstlisting}

The type declaration \lstinline|struct list1| is illegal because it
has a reference to itself.
In the type declaration \lstinline|struct list2| the self reference is guarded
through a pointer type, and therefore legal.
Of course, this generalizes to mutual recursive types like:

\begin{lstlisting}
struct tree { int hd; struct forest *p_children; };
struct forest { struct tree *p_hd; struct forest *p_tl; };
\end{lstlisting}

\begin{definition}
\label{definition:type_valid}
The following judgments are defined by mutual induction:
\begin{itemize}
\item The judgment \(\ptrtypevalid \Gamma {\ptype\tau}\) describes
	\emph{point-to types \(\ptype\tau\) to which a pointer may point}:
\begin{gather*}
\AXC{\strut}
\UIC{\(\ptrtypevalid \Gamma \TAny\)}
\DP\qquad
\AXC{\strut\(\ptrtypevalid \Gamma {\vec\tau}\)}
\AXC{\strut\(\ptrtypevalid \Gamma \tau\)}
\BIC{\(\ptrtypevalid \Gamma {\TFun {\vec\tau} \tau}\)}
\DP \qquad
\AXC{\strut\(\basetypevalid \Gamma {\btype\tau}\)}
\UIC{\(\ptrtypevalid \Gamma {\TType {\TBase{\btype\tau}}}\)}
\DP\qquad
\AXC{\strut\(\typevalid \Gamma \tau\)}
\AXC{\strut\(n \neq 0\)}
\BIC{\(\ptrtypevalid \Gamma {\TType {\TArray \tau n}}\)}
\DP
\\[0.5em]
\AXC{}
\UIC{\(\ptrtypevalid \Gamma {\TType {\TStruct t}}\)}
\DP\qquad
\AXC{}
\UIC{\(\ptrtypevalid \Gamma {\TType {\TUnion t}}\)}
\DP   
\end{gather*}
\item The judgment \(\basetypevalid \Gamma {\btype\tau}\) describes \emph{valid
	base types \(\btype\tau\)}:
\begin{gather*}
\AXC{\strut}
\UIC{\(\basetypevalid \Gamma {\TInt{\itype\tau}}\)}
\DP\qquad
\AXC{\strut\(\ptrtypevalid \Gamma {\ptype\tau}\)}
\UIC{\(\basetypevalid \Gamma {\TPtr {\ptype\tau}}\)}
\DP\qquad
\AXC{\strut}
\UIC{\(\basetypevalid \Gamma \TVoid\)}
\DP
\end{gather*}
\item The judgment \(\typevalid \Gamma \tau\) describes \emph{valid types
	\(\tau\)}:
\begin{gather*}
\AXC{\strut\(\basetypevalid \Gamma {\btype\tau}\)}
\UIC{\(\typevalid \Gamma {\TBase{\btype\tau}}\)}
\DP\qquad
\AXC{\strut\(\typevalid \Gamma \tau\)}
\AXC{\strut\(n \neq 0\)}
\BIC{\(\typevalid \Gamma {\TArray \tau n}\)}
\DP\qquad
\AXC{\strut\(t \in \compounddom \Gamma\)}
\UIC{\(\typevalid \Gamma {\TStruct t}\)}
\DP\qquad
\AXC{\strut\(t \in \compounddom \Gamma\)}
\UIC{\(\typevalid \Gamma {\TUnion t}\)}
\DP
\end{gather*}
\end{itemize}
\end{definition}

\begin{definition}
\label{definition:env_valid}
The judgment \(\envvalid \Gamma\) describes \emph{well-formed environments
\(\Gamma\)}.
It is inductively defined as:
\begin{gather*}
\renewcommand{\defaultHypSeparation}{\hskip.7em}
\renewcommand{\ScoreOverhang}{0.1em}
\AXC{\strut}
\UIC{\(\envvalid \emptyset\)}
\DP\quad
\AXC{\strut\(\envvalid \Gamma\)}
\AXC{\(\typevalid \Gamma {\vec\tau}\)}
\AXC{\(\vec\tau \neq \nil\)}
\AXC{\(t \notin \compounddom \Gamma\)}
\QIC{\(\envvalid {\einsert t {\vec\tau} \Gamma}\)}
\DP\quad
\AXC{\strut\(\envvalid \Gamma\)}
\AXC{\(\ptrtypevalid \Gamma {\vec\tau}\)}
\AXC{\(\ptrtypevalid \Gamma \tau\)}
\AXC{\(f \notin \fundom \Gamma\)}
\QIC{\(\envvalid {\einsert f {\pair {\vec\tau} \tau} \Gamma}\)}
\DP
\end{gather*}
\end{definition}

Note that \(\typevalid \Gamma \tau\) does not imply \(\envvalid \Gamma\).
Most results therefore have \(\envvalid \Gamma\) as a premise.
These premises are left implicit in this paper.

In order to support (mutually) recursive struct and union types, pointers to
incomplete struct and union types are permitted in the judgment
\(\ptrtypevalid \Gamma {\ptype\tau}\) that describes types to which pointers
are allowed, but forbidden in the judgment \(\typevalid \Gamma \tau\) of
validity of types.
Let us consider the following type declarations:

\begin{lstlisting}
struct S2 { struct S2 x; };   /* illegal */
struct S3 { struct S3 *p; };  /* legal */
\end{lstlisting}

Well-formedness \(\envvalid \Gamma\) of the environment
\(\Gamma \defined \mathtt{S3} :
	\listlit{\TBase{\TPtr{\TType {\TStruct {\mathtt{S3}}}}}}\) can
be derived using the judgments
\(\ptrtypevalid \emptyset {\TStruct {\mathtt{S3}}}\),
\(\basetypevalid \emptyset {\TPtr {\TType {\TStruct {\mathtt{S3}}}}}\),
\(\typevalid \emptyset {\TBase{\TPtr {\TType {\TStruct {\mathtt{S3}}}}}}\), and
thus \(\envvalid \Gamma\).
The environment \(\mathtt{S2} : \listlit{\TStruct {\mathtt{S2}}}\) is ill-formed
because we do not have
\(\typevalid \emptyset {\TStruct {\mathtt{S2}}}\).

The typing rule for function pointers types is slightly more delicate.
This is best illustrated by an example:

\begin{lstlisting}
union U { int i; union U (*f) (union U); };
\end{lstlisting}

This example displays a recursive self reference to a union type through
a function type, which is legal in \C{} because function types are in fact
pointer types.
Due to this reason, the premises of \(\ptrtypevalid \Gamma {\TFun {\vec\tau} \tau}\)
are \(\ptrtypevalid \Gamma {\vec\tau}\) and \(\ptrtypevalid \Gamma \tau\)
instead of \(\typevalid \Gamma {\vec\tau}\) and \(\typevalid \Gamma \tau\).
Well-formedness of the above union type can be derived as follows:

\begin{equation*}
\AXC{\(\envvalid \Gamma\)}
	\AXC{}\UIC{\(\basetypevalid \emptyset \TSignedInt\)}
\UIC{\(\typevalid \emptyset {\TBase \TSignedInt}\)}
	\AXC{}\UIC{\(\ptrtypevalid \emptyset {\TUnion {\mathtt U}}\)}
	\AXC{}\UIC{\(\ptrtypevalid \emptyset {\TUnion {\mathtt U}}\)}
	\BIC{\(\ptrtypevalid \emptyset
		{\TFun {\TUnion {\mathtt U}} {\TUnion {\mathtt U}}}\)}
	\UIC{\(\basetypevalid \emptyset
		{\TPtr {(\TFun {\TUnion {\mathtt U}} {\TUnion {\mathtt U}})}}\)}
\UIC{\(\typevalid \emptyset
	{\TBase {\TPtr {(\TFun {\TUnion {\mathtt U}} {\TUnion {\mathtt U}})}}}\)}
\TIC{\(\envvalid {\einsert {\mathtt U} {\listlit {\TBase \TSignedInt,
	\TBase {\TPtr {(\TFun {\TUnion {\mathtt U}} {\TUnion {\mathtt U}})}}}}
	\Gamma}\)}
\DP
\end{equation*}

In order to define operations by recursion over the structure of well-formed
types (see for example Definition~\ref{definition:val_unflatten}, which turns a
sequence of bits into a value), we often
need to perform recursive calls on the types of fields of structs and unions.
In \Coq{} we have defined a custom recursor and induction principle using
well-founded recursion.
In this paper, we will use these implicitly.

Affeldt \etal~\cite{aff:mar:13,aff:sak:14} have formalized non-cyclicity of
types using a complex constraint on paths through types.
Our definition of validity of environments
(Definition~\ref{definition:env_valid}) follows the structure of type
environments, and is more easy to use (for example to implement the
aforementioned recursor and induction principle).

There is a close
correspondence between array and pointer types in \C{}.
Arrays are not first class types, and except for special cases such as
initialization, manipulation of arrays is achieved via pointers.
We consider arrays as first class types so as to
avoid having to make exceptions for the case of arrays all the time.

Due to this reason, more types are valid in \CHtwoO{} than
in \Celeven{}.
The translator in~\cite{kre:wie:15,kre:15:phd} resolves exceptional cases for arrays.
For example, a function parameter of array type acts like a parameter of pointer
type in \Celeven~\cite[6.7.6.3]{iso:12}\footnote{The array size is ignored
unless the \lstinline|static| keyword is used.
In case \lstinline|f| would have the prototype
\lstinline|void f(int a[static 10])|, the pointer
\lstinline|a| should provide access to an array of at least 10
elements~\cite[6.7.6.3]{iso:12}.
The \lstinline|static| keyword is not supported by \CHtwoO.}.

\begin{lstlisting}
void f(int a[10]);
\end{lstlisting}

The corresponding type of the function \lstinline|f| is thus
\(\TFun {\TBase {\TPtr {(\TBase \TSignedInt)}}} {\TBase \TVoid}\).
Note that the type
\(\TFun {\TArray {(\TBase \TSignedInt)} {10}} {\TBase \TVoid}\) is also
valid, but entirely different, and never generated by the translator
in~\cite{kre:wie:15,kre:15:phd}.

\subsection{Implementation environments}
\label{section:env_spec}

We finish this section by extending integer coding environments to describe
implementation defined properties related the layout of struct
and union types.
The author's PhD thesis~\cite{kre:15:phd} also considers the
implementation-defined behavior of integer operations (such as addition and
division) and defines inhabitants
of this interface corresponding to actual computing architectures.

\begin{definition}
\label{definition:env_spec}
A \emph{implementation environment with ranks \(K\)} consists of an integer
coding environment with ranks \(K\) and functions:
\begin{equation*}
\sizeofSym \Gamma :{} \type \to \nat \qquad
\alignofSym \Gamma :{} \type \to \nat \qquad
\fieldsizesSym \Gamma :{} \lst\type \to \lst\nat
\end{equation*}
These functions should satisfy:
\begin{align*}
\sizeof \Gamma {(\TBase {\TInt {\IntType {si} k}})} ={}& \ranksize k 
&&	\sizeof \Gamma {(\TBase {\TPtr {\ptype\tau}})} \ne{} 0 \qquad
	\sizeof \Gamma {\TBase \TVoid} \ne{} 0 \\
\sizeof \Gamma {(\TArray \tau n)} ={}& n * \sizeof \Gamma \tau \\
\sizeof \Gamma {(\TStruct t)} ={}& \Sigma\; \fieldsizes \Gamma {\vec \tau}
	&& \textnormal{if \(\elookup t \Gamma = \Some {\vec \tau}\)} \\
\sizeof \Gamma {\tau_i} \le{}& z_i
	\textnormal{ and } \length {\vec \tau} = \length {\vec z}
	&& \textnormal{if \(\fieldsizes \Gamma {\vec\tau} = \vec z\),
	for each \(i < \length {\vec\tau}\)} \\
\sizeof \Gamma {\tau_i} \le{}& \sizeof \Gamma {(\TUnion t)}
	&& \textnormal{if \(\elookup t \Gamma = \Some {\vec\tau}\),
	for each \(i < \length {\vec\tau}\)} \\
\alignof \Gamma \tau \divides{}& \alignof \Gamma {(\TArray \tau n)} \\
\alignof \Gamma {\tau_i} \divides{}& \alignof \Gamma {(\TStruct t)}
	&& \textnormal{if \(\elookup t \Gamma = \Some {\vec\tau}\),
	for each \(i < \length {\vec\tau}\)} \\
\alignof \Gamma {\tau_i} \divides{}& \alignof \Gamma {(\TUnion t)}
	&& \textnormal{if \(\elookup t \Gamma = \Some {\vec\tau}\),
	for each \(i < \length {\vec\tau}\)} \\
\alignof \Gamma \tau \divides{}& \sizeof \Gamma \tau
	&& \textnormal{if \(\typevalid \Gamma \tau\)} \\
\alignof \Gamma {\tau_i} \divides{}& \offsetof \Gamma {\vec\tau} i
	&& \textnormal{if \(\typevalid \Gamma {\vec\tau}\),
	for each \(i < \length {\vec \tau}\)}
\end{align*}
Here, we let \(\offsetof \Gamma {\vec \tau} i\) denote
\(\Sigma_{j < i} (\fieldsizes \Gamma {\vec\tau})_j\).
The functions \(\sizeofSym \Gamma\), \(\alignofSym \Gamma\), and
\(\fieldsizesSym \Gamma\) should be closed under weakening of \(\Gamma\).
\end{definition}

\begin{notation}
Given an implementation environment, we let:
\begin{flalign*}
\bitsizeof \Gamma \tau \defined{}&
	\sizeof \Gamma \tau \cdot \charbits \\
\bitoffsetof \Gamma \tau j \defined{}&
	\offsetof \Gamma \tau j \cdot \charbits \\
\fieldbitsizes \Gamma \tau \defined{}&
	\fieldsizes \Gamma \tau \cdot \charbits
\end{flalign*}
\end{notation}

We let \(\sizeof \Gamma \tau\) specify the number of bytes out of
which the object representation of an object of type \(\tau\) constitutes.
Objects of type \(\tau\) should be allocated at addresses that are a multiple
of \(\alignof \Gamma \tau\).
We will prove that our abstract notion of addresses satisfies this property
(see Lemma~\ref{lemma:align_of_addr_object_offset}).
The functions \(\sizeofSym \Gamma\), \(\alignofSym \Gamma\) correspond to the
\lstinline|sizeof| and \lstinline|_Alignof| operators~\cite[6.5.3.4]{iso:12},
and \(\offsetofSym \Gamma\) corresponds to the \lstinline|offsetof|
macro~\cite[7.19p3]{iso:12}.
The list \(\fieldsizes \Gamma {\vec\tau}\) specifies the layout of a struct
type with fields \(\vec\tau\) as follows:

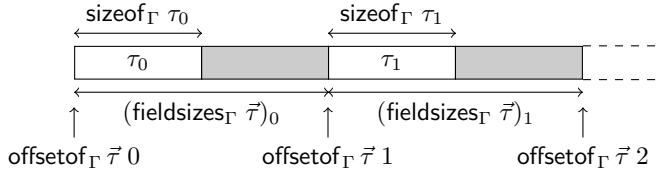
\begin{equation*}
\begin{tikzpicture}[
	start chain=going right,
	node distance=-\pgflinewidth,
	every node/.style={shape=rectangle,minimum height=1.3em,minimum width=5em}
]
\node[draw,on chain] (x0) {\(\tau_0\)};
\node[draw,on chain,fill=gray,fill opacity=0.4,right=of x0] (pad0) {};
\node[draw,on chain,right=of pad0] (x1) {\(\tau_1\)};
\node[draw,on chain,fill=gray,fill opacity=0.4,right=of x1] (pad1) {};
\draw[dashed] (pad1.north east) -- +(1,0);
\draw[dashed] (pad1.south east) -- +(1,0);
\draw[transform canvas={yshift=0.5em},<->] (x0.north west) --
	node[above]{\(\sizeof \Gamma {\tau_0}\)} (x0.north east);
\draw[transform canvas={yshift=0.5em},<->] (x1.north west) --
	node[above]{\(\sizeof \Gamma {\tau_1}\)} (x1.north east);
\draw[transform canvas={yshift=-0.5em},<->] (x0.south west) --
	node[below]{\((\fieldsizes \Gamma {\vec\tau})_0\)} (pad0.south east);
\draw[transform canvas={yshift=-0.5em},<->] (x1.south west) --
	node[below]{\((\fieldsizes \Gamma {\vec\tau})_1\)} (pad1.south east);
\draw[transform canvas={yshift=-1.1em},<-] (x0.south west) -- +(0,-0.4)
	node[below] {\(\offsetof \Gamma {\vec \tau} 0\)};
\draw[transform canvas={yshift=-1.1em},<-] (x1.south west) -- +(0,-0.4)
	node[below] {\(\offsetof \Gamma {\vec \tau} 1\)};
\draw[transform canvas={yshift=-1.1em},<-] (pad1.south east) -- +(0,-0.4)
	node[below] {\(\offsetof \Gamma {\vec \tau} 2\)};
\end{tikzpicture}
\end{equation*}

\section{Permissions and separation algebras}
\label{section:permissions}

Permissions control whether memory operations such as a read or store are
allowed or not.
In order to obtain the highest level of precision, we tag each individual
bit in memory with a corresponding permission.
In the operational semantics, permissions have two main purposes:

\begin{itemize}
\item Permissions are used to formalize the \emph{sequence point restriction}
	which assigns undefined behavior to programs in which an object in memory
	is modified more than once in between two sequence points.
\item Permissions are used to distinguish objects in memory that are writable
	from those that are read-only (\emph{const qualified} in \C{} terminology).
\end{itemize}

In the axiomatic semantics based on separation logic, permissions play an
important role for \emph{share accounting}.
We use share accounting for \emph{subdivision of permissions} among multiple
subexpressions to ensure that:

\begin{itemize}
\item Writable objects are unique to each subexpression.
\item Read-only objects may be shared between subexpressions.
\end{itemize}

This distinction is originally due to Dijkstra~\cite{dij:68:seq} and
is essential in separation logic with permissions~\cite{bor:cal:hea:par:05}.
The novelty of our work is to use separation logic with permissions for
non-determinism in expressions in \C.
Share accounting gives rise to a natural treatment of \C's sequence point
restriction.

\emph{Separation algebras} as introduced by Calcagno
\etal{}~\cite{cal:hea:yan:07} abstractly capture common structure of
subdivision of permissions.
We present a generalization of separation algebras that is well-suited for \C{}
verification in \Coq{} and use this generalization to build the permission system
and memory model compositionally.
The permission system will be constructed as a telescope of
separation algebras:
\[
	\perm \quad \defined \quad
	\underbrace {\lockable
	  {\counter \Q}}_{\mathclap{\textnormal{non-const qualified}}}
	\qquad + \qquad
	\underbrace{\Q}_{\mathclap{\textnormal{const qualified}}}
\]
Here, \(\Q\) is the separation algebra of fractional permissions,
\(\counterSym\) is a functor that extends a separation algebra with a counting
component, and \(\lockableSym\) is a functor that extends a separation algebra
with a lockable component (used for the sequence point restriction).
This section explains these functors and their purposes.

\subsection{Separation logic and share accounting}
\label{section:separation_and_share}

Before we will go into the details of the \CHtwoO{} permission system, we
briefly introduce separation logic.
Separation logic~\cite{hea:rey:yan:01} is an extension of Hoare logic that
provides better means to
reason about imperative programs that use mutable data structures and pointers.
The key feature of separation logic is the \emph{separating conjunction}
\(P \axSep Q\) that allows one
to subdivide the memory into two disjoint parts: a part described by \(P\) and
another part described by \(Q\).
The separating conjunction is most prominent in the \emph{frame rule}.
\begin{equation*}
\AXC{\(\axShort P s Q\)}
\UIC{\(\axShort {P \axSep R} s {Q \axSep R}\)}
\DP
\end{equation*}

This rule enables local reasoning.
Given a Hoare triple \(\axShort P s Q\), this rule allows one to derive
that the triple also holds when the memory is extended with a disjoint part
described by \(R\).
The frame rule shows its merits when reasoning about functions.
There it allows one to consider a function in the context of the memory the
function actually uses, instead of having to consider the function in the
context of the entire program's memory.
However, already in derivations of small programs the use of the
frame rule can be demonstrated\footnote{Contrary to traditional separation
logic, we do not give local variables a special status of being \emph{stack
allocated}.
We do so, because in \C{} even local variables are allowed to have pointers
to them.}:

\vspace{-1em}
\begin{small}
\begin{equation*}
\renewcommand{\ScoreOverhang}{0.1em}
\AXC{\phantom{a}}\UIC{\(\axShort
	{\axSingletonShort {\mathtt x} {} 0}
	{\mbox{\lstinline|x:=10|}}
	{\axSingletonShort {\mathtt x} {} 10}\)}
\UIC{\(\axShort
	{\axSingletonShort {\mathtt x} {} 0 \axSep \axSingletonShort {\mathtt y} {} 0}
	{\mbox{\lstinline|x:=10|}}
	{\axSingletonShort {\mathtt x} {} 10 \axSep \axSingletonShort {\mathtt y} {} 0}\)}
\AXC{}\UIC{\(\axShort
	{\axSingletonShort {\mathtt y} {} 0}
	{\mbox{\lstinline|y:=12|}}
	{\axSingletonShort {\mathtt y} {} 12}\)}
\UIC{\(\axShort
	{\axSingletonShort {\mathtt x} {} 10 \axSep \axSingletonShort {\mathtt y} {} 0}
	{\mbox{\lstinline|y:=12|}}
	{\axSingletonShort {\mathtt x} {} 10 \axSep \axSingletonShort {\mathtt y} {} 12}\)}
\BIC{\(\axShort
	{\axSingletonShort {\mathtt x} {} 0 \axSep \axSingletonShort {\mathtt y} {} 0}
	{\mbox{\lstinline|x:=10; y:=12|}}
	{\axSingletonShort {\mathtt x} {} 10 \axSep \axSingletonShort {\mathtt y} {} 12}\)}
\DP
\end{equation*}
\end{small}

The \emph{singleton assertion} \(\axSingletonShort a {} v\) denotes that the
memory consists of exactly one object with value \(v\) at address \(a\).
The assignments are not considered in the context of the entire memory,
but just in the part of the memory that is used.

The key observation that led to our separation logic for
\C, see also~\cite{kre:14:popl,kre:15:phd}, is the correspondence between
non-determinism in expressions and a form of concurrency.
Inspired by the rule for the parallel composition~\cite{hea:04}, we
have rules for each operator \(\binOp\) that are of the following shape.
\begin{equation*}
\AXC{\(\axShort {P_1} {e_1} {Q_1}\)}
\AXC{\(\axShort {P_2} {e_2} {Q_2}\)}
\BIC{\(\axShort {P_1 \axSep P_2} {\EBinOp \binOp {e_1} {e_2}}
	{Q_1 \axSep Q_2}\)}
\DP
\end{equation*}

The intuitive idea of this rule is that if the memory can be subdivided into
two parts in which the subexpressions \(e_1\) and \(e_2\) can
be executed safely, then the expression \(\EBinOp \binOp {e_1} {e_2}\) can
be executed safely in the whole memory.
Non-interference of the side-effects of \(e_1\) and \(e_2\) is guaranteed
by the separating conjunction.
It ensures that the parts of the memory
described by \(P_1\) and \(P_2\) do not have overlapping areas that will be
written to.
We thus effectively rule out expressions with undefined behavior such as
\lstinline|(x = 3) + (x = 4)|
(see Section~\ref{section:challenge:sequence_point} for discussion).

Subdividing the memory into multiple parts is not a simple operation.
In order to illustrate this, let us consider a shallow embedding of assertions
of separation logic \(P,Q : \mem \to \Prop\) (think of \(\mem\) as being the
set of finite partial function from some set of object identifiers to some set
of objects.
The exact definition in the context of \CHtwoO{} is given in
Definition~\ref{definition:memory}).
In such a shallow embedding, one would define the separating conjunction
as follows:
\begin{flalign*}
P \axSep Q \defined{}& \lambda m \wsdot
	\exists m_1\,m_2 \wsdot m = \sepunion {m_1} {m_2} \land P\,m_1 \land Q\,m_2.
\end{flalign*}

The operation \(\sepunionSym\) is \emph{not} the
disjoint union of finite partial functions, but a more fine grained operation.
There are two reasons for that.
Firstly, subdivision of memories should allow for partial overlap, as
long as writable objects are unique to a single part.
For example, the expression \lstinline|x + x| has defined behavior, but the
expressions \lstinline|x + (x = 4)| and \lstinline|(x = 3) + (x = 4)| have not.

We use separation logic with permissions~\cite{bor:cal:hea:par:05} to deal
with partial overlap of memories.
That means, we equip the singleton assertion 
\(\axSingletonShort a \gamma v\) with a permission \(\gamma\).
The essential property of the singleton assertion is that given a writable
permission \(\gamma_w\) there is a readable permission \(\gamma_r\) with:
\[
	(\axSingletonShort a {\gamma_w} v) \quad\iff\quad
	(\axSingletonShort a {\gamma_r} v) \axSep (\axSingletonShort a {\gamma_r} v).
\]

The above property is an instance of a slightly more general property.
We consider a binary operation \(\sepunionSym\) on permissions so we can
write:
\[
	(\axSingletonShort a {\sepunion {\gamma_1} {\gamma_2}} v) \quad\iff\quad
	(\axSingletonShort a {\gamma_1} v) \axSep (\axSingletonShort a {\gamma_2} v).
\]

Secondly, it should be possible to subdivide array, struct and union objects
into subobjects corresponding to their elements.
For example, in the case of an array \lstinline|int a[2]|, the
expression \lstinline|(a[0] = 1) + (a[1] = 4)| has defined
behavior, and we should be able to prove so.
The essential property of the singleton assertion for an 
\(\VArray {} {\listlit{y_0, \dotsc, y_{n-1}}}\) value is:
\[
	(\axSingletonShort a \gamma {\VArray {} {\listlit{v_0, \dotsc, v_{n-1}}}})
	\quad\iff\quad
	(\axSingletonShort {a[0]} \gamma {v_0}) \axSep \dotsb \axSep
	(\axSingletonShort {a[n-1]} \gamma {v_{n-1}}).
\]

This paper does not describe the \CHtwoO{} separation logic
and its shallow embedding of assertions.
These are described in the author's PhD thesis~\cite{kre:15:phd}.
Instead, we consider just the operations \(\sepunionSym\) on permissions and
memories.

\subsection{Separation algebras}
\label{section:separation_algebra}

As shown in the previous section, the key operation needed to define a
shallow embedding of separation logic with permissions is a binary operation
\(\sepunionSym\) on memories and permissions.
Calcagno \etal{} introduced the notion of a
\emph{separation algebra}~\cite{cal:hea:yan:07} so as to capture common
properties of the \(\sepunionSym\) operation.
A \emph{separation algebra} \mbox{\((A,\sepempty,\sepunionSym)\)}
is a partial cancellative commutative monoid (see
Definition~\ref{definition:separation_algebra} for our actual
definition).
Some prototypical instances of separation algebras are:

\begin{itemize}
\item Finite partial functions \((\map A B,\sepempty,\sepunionSym)\), where
	\(\sepempty\) is the empty finite partial function, and \(\sepunionSym\) the
	disjoint union on finite partial functions.
\item The Booleans \((\bool,\false,\lor)\).
\item Boyland's fractional permissions \(([0,1]_\Q,0,+)\) where \(0\) denotes
	no access, \(1\) denotes writable access, and \(0 < \_ < 1\) denotes read-only
	access~\cite{bor:cal:hea:par:05,boy:03}.
\end{itemize}

Separation algebras are also closed under various constructs (such as products and
finite functions), and complex instances can thus be built compositionally.

When formalizing separation algebras in the \Coq{} proof assistant, we quickly
ran into some problems:

\begin{itemize}
\item Dealing with partial operations such as \(\sepunionSym\) is cumbersome,
	see Section~\ref{section:coq_partial}.
\item Dealing with subset types (modeled as \(\Sigma\)-types) is inconvenient.
\item Operations such as the difference operation \(\sepdifferenceSym\) cannot
	be	defined constructively from the laws of a separation algebra.
\end{itemize}

In order to deal with the issue of partiality, we turn \(\sepunionSym\) into a
total operation.
Only in case \(x\) and \(y\) are \emph{disjoint}, notation \(\sepdisjoint x y\),
we require \(\sepunion x y\) to satisfy the laws of a separation algebra.
Instead of using subsets, we equip separation algebras with a predicate
\(\sepvalidSym : A \to \Prop\) that explicitly describes a subset of the carrier
\(A\).
Lastly, we explicitly add a difference operation \(\sepdifferenceSym\).

\begin{definition}
\label{definition:separation_algebra}
A \emph{separation algebra} consists of a type \(A\), with:
\begin{itemize}
\item An element \(\sepempty : A\)
\item A predicate \(\sepvalidSym : A \to \Prop\)
\item Binary relations
	\(\sepdisjointSym,\ \sepsubseteqSym{} : A \to A \to \Prop\)
\item Binary operations \(\sepunionSym,\ \sepdifferenceSym{} : A \to A \to A\)
\end{itemize}

\noindent
Satisfying the following laws:
\begin{enumerate}[series=sep_alg]
\item If \(\sepvalid x\), then \(\sepdisjoint \sepempty x\)
	and \(\sepunion \sepempty x = x\)
	\label{item:sep_left_id}
\item If \(\sepdisjoint x y\), then \(\sepdisjoint y x\)
	and \(\sepunion x y = \sepunion y x\)
	\label{item:sep_commutative}
\item If \(\sepdisjoint x y\) and \(\sepdisjoint {\sepunion x y} z\),
	then
	\(\sepdisjoint y z\), \(\sepdisjoint x {\sepunion y z}\), and
	\(\sepunion x {(\sepunion y z)} = \sepunion {(\sepunion x y)} z\)
	\label{item:sep_associative}
\item If \(\sepdisjoint z x\), \(\sepdisjoint z y\) and
	\(\sepunion z x = \sepunion z y\), then \(x = y\)
	\label{item:sep_cancel}
\item If \(\sepdisjoint x y\), then \(\sepvalid x\) and
	\(\sepvalid {(\sepunion x y)}\)
	\label{item:sep_union_valid}
\item If \(\sepdisjoint x y\) and \(\sepunion x y = \sepempty\),
	then \(x = \sepempty\)
	\label{item:sep_positive}
\item If \(\sepdisjoint x y\), then \(\sepsubseteq x {\sepunion x y}\)
	\label{item:sep_union_subseteq}
\item If \(\sepsubseteq x y\), then \(\sepdisjoint x {\sepdifference y x}\)
	and \(\sepunion x {\sepdifference y x} = y\)
	\label{item:sep_union_difference}
\end{enumerate}
\end{definition}

Laws~\ref{item:sep_left_id}--\ref{item:sep_cancel} describe the
traditional laws of a separation algebra: identity, commutativity,
associativity and cancellativity.
Law~\ref{item:sep_union_valid} ensures that \(\sepvalidSym\) is closed under
the \(\sepunionSym\) operation.
Law~\ref{item:sep_positive} describes positivity.
Laws~\ref{item:sep_union_subseteq} and~\ref{item:sep_union_difference} fully
axiomatize the \(\sepsubseteqSym\) relation and \(\sepdifferenceSym\) operation.
Using the positivity and cancellation law, we obtain that \(\sepsubseteqSym\)
is a partial order in which \(\sepunionSym\) is order preserving and respecting.

In case of permissions, the \(\sepempty\) element is used to split objects of
compound types (arrays and structs) into multiple parts.
We thus use separation algebras instead of \emph{permission
algebras}~\cite{hea:rey:yan:01}, which are a variant of separation algebras
without an \(\sepempty\) element.

\begin{definition}
\label{definition:boolean_algebra}
The \emph{Boolean separation algebra \(\bool\)} is defined as: \hfill
\begin{align*}
\sepvalid x \defined{}& \True & \sepempty \defined{}& \false \\
\sepdisjoint x y \defined{}& \neg x \lor \neg y
	& \sepunion x y \defined{}& x \lor y \\
\sepsubseteq x y \defined{}& x \impl y
	& \sepdifference x y \defined{}& x \land \neg y
\end{align*}
\end{definition}

In the case of fractional permissions \([0,1]_\Q\) the problem of partiality
and subset types already clearly appears.
The \(\sepunionSym\) operation (here \(+\)) can `overflow'.
We remedy this problem by having all operations operate on
pre-terms (here \(\Q\)) and the predicate \(\sepvalidSym\) describes validity
of pre-terms (here \(0 \le \_ \le 1\)).

\begin{definition}
\label{definition:fractional_algebra}
The \emph{fractional separation algebra \(\Q\)} is defined as: \hfill
\begin{align*}
\sepvalid x \defined{}& 0 \le x \le 1
	& \sepempty \defined{}& 0 \\
\sepdisjoint x y \defined{}& 0 \le x, y\ \land\ x + y \le 1
	& \sepunion x y \defined{}& x + y \\
\sepsubseteq x y \defined{}& 0 \le x \le y \le 1
	& \sepdifference x y \defined{}& x - y
\end{align*}
\end{definition}

The version of separation algebras by Klein \etal~\cite{kle:kol:boy:12} in
\Isabelle{} also models \(\sepunionSym\) as a total operation and uses a
relation \(\sepdisjointSym\).
There are some differences:

\begin{itemize}
\item We include a predicate \(\sepvalidSym\) to prevent having to deal
	with subset types.
\item They have weaker premises for associativity
	(law~\ref{item:sep_associative}), namely
	\(\sepdisjoint x y\), \(\sepdisjoint y z\) and \(\sepdisjoint x z\)
	instead of \(\sepdisjoint x y\) and \(\sepdisjoint {\sepunion x y} z\).
	Ours are more natural, \eg{} for fractional permissions
	one has \(\sepdisjoint {0.5} {0.5}\) but not \(\sepdisjoint {0.5 + 0.5}
	{0.5}\), and it thus makes no sense to require
	\(\sepunion {0.5} {(\sepunion {0.5} {0.5})}
	= \sepunion {(\sepunion {0.5} {0.5})} {0.5}\) to hold.
\item Since \Coq{} (without axioms) does not have a choice operator, the
	\(\sepdifferenceSym\) operation cannot be defined in terms of
	 \(\sepunionSym\).
	 \Isabelle{} has a choice operator.
\end{itemize}

Dockins \etal{}~\cite{doc:hob:app:09} have formalized a hierarchy of
different separation algebras in \Coq{}.
They have dealt with the issue of partiality by treating \(\sepunionSym\) as a
relation instead of a function.
This is unnatural, because equational reasoning becomes impossible and
one has to name all auxiliary results.

Bengtson~\etal{}~\cite{ben:jen:cie:bir:11} have formalized separation algebras
in \Coq{} to reason about object-oriented programs.
They have treated \(\sepunionSym\) as a partial function, and have not defined
any complex permission systems.

\subsection{Permissions}
\label{section:cpermissions}
We classify permissions using \emph{permission kinds}.

\begin{definition}
The lattice of \emph{permission kinds} \((\pkind,\subseteq)\) is defined as:
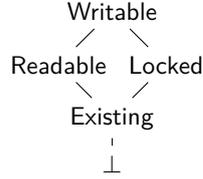
\begin{equation*}
\begin{tikzpicture}[node distance=3em]
\node (Writable) [yshift=1em] {\(\Some\Writable\)};
\node (Readable) [below left of=Writable] {\(\Some\Readable\)};
\node (Locked) [below right of=Writable] {\(\Some\Locked\)};
\node (Existing) [below left of=Locked] {\(\Some\Existing\)};
\node (None) [below of=Existing,yshift=1.2em] {\(\None\)};
\draw (Writable) -- (Readable) -- (Existing) -- (None);
\draw (Writable) -- (Locked) -- (Existing);
\end{tikzpicture}
\end{equation*}
\end{definition}

The order \(k_1 \subseteq k_2\) expresses that \(k_1\) allows fewer operations
than \(k_2\).
This organization of permissions is inspired by
Leroy \etal~\cite{ler:app:bla:ste:12}.
The intuitive meaning of the permission kinds is as follows:

\begin{itemize}
\item \(\Some\Writable\).
	\emph{Writable permissions} allow reading and writing.
\item \(\Some\Readable\).
	\emph{Read-only permissions} allow solely reading.
\item \(\Some\Existing\).
	\emph{Existence permissions}~\cite{bor:cal:hea:par:05} are used for objects
	that are known to exist but whose value cannot be used.
	Existence permissions are essential because \C{} only permits pointer
	arithmetic on pointers that refer to objects that have not been deallocated
	(see Section~\ref{section:challenge:indeterminate_pointers} for discussion).
\item \(\Some\Locked\).
	\emph{Locked permissions} are used to formalize the sequence point
	restriction.
	When an object is modified during the execution of an expression, it is
	temporarily given a locked permission to forbid any read/write accesses until
	the next sequence point.
	
	For example, in \lstinline|(x = 3) + *p;| the assignment
	\lstinline|x = 3| locks the permissions of the object \lstinline|x|.
	Since future read/write accesses to \lstinline|x| are forbidden, accessing
	\lstinline|*p| results in undefined in case \lstinline|p| points to
	\lstinline|x|.
	At the sequence point \lstinline|;|, the original permission of
	\lstinline|x| is restored.
	
	Locked permissions are different from existence permissions because the
	operational semantics can change writable permissions into locked permissions
	and \viceversa, but cannot do that with existing permissions.
\item \(\None\).
	\emph{Empty permissions} allow no operations.
\end{itemize}

In our separation logic we do not only have control which operations are
allowed, but also have to deal with share accounting.

\begin{itemize}
\item We need to subdivide objects with writable or
	read-only permission into multiple parts with read-only permission.
	For example, in the expression \lstinline|x + x|, both subexpressions require
	\lstinline|x| to have at least read-only permission.
\item We need to subdivide objects with writable
	permission into a part with existence permission and a part
	with writable permission.
	For example, in the expression \lstinline|*(p + 1) = (*p = 1)|, the
	subexpression \lstinline|*p = 1| requires \lstinline|*p| to have writable
	permission, and the subexpression \lstinline|*(p + 1)| requires \lstinline|*p|
	to have at least existence permission so as to perform pointer arithmetic on
	\lstinline|p|.
\end{itemize}

When reassembling subdivided permissions (using \(\sepunionSym\)), we need
to know when the original permission is reobtained.
Therefore, the underlying permission system needs to have more structure, and 
cannot consist of just the permission kinds.

\begin{definition}
\label{definition:permissions}
\emph{\CHtwoO{} permissions \(\perm\)} are defined as:
\begin{equation*}
\begin{tikzpicture}[label/.style={pos=1,anchor=south,text depth=0em}]
\node (perm) {\(
	\gamma \in \perm \quad \defined \quad
	\underbrace {\lockable
	  {\counter \Q}}_{\mathclap{\textnormal{non-const qualified}}}
	\qquad\qquad + \qquad\qquad
	\underbrace{\Q}_{\mathclap{\textnormal{const qualified}}}\)};
\draw[thick,<-] ($(perm.north west) + (2.95,0)$) to[out=90,in=-90]
	node[label] {Lockable SA} +(-1.925,1.8em);
\draw[thick,<-] ($(perm.north west) + (3.3,0)$) to[out=90,in=-90]
	node[label] {Counting SA} +(0,1.8em);
\draw[thick,<-] ($(perm.north west) + (3.65,0)$) to[out=90,in=-90]
	node[label] (frac) {Fractional SA} +(2,1.8em);
\draw[thick,<-] ($(perm.north west) + (7.6,0)$) to[out=90,in=-90] (frac.south);
\end{tikzpicture}
\end{equation*}
where \(\lockable A \defined \{ \LLockedSym, \LUnlockedSym \} \times A\)
and \(\counter A \defined \Q \times A\).
\end{definition}

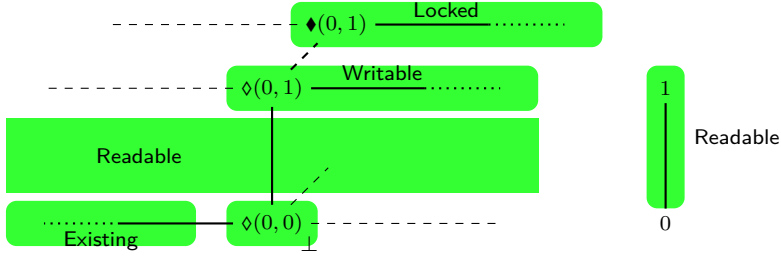
\begin{figure}
\centering
\pgfdeclarelayer{back}
\pgfsetlayers{back,main}
\begin{tikzpicture}[node distance=4em]
\node (1) {\(\LUnlocked{(0,1)}\)};
\node (1before) [below left of=1] {};
\node[below of=1,yshift=-2em] (0) {\(\LUnlocked{(0,0)}\)};
\node[above right of=0] (0after) {};
\node[above right of=1] (Locked) {\(\LLocked{(0,1)}\)};
\draw[thick] (1) -- +(2,0);
 \draw[thick,dotted] ($(1) +(2,0)$) --  +(1,0);
\draw[thin,dashed] (1) -- +(-3,0);
\draw[thin,dashed] (Locked) -- +(-3,0); 
\draw[thick] (0) -- +(-2,0);
 \draw[thick,dotted] ($(0) +(-2,0)$) -- +(-1,0);
\draw[thin,dashed] (0) -- +(3,0); 
\draw[thick] (Locked) -- +(2,0);
 \draw[thick,dotted] ($(Locked) +(2,0)$) -- +(1,0);
\draw[thick] (0) -- (1);
\draw[thick,dashed] (1) -- (Locked);
\draw[thin,dashed] (0) -- (0after);
\begin{pgfonlayer}{back}
\fill[fill=green,opacity=0.8, path fading=west]
  ($(0) +(0,0.4)$) rectangle
  node [opacity=1] {\(\Readable\)} ($(1) +(-3.5,-0.4)$);
\fill[fill=green,opacity=0.8, path fading=east]
  ($(0) +(0,0.4)$) rectangle ($(1) +(3.5,-0.4)$);
\fill[fill=green,fill opacity=0.8, rounded corners,path fading=east]
	($(1) +(-0.6,-0.3)$) rectangle
	node [above,opacity=1] {\(\Writable\)} ($(1) +(3.5,0.3)$);
\fill[fill=green,opacity=0.8, rounded corners,path fading=west]
	($(0) +(-1,-0.3)$) rectangle
	node [below,opacity=1] {\(\Existing\)} ($(0) +(-3.5,0.3)$);
\fill[fill=green,opacity=0.8,rounded corners]
	($(0) +(-0.6,-0.3)$) rectangle
	node [right,opacity=1,xshift=0.9em,yshift=-0.9em] {\(\None\)}
	($(0) +(0.6,0.3)$);
\fill[fill=green,opacity=0.8, rounded corners,path fading=east]
	($(Locked) +(-0.6,-0.3)$) rectangle
	node [above,opacity=1] {\(\Locked\)} ($(Locked) +(3.5,0.3)$);
\end{pgfonlayer}

\node[right=15em of 1] (1ro) {\(1\)};
\node[below of=1ro,yshift=-2em] (0ro) {\(0\)};
\draw[thick] (0ro) -- (1ro);
\begin{pgfonlayer}{back}
\fill[fill=green,opacity=0.8,rounded corners]
	($(0ro) +(-0.25,0.2)$) rectangle
	node [right,opacity=1,xshift=0.9em] {\(\Readable\)}
	($(1ro) +(0.25,0.3)$);
\end{pgfonlayer}
\end{tikzpicture}
\caption{The \CHtwoO{} permission system.}
\label{figure:permissions}
\end{figure}

The author's PhD thesis~\cite{kre:15:phd} gives the exact definition
of the separation algebra operations on permissions by defining these one by one
for the counting separation algebra \(\counterSym\), the lockable separation
algebra \(\lockableSym\), and the separation algebra on sums \(+\).
This section gives a summary of the important aspects of the permission system.

We combine fractional permissions to account for read-only/writable permissions
with counting permissions to account for the number of existence permissions
that have been handed out.
Counting permissions have originally been introduced by Bornat
\etal~\cite{bor:cal:hea:par:05}.
The annotations \(\{ \LLockedSym, \LUnlockedSym \}\) describe whether a
permission is locked \(\LLockedSym\) or not \(\LUnlockedSym\).
Only writable permission have a locked variant.

Const permissions are used for objects declared with the \lstinline|const|
qualifier.
Modifying an object with const permissions results in undefined behavior.
Const permissions do not have a locked variant or a counting
component as they do not allow writing.

Figure~\ref{figure:permissions} indicates the \(\sepvalidSym\) predicate by the
areas marked green and displays how the elements of the permission
system project onto their kinds.
The operation \(\sepunionSym\) is defined roughly as the point-wise
addition and \(\sepdifferenceSym\) as point-wise subtraction.

We will define an operation \(\sephalf : \perm \to \perm\) to subdivide a writable
or read-only permission into read-only permissions.
\begin{flalign*}
\sephalf \gamma \defined{}&
	\begin{cases}
	\LUnlocked {\pair {0.5 \cdot x} {0.5 \cdot y}} &
		\textnormal{if \(\gamma = \LUnlocked {\pair x y}\)} \\
	0.5 \cdot x &
		\textnormal{if \(\gamma = x \in \Q\)} \\
	\gamma & \textnormal{otherwise, dummy value}
	\end{cases}
\end{flalign*}

Given a writable or read-only permission \(\gamma\), the subdivided read-only
permission \(\sephalf \gamma\) enjoys
\(\sepdisjoint {\sephalf \gamma} {\sephalf \gamma}\) and
\(\sepunion {\sephalf \gamma} {\sephalf \gamma} = \gamma\).

The existence permission \(\permtoken \defined \LUnlocked {\pair {-1} 0}\) is
used in combination with the \(\sepdifferenceSym\) operation
to subdivide a writable permission \(\gamma\) into a writable permission
\(\sepdifference \gamma \permtoken\) and
an existence permission \(\permtoken\).
We have \(\sepunion \gamma {(\sepdifference \gamma \permtoken)} = \gamma\) by
law~\ref{item:sep_union_difference} of separation algebras.
Importantly, only objects with \(\LUnlocked {\pair 0 1}\) permission can be
deallocated, whereas objects with \(\sepdifference \gamma \permtoken\)
permission cannot (see Definition~\ref{definition:memory_operations}) because
expressions such as \lstinline|(p == p) + (free(p),0)| have undefined behavior.

\subsection{Extended separation algebras}

We extend separation algebras with a split operation \(\sephalfSym\) and
predicates to distinguish permissions in our memory model.

\begin{definition}
\label{definition:ext_separation_algebra}
An \emph{extended separation algebra} extends a separation algebra with:
\begin{itemize}
\item Predicates \(\sepsplittableSym, \sepunmappedSym,
	\sepunsharedSym : A \to \Prop\) 
\item A unary operation \(\sephalfSym : A \to A\)
\end{itemize}
\noindent
Satisfying the following laws:
\begin{enumerate}[resume=sep_alg]
\item If \(\sepdisjoint x x\), then \(\sepsplittable {(\sepunion x x)}\)
	\label{item:sep_splittable_union}
\item If \(\sepsplittable x\), then \(\sepdisjoint {\sephalf x} {\sephalf x}\)
	and \(\sepunion {\sephalf x} {\sephalf x} = x\)
	\label{item:sep_union_half}
\item If \(\sepsplittable y\) and \(\sepsubseteq x y\), then
	\(\sepsplittable x\)
	\label{item:sep_splittable_weaken}
\item If \(\sepdisjoint x y\) and \(\sepsplittable {(\sepunion x y)}\),
	then \(\sephalf {(\sepunion x y)} = \sepunion {\sephalf x} {\sephalf y}\)
	\label{item:sep_union_half_distr}
\item \(\sepunmapped \sepempty\),
	and if \(\sepunmapped x\), then \(\sepvalid x\)
	\label{item:sep_unmapped_empty_valid}
\item If \(\sepunmapped y\) and \(\sepsubseteq x y\), then \(\sepunmapped x\)
\item If \(\sepdisjoint x y\), \(\sepunmapped x\) and \(\sepunmapped y\),
	then \(\sepunmapped {(\sepunion x y)}\)
	\label{item:sep_unmapped_union}
\item \(\sepunshared x\) iff \(\sepvalid x\) and
	for all \(y\) with \(\sepdisjoint x y\) we have \(\sepunmapped y\)
	\label{item:sep_unshared_spec}
\item Not both \(\sepunshared x\) and \(\sepunmapped x\)
\item There exists an \(x\) with \(\sepvalid x\) and
	not \(\sepunmapped x\)
	\label{item:sep_inhabited}
\end{enumerate}
\end{definition}

The \(\sephalfSym\)-operation is partial, but described by a total function
whose result \(\sephalf x\) is only meaningful if \(\sepsplittable x\) holds.
Law~\ref{item:sep_splittable_weaken} ensures that splittable permissions
are infinitely splittable, and law~\ref{item:sep_union_half_distr}
ensures that \(\sephalfSym\) distributes over \(\sepunionSym\).

The predicates \(\sepunmappedSym\) and \(\sepunsharedSym\) associate an
intended semantics to elements of a separation algebra.
Let us consider fractional permissions to indicate the intended meaning of
these predicates.

\begin{definition}
The \emph{fractional separation algebra \(\Q\)} is extended with:
\begin{align*}
\sepsplittable x \defined{}& 0 \le x \le 1
	& \sephalf x \defined{}& 0.5 \cdot x \\
\sepunmapped x \defined{}& x = 0
	& \sepunshared x \defined{}& x = 1
\end{align*}
\end{definition}

Remember that permissions will be used to annotate each individual bit in
memory.
Unmapped permissions are \emph{on the bottom}: they do not allow their bit to
be used in any way.
Exclusive permissions are \emph{on the top}: they are the sole owner of a bit
and can do anything to that bit without affecting disjoint bits.

Fractional permissions have exactly one unmapped element and exactly one
exclusive element, but in the \CHtwoO{} permission system this is not the case.
The elements of the \CHtwoO{} permission system are classified as follows:

\begin{center}
\begin{tabular}{c|c|l}
\(\ \sepunmappedSym\ \) & \(\ \sepunsharedSym\ \) & Examples \\ \hline \hline
& \checkmark & \(\Some\Writable\) and \(\Some\Locked\) permissions \\
& & \(\Some\Readable\) permissions \\
\checkmark & & The \(\sepempty\) permission and \(\Some\Existing\) permissions
\end{tabular}
\end{center}

In order to abstractly describe bits annotated with permissions we
define the tagged separation algebra \(\tagged A t T\).
In its concrete use \(\tagged \perm \BIndet \bit\) in the memory model
(Definition~\ref{definition:pbit}), the
elements \(\pair \gamma b\) consist of a permission \(\gamma \in \perm\) and bit
\(b \in \bit\).
We use the symbolic bit \(\BIndet\) that represents indeterminate
storage to ensure that bits with \(\sepunmappedSym\) permissions indeed have no 
usable value.

\begin{definition}
\label{definition:sep_tagged}
Given a separation algebra \(A\) and a set of tags \(T\) with default tag
\(t \in T\), the \emph{tagged separation algebra
\(\tagged A t T \defined A \times T\) over \(A\)} is defined as:
\begin{flalign*}
\sepvalid {\pair x y} \defined{}&
	\sepvalid x \land (\sepunmapped x \impl y = t) \\
\sepempty \defined{}& \pair \sepempty t \\
\sepdisjoint {\pair x y} {\pair {x'} {y'}} \defined{}&
	\sepdisjoint x {x'}
	\begin{array}[t]{l}
	\land\ (\sepunmapped x \lor y = y' \lor \sepunmapped {x'}) \\
	\land\ (\sepunmapped x \impl y = t)
	\land (\sepunmapped {x'} \impl y' = t)
	\end{array} \\
\sepunion {\pair x y} {\pair {x'} {y'}} \defined{}&
	\begin{cases}
	\pair {\sepunion x {x'}} {y'} & \textnormal{if \(y = t\)} \\
	\pair {\sepunion x {x'}} y & \textnormal{otherwise}
	\end{cases} \\
\sepsplittable {\pair x y} \defined{}&
	\sepsplittable x \land (\sepunmapped x \impl y = t) \\
\sephalf {\pair x y} \defined{}& \pair {\sephalf x} y \\
\sepunmapped {\pair x y} \defined{}& \sepunmapped x \land y = t \\
\sepunshared {\pair x y} \defined{}& \sepunshared x 
\end{flalign*}
The definitions of the omitted relations and operations are as expected.
\end{definition}

\section{The memory model}
\label{section:memory}

This section defines the \CHtwoO{} memory model whose external interface
consists of operations with the following types:
\allowdisplaybreaks[1]
\begin{flalign*}
\memlookupSym \Gamma :{}& \addr\to\mem\to\option\val \\
\memforceSym \Gamma :{}& \addr\to\mem\to\mem \\
\meminsertSym \Gamma :{}& \addr\to\mem\to\val\to\mem \\ 
\memwritableSym \Gamma :{}& \addr\to\mem\to\Prop \\ 
\memlockSym \Gamma :{}& \addr\to\mem\to\mem \\
\memunlockSym :{}& \lockset \to\mem\to\mem \\
\memallocSym \Gamma :{}& \memindex\to\val\to\bool\to\mem\to\mem \\
\mdom :{}& \mem \to \finset \memindex \\
\memfreeableSym :{}& \addr\to\mem\to\Prop \\
\memfreeSym :{}& \memindex\to\mem\to\mem
\end{flalign*}

\begin{notation}
We let \(\memlookup \Gamma a m \defined \memlookupPlain \Gamma a m\) and
\(\meminsert \Gamma a v m \defined \meminsertPlain \Gamma a v m\).
\end{notation}

Many of these operations depend on the typing environment \(\Gamma\) which
assigns fields to structs and unions (Definition~\ref{definition:env}).
This dependency is required because these operations need to be aware of the
layout of structs and unions.

The operation \(\memlookup \Gamma a m\) yields the value stored at address
\(a\) in memory \(m\).
It fails with \(\None\) if the permissions are insufficient,
effective types are violated, or \(a\) is an end-of-array
address.
Reading from (the abstract) memory is not a pure operation.
Although it does not affect the memory contents, it may affect the effective
types~\cite[6.5p6-7]{iso:12}.
This happens for example in case type-punning is performed (see
Section~\ref{section:challenge:type_punning}).
This impurity is factored out by the operation \(\memforce \Gamma a m\).

The operation \(\meminsert \Gamma a v m\) stores the value \(v\) at address
\(a\) in memory \(m\).
A store is only permitted in case permissions are sufficient, effective
types are not violated, and \(a\) is not an end-of-array address.
The proposition \(\memwritable \Gamma a m\) describes the side-conditions
necessary to perform a store.

After a successful store, the operation \(\memlock \Gamma a m\) is used to
lock the object at address \(a\) in memory \(m\).
The lock operation temporarily reduces the permissions to \(\Locked\)
so as to prohibit future accesses to \(a\).
Locking yields a formal treatment of the sequence point restriction (which
states that modifying an object more than once between two sequence points 
results in undefined behavior, see
Section~\ref{section:challenge:sequence_point}).

The operational semantics accumulates a set \(\Omega \in \lockset\) of addresses
that have been written to (Definition~\ref{definition:lockset})
and uses the operation \(\memunlock \Omega m\) at the subsequent sequence point
(which may be at the semicolon that terminates a full expression).
The operation \(\memunlock \Omega m\) restores the permissions of the addresses
in \(\Omega\) and thereby makes future accesses to the addresses in \(\Omega\)
possible again.
The author's PhD thesis~\cite{kre:15:phd} describes in detail how sequence points
and locks are treated in the operational semantics.

The operation \(\memalloc \Gamma o v \mu m\) allocates a new object with
value \(v\) in memory \(m\).
The object has object identifier \(o \notin \mdom m\) which is
non-deterministically chosen by the operation semantics.
The Boolean \(\mu\) expresses whether the new object is allocated by
\lstinline|malloc|.

Accompanying \(\memallocSym \Gamma\), the operation \(\memfree o m\) deallocates
a previously allocated object with object identifier \(o\) in memory \(m\).
In order to deallocate dynamically obtained memory via \lstinline|free|, the
side-condition \(\memfreeable a m\) describes that the permissions are
sufficient for deallocation, and that \(a\) points to a
\lstinline|malloc|ed object.

\subsection{Representation of pointers}
\label{section:pointers}

Adapted from \CompCert~\cite{ler:bla:08,ler:app:bla:ste:12},
we represent memory states as finite
partial functions from \emph{object identifiers} to \emph{objects}.
Each local, global and static variable, as well as each invocation of
\lstinline|malloc|, is
associated with a unique object identifier of a separate object in memory.
This approach separates unrelated objects by construction, and is therefore
well-suited for reasoning about memory transformations.

We improve on \CompCert{} by modeling objects as structured
trees instead of arrays of bytes to keep track of padding bytes and
the variants of unions.
This is needed to faithfully describe \Celeven's notion of effective types
(see page~\pageref{page:effective_types} of Section~\ref{section:introduction} for
an informal description).
This approach allows us to describe various undefined behaviors of
\Celeven{} that have not been considered by others (see
Sections~\ref{section:challenge:byte_level}
and~\ref{section:challenge:type_punning}).

In the \CompCert{} memory model, pointers are represented as pairs \((o,i)\)
where \(o\) is an object identifier and \(i\) is
a byte offset into the object with object identifier \(o\).
Since we represent objects as trees instead of as arrays of bytes, we represent
pointers as paths through these trees rather than as byte offsets.

\begin{definition}
\emph{Object identifiers \(o \in \memindex\)} are elements of a fixed countable
set.
In the \Coq{} development we use binary natural numbers, but since
we do not rely on any properties apart from countability, we keep the
representation opaque.
\end{definition}

We first introduce a typing environment to relate the shape of paths
representing pointers to the types of objects in memory.

\begin{definition}
\label{definition:memenv}
\label{definition:index_typed}
\label{definition:index_alive}
\emph{Memory typing environments \(\Delta \in \memenv\)} are finite
partial functions \(\map \memindex {(\type \times \bool)}\).
Given a memory environment \(\Delta\):
\begin{enumerate}
\item An \emph{object identifier \(o\) has type \(\tau\)},
	notation \(\indextyped \Delta o \tau\),
	if \(\mlookup o \Delta = \Some {\pair \tau \beta}\)
	for a \(\beta\).
\item An \emph{object identifier \(o\) is alive},
	notation \(\indexalive \Delta o\),
	if \(\mlookup o \Delta = \Some {\pair \tau \false}\)
	for a \(\tau\).
\end{enumerate}
\end{definition}

Memory typing environments evolve during program execution.
The code below is annotated with the corresponding memory environments in red.

\begin{lstlisting}
short x;
$\color{red}{\Delta_1 =	\{
	o_1 \mapsto \pair {\TBase {\TInt {\IntType \Signed \shortrank}}} \false \}}$
int *p;    
$\color{red}{\Delta_2 = \{
	o_1 \mapsto \pair {\TBase {\TInt {\IntType \Signed \shortrank}}} \false,
	o_2 \mapsto \pair {\TBase {\TPtr {\TType {\TBase {\TInt \TSignedInt}}}}} \false \}}$
p = malloc(sizeof(int));
$\color{red}{\Delta_3 = \{
	o_1 \mapsto \pair {\TBase {\TInt {\IntType \Signed \shortrank}}} \false,
	o_2 \mapsto \pair {\TBase {\TPtr {\TType {\TBase {\TInt \TSignedInt}}}}} \false,
	o_3 \mapsto \pair {\TBase {\TInt \TSignedInt}} \false} \}$
free(p);
$\color{red}{\Delta_4 = \{
	o_1 \mapsto \pair {\TBase {\TInt {\IntType \Signed \shortrank}}} \false,
	o_2 \mapsto \pair {\TBase {\TPtr {\TType {\TBase {\TInt \TSignedInt}}}}} \false,
	o_3 \mapsto \pair {\TBase {\TInt \TSignedInt}} \true \}}$
\end{lstlisting}

Here, \(o_1\) is the object identifier of the variable \lstinline|x|,
\(o_2\) is the object identifier of the variable \lstinline|p| and \(o_3\)
is the object identifier of the storage obtained via \lstinline|malloc|.

Memory typing environments also keep track of objects that have been
deallocated.
Although one cannot directly create a pointer to a deallocated object, existing
pointers to such objects remain in memory after deallocation (see the pointer
\lstinline|p| in the above example).
These pointers, also called \emph{dangling} pointers, cannot actually be used.

\begin{definition}
\label{definition:pointers}
\emph{References}, \emph{addresses} and \emph{pointers} are
inductively defined as:
\begin{flalign*}
\syntax{refseg} \\
\syntax{ref} \\
\syntax{addr} \\
\syntax{ptr}
\end{flalign*}
\end{definition}

References are paths from the top of an object in memory to some subtree of that
object.
The shape of references matches the structure of types:

\begin{itemize}
\item The reference \(\RArray i \tau n\) is used to select the \(i\)th element
	of a \(\tau\)-array of length \(n\).
\item The reference \(\RStruct i t\) is used to select the \(i\)th field of a
	struct \(t\).
\item The reference \(\RUnion i t q\) is used to select the \(i\)th variant of a
	union \(t\).
\end{itemize}

References can describe most pointers in \C{} but cannot account for
end-of-array pointers and pointers to individual bytes.
We have therefore defined the richer notion of \emph{addresses}.
An address \(\Addr o \tau {\vec r} i \sigma {\ptype\sigma}\) consists of:

\begin{itemize}
\item An object identifier \(o\) with type \(\tau\).
\item A reference \(\vec r\) to a subobject of type \(\sigma\) in the entire
	object of type \(\tau\).
\item An offset \(i\) to a particular byte in the subobject of type \(\sigma\)
	(note that one cannot address individual bits in \C).
\item The type \(\ptype\sigma\) to which the address is cast.
	We use a points-to type in order to account for casts to the anonymous
	\lstinline|void*| pointer, which is represented as the points-to type
	\(\TAny\).
	This information is needed to define, for example, pointer arithmetic,
	which is sensitive to the type of the address.
\end{itemize}

In turn, pointers extend addresses with a \cNULL{} pointer
\(\NULL {\ptype\sigma}\) for each type \(\ptype\sigma\), and function pointers
\(\FunPtr f {\vec \tau} {\tau}\) which contain the name and type of a function.

Let us consider the following global variable declaration:

\begin{lstlisting}
struct S {
  union U { signed char x[2]; int y; } u;
  void *p;
} s;
\end{lstlisting}

The formal representation of the pointer \lstinline|(void*)(s.u.x + 2)| is:
\[\Ptr {\Addr {o_{\mathtt s}}
	{\TStruct {\mathtt S}}
	{\RStruct 0 {\mathtt S}\,
		\RUnion 0 {\mathtt U} \RUnfrozen\,
		\RArray 0 {\TBase {\TInt {\IntType \Signed \charrank}}} 2}
	2 {\TBase {\TInt {\IntType \Signed \charrank}}} \TAny}.
\]

Here, \(o_{\mathtt s}\) is the object identifier associated with the
variable \lstinline|s| of type \(\TStruct {\mathtt S}\).
The reference \(\RStruct 0 {\mathtt S}\, \RUnion 0 {\mathtt U} \RUnfrozen\,
\RArray 0 {\TBase {\TInt {\IntType \Signed \charrank}}} 2\) and byte-offset
\(2\) describe that the pointer refers to the third byte of the array
\lstinline|s.u.x|.
The pointer refers to an object of type
\(\TBase {\TInt {\IntType \Signed \charrank}}\).
The annotation \(\TAny\) describes that the pointer has been cast to type
\lstinline|void*|.

The annotations \(q \in \{ \RFrozen, \RUnfrozen \}\) on references
\(\RUnion i s q\) describe whether type-punning is allowed or not.
The annotation \(\RUnfrozen\) means that type-punning is allowed, \ie{}
accessing another variant than the current one has defined behavior.
The annotation \(\RFrozen\) means that type-punning is forbidden.
A pointer whose annotations are all of the shape \(\RFrozen\), and
thereby does not allow type-punning at all, is called \emph{frozen}.

\begin{definition}
\label{definition:ptr_freeze}
The \emph{freeze} function \(\freezeSym : \rrefseg \to \rrefseg\) is defined as:
\begin{equation*}
\freeze {\RArray i \tau n} \defined{} \RArray i \tau n \qquad
\freeze {\RStruct i t} \defined{} \RStruct i t \qquad
\freeze {\RUnion i t q} \defined{} \RUnion i t \RFrozen
\end{equation*}
A reference segment \(r\) is \emph{frozen}, notation \(\frozen r\), if
\(\freeze r = r\).
Both \(\freezeSym\) and \(\frozenSym\) are lifted to references,
addresses, and pointers in the expected way.
\end{definition}

Pointers stored in memory are always in frozen shape.
Definitions~\ref{definition:ctree_lookup} and~\ref{definition:base_val_flatten}
describe the formal treatment of effective types and frozen pointers, but for
now we reconsider the example from Section~\ref{section:challenge:type_punning}:

\begin{lstlisting}
union U { int x; short y; } u = { .x = 3 };
short *p = &u.y;
printf("%d\n", *p);  // Undefined
printf("%d\n", u.y); // OK
\end{lstlisting}

Assuming the object \lstinline|u| has object identifier \(o_{\mathtt u}\), the
pointers \lstinline|&u.x|, \lstinline|&u.y| and \lstinline|p| have the following
formal representations:
\begin{flalign*}
\textnormal{\lstinline|\&u.x|:}\quad&
\Ptr {\Addr {o_{\mathtt u}}
	{\TUnion {\mathtt U}}
	{\RUnion 0 {\mathtt U} \RUnfrozen}
	0
	{\TBase \TSignedInt}
	{\TType {\TBase \TSignedInt}}} \\
\textnormal{\lstinline|\&u.y|:}\quad&
\Ptr {\Addr {o_{\mathtt u}}
	{\TUnion {\mathtt U}}
	{\RUnion 1 {\mathtt U} \RUnfrozen}
	0
	{\TBase {\TInt {\IntType \Signed \shortrank}}}
	{\TType {\TBase {\TInt {\IntType \Signed \shortrank}}}}} \\
\textnormal{\lstinline|p|:}\quad&
\Ptr {\Addr {o_{\mathtt u}}
	{\TUnion {\mathtt U}}
	{\RUnion 0 {\mathtt U} \RFrozen}
	0
	{\TBase {\TInt {\IntType \Signed \shortrank}}}
	{\TType {\TBase {\TInt {\IntType \Signed \shortrank}}}}}
\end{flalign*}

These pointers are likely to have the same object representation on
actual computing architectures.
However, due to effective types, \lstinline|&u.y| may be used for type-punning
but \lstinline|p| may not.
It is thus important that we distinguish these pointers in the formal
memory model.

The additional structure of pointers is also needed to determine whether
pointer subtraction has defined behavior.
The behavior is only defined if the given pointers both point
to an element of the same array object~\cite[6.5.6p9]{iso:12}.
Consider:

\begin{lstlisting}
struct S { int a[3]; int b[3]; } s;
s.a - s.b;       // Undefined, different array objects
(s.a + 3) - s.b; // Undefined, different array objects
(s.a + 3) - s.a; // OK, same array objects
\end{lstlisting}

Here, the pointers \lstinline|s.a + 3| and \lstinline|s.b|
have different representations in the \CHtwoO{} memory model.
The author's PhD thesis~\cite{kre:15:phd} gives the formal definition of pointer
subtraction.

We will now define typing judgments for references, addresses and pointers.
The judgment for references \(\reftyped \Gamma {\vec r} \tau \sigma\) states
that \(\sigma\) is a \emph{subobject type of} \(\tau\) which can be obtained
via the reference \(\vec r\) (see also
Definition~\ref{definition:subtype}).
For example, \lstinline|int[2]| is a subobject type of
\lstinline|struct S { int x[2]; int y[3]; }| via
\(\RStruct 0 {\mathtt S}\).

\begin{definition}
\label{definition:ref_typed}
The judgment \(\reftyped \Gamma {\vec r} \tau \sigma\) describes that
\emph{\(\vec r\) is a valid reference from \(\tau\) to \(\sigma\)}.
It is inductively defined as:
\begin{gather*}
\AXC{\strut}
\UIC{\strut\(\reftyped \Gamma \nil \tau \tau\)}
\DP\qquad
\AXC{\strut\(\reftyped \Gamma {\vec r} \tau {\TArray \sigma n}\)}
\AXC{\strut\(i < n\)}
\BIC{\(\reftyped \Gamma {\vec r \RArray i \sigma n} \tau \sigma\)}
\DP\\[0.5em]
\renewcommand{\defaultHypSeparation}{\hskip.2em}
\AXC{\(\reftyped \Gamma {\vec r} \tau {\TStruct t}\)}
\AXC{\(\elookup t \Gamma = \Some {\vec\sigma}\)}
\AXC{\(i < \length {\vec\sigma}\)}
\TIC{\(\reftyped \Gamma {\vec r \RStruct i t} \tau {\sigma_i}\)}
\DP\quad
\renewcommand{\defaultHypSeparation}{\hskip.2em}
\AXC{\(\reftyped \Gamma {\vec r} \tau {\TUnion t}\)}
\AXC{\(\elookup t \Gamma = \Some {\vec\sigma}\)}
\AXC{\(i < \length {\vec\sigma}\)}
\TIC{\(\reftyped \Gamma {\vec r \RUnion i t q} \tau {\sigma_i}\)}
\DP
\end{gather*}
\end{definition}

The typing judgment for addresses is more involved than the judgment for
references.
Let us first consider the following example:

\begin{lstlisting}
int a[4];
\end{lstlisting}

Assuming the object \lstinline|a| has object identifier \(o_{\mathtt a}\), the
end-of-array pointer \lstinline|a+4| could be represented in at least the
following ways
(assuming \(\sizeof {} {(\TBase {\TInt {\IntType \Signed \intrank}})} = 4\)):
\begin{gather*}
\Ptr {\Addr {o_{\mathtt a}}
	{\TArray {\TBase {\TInt \TSignedInt}} 4}
	{\RArray 0 {\TBase {\TInt \TSignedInt}} 4}
	{16}
	{\TBase {\TInt \TSignedInt}}
	{\TType {\TBase {\TInt \TSignedInt}}}} \\
\Ptr {\Addr {o_{\mathtt a}}
	{\TArray {\TBase {\TInt \TSignedInt}} 4}
	{\RArray 3 {\TBase {\TInt \TSignedInt}} 4}
	4
	{\TBase {\TInt \TSignedInt}}
	{\TType {\TBase {\TInt \TSignedInt}}}}
\end{gather*}

In order to ensure canonicity of pointer representations, we let the typing
judgment for addresses ensure that the reference \(\vec r\) of
\(\Addr o \tau {\vec r} i \sigma {\ptype\sigma}\)
always refers to the first element of an array subobject.
This renders the second representation illegal.

\begin{definition}
The relation \(\ptrcastable \tau {\ptype\sigma}\), type \emph{\(\tau\) is
pointer castable to \(\ptype\sigma\)}, is inductively defined by
\(\ptrcastable \tau {\TType \tau}\),
\(\ptrcastable \tau {\TType {\TBase {\TInt {\IntType \Unsigned \charrank}}}}\),
and \(\ptrcastable \tau \TAny\).
\end{definition}

\begin{definition}
\label{definition:addr_typed}
The judgment \(\ptrtyped \Gamma \Delta a {\ptype\sigma}\) describes that
\emph{the address \(a\) refers to type \(\ptype\tau\)}.
It is inductively defined as:
\begin{gather*}
\AXC{\(\indextyped \Delta o \tau\qquad
	\typevalid \Gamma \tau\qquad
	\reftyped \Gamma {\vec r} \tau \sigma\quad\)}
\noLine\UIC{\(\refoffset {\vec r} = 0\qquad
	i \le \sizeof \Gamma \sigma \cdot \refsize {\vec r}\qquad
	\sizeof \Gamma {\ptype\sigma} \divides i\qquad
	\ptrcastable \sigma {\ptype\sigma}\)}
\UIC{\(\addrtyped \Gamma \Delta
	{\Addr o \tau {\vec r} i \sigma {\ptype\sigma}} {\ptype\sigma}\)}
\DP
\end{gather*}

\noindent
Here, the helper functions \(\refoffsetSym, \refsizeSym : \rref \to \nat\) are
defined as:
\begin{gather*}
\refoffset {\vec r} \defined{}
	\begin{cases}
	i & \textnormal{if \(\vec r = \vec r_2 \RArray i \tau n\)} \\
	0 & \textnormal{otherwise}
	\end{cases}
\qquad
\refsize {\vec r} \defined{}
	\begin{cases}
	n & \textnormal{if \(\vec r = \vec r_2 \RArray i \tau n\)} \\
	1 & \textnormal{otherwise}
	\end{cases}
\end{gather*}
\end{definition}

We use an intrinsic encoding of syntax, which means that terms
contain redundant type annotations so we can read off types.
Functions to read off types are named \(\typeofSym\) and will not be defined
explicitly.
Type annotations make it more convenient to define
operations that depend on types (such as \(\refoffsetSym\) and \(\refsizeSym\)
in Definition~\ref{definition:addr_typed}).
As usual, typing judgments ensure that type annotations are consistent.

The premises \(i \le \sizeof \Gamma \sigma \cdot \refsize {\vec r}\) and
\(\sizeof \Gamma {\ptype\sigma} \divides i\) of the typing rule ensure that the
byte offset \(i\) is aligned and within range.
The inequality \(i \le \sizeof \Gamma \sigma \cdot \refsize {\vec r}\)
is non-strict so as to allow end-of-array pointers.

\begin{definition}
\label{definition:addr_strict}
An address \(\Addr o \tau {\vec r} i \sigma {\ptype\sigma}\) is called
\emph{strict}, notation \(\addrstrict \Gamma a\), in case it satisfies
\(i < \sizeof \Gamma \sigma \cdot \refsize {\vec r}\).
\end{definition}

The judgment \(\ptrcastable \tau {\ptype\sigma}\) does not describe
the typing restriction of cast expressions.
Instead, it defines the invariant that each address
\(\Addr o \tau {\vec r} i \sigma {\ptype\sigma}\) should satisfy.
Since \C{} is not type safe, pointer casting
has \(\ptrcastable \tau {\ptype\sigma}\)
as a run-time side-condition:

\begin{lstlisting}
int x, *p = &x;
void *q = (void*)p;    // OK, $\ptrcastable {\TBase {\TInt \TSignedInt}} \TAny$
int *q1 = (int*)q;     // OK, $\ptrcastable {\TBase {\TInt \TSignedInt}}%
                          {\TType {\TBase {\TInt \TSignedInt}}}$
short *q2 = (short*)p; // Statically ill-typed                    
short *q3 = (short*)q; // Undefined behavior, $\notptrcastable {\TBase {\TInt \TSignedInt}}%
                          {\TType {\TBase {\TInt {\IntType \Signed \shortrank}}}}$
\end{lstlisting}

\begin{definition}
\label{definition:ptr_typed}
The judgment \(\ptrtyped \Gamma \Delta p {\ptype\sigma}\) describes that
\emph{the pointer \(p\) refers to type \(\ptype\tau\)}.
It is inductively defined as:
\begin{gather*}
\AXC{\(\ptrtypevalid \Gamma {\ptype\sigma}\)}
\UIC{\(\ptrtyped \Gamma \Delta {\NULL {\ptype\sigma}} {\ptype\sigma}\)}
\DP\qquad
\AXC{\(\addrtyped \Gamma \Delta a {\ptype\sigma}\)}
\UIC{\(\ptrtyped \Gamma \Delta {\Ptr a} {\ptype\sigma}\)}
\DP\qquad
\AXC{\(\elookup f \Gamma = \Some {\pair {\vec\tau} \tau}\)}
\UIC{\(\ptrtyped \Gamma \Delta
	{\FunPtr f {\vec\tau} \tau} {\TFun {\vec\tau} \tau}\)}
\DP
\end{gather*}
\end{definition}

Addresses \(\Addr o \tau {\vec r \RArray j \sigma n} i \sigma {\ptype\sigma}\)
that point to an element of \(\TArray \tau n\) always have their
reference point to the first element of the array, \ie{} \(j = 0\).
For some operations we use the \emph{normalized reference} which refers to the
actual array element.

\begin{definition}
\label{definition:addrref}
The functions \(\addrindexSym : \addr \to \memindex\),
\(\addrrefSym \Gamma : \addr \to \rref\), and
\(\addrrefbyteSym \Gamma : \addr \to \nat\) obtain the \emph{index},
\emph{normalized reference}, and \emph{normalized byte offset}.
%
\begin{flalign*}
\addrindex {\Addr o \tau {\vec r} i \sigma {\ptype\sigma}} \defined{}& o \\
\addrref \Gamma {\Addr o \tau {\vec r} i \sigma {\ptype\sigma}} \defined{}&
	\refsetoffset {(i \div \sizeof \Gamma \sigma)} {\vec r} \\
\addrrefbyte \Gamma {\Addr o \tau {\vec r} i \sigma {\ptype\sigma}} \defined{}&
  i \mod (\sizeof \Gamma \sigma)
\end{flalign*}
Here, the function \(\refsetoffsetSym : \nat \to \rref \to \rref\) is defined
as:
\begin{gather*}
\refsetoffset j {\vec r} \defined{}
	\begin{cases}
	\vec{r_2} \RArray j \tau n &
		\textnormal{if \(\vec r = \vec r_2 \RArray i \tau n\)} \\
	r & \textnormal{otherwise}
	\end{cases}
\end{gather*}
\end{definition}

Let us display the above definition graphically.
Given an address \(\Addr o \tau {\vec r} i \sigma {\ptype\sigma}\), the
normalized reference and normalized byte offset are as follows:
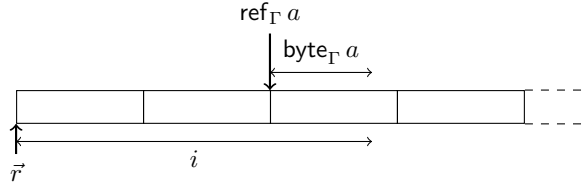
\begin{equation*}
\begin{tikzpicture}[
	start chain=going right,
	node distance=-\pgflinewidth,
	every node/.style={shape=rectangle,minimum height=1.3em,minimum width=5em}
]
\node[draw,on chain] (x1) {};
\node[draw,on chain] (x2) {};
\node[draw,on chain] (x3) {};
\node[draw,on chain] (x4) {};
\draw[dashed] (x4.north east) -- +(2.5em,0);
\draw[dashed] (x4.south east) -- +(2.5em,0);
\draw[thick,<-] (x1.south west) -- +(0,-0.4) node[below] {\(\vec r\)};
\draw[transform canvas={yshift=-0.7em},<->] (x1.south west) --
	node[below]{\(i\)} ($(x3.south) +(1.5em,0)$);
\draw[thick,<-] (x3.north west) -- +(0,0.75) node[above] {\(\addrref \Gamma a\)};
\draw[transform canvas={yshift=0.7em},<->] (x3.north west) --
	node[above]{\(\addrrefbyte \Gamma a\)} ($(x3.north) +(1.5em,0)$);
\end{tikzpicture}
\end{equation*}

For end-of-array addresses the normalized reference is ill-typed because
references cannot be end-of-array.
For strict addresses the normalized reference is well-typed.

\begin{definition}
The judgment \(\ptralive \Delta p\) describes that \emph{the pointer \(p\) is
alive}.
It is inductively defined as:
\begin{equation*}
\AXC{\strut}
\UIC{\(\ptralive \Delta {\NULL {\ptype\sigma}}\)}
\DP\qquad
\AXC{\strut\(\indexalive \Delta  {(\addrindex a)}\)}
\UIC{\(\ptralive \Delta {\Ptr a}\)}
\DP\qquad
\AXC{\strut}
\UIC{\(\ptralive \Delta {\FunPtr f {\vec\tau} \tau}\)}
\DP
\end{equation*}
The judgment \(\indexalive \Delta o\) on object identifiers is defined in
Definition~\ref{definition:index_alive}.
\end{definition}

For many operations we have to distinguish addresses that refer to an entire
object and addresses that refer to an individual byte of an object.
We call addresses of the later kind \emph{byte addresses}.
For example:

\begin{lstlisting}
int x, *p = &x;                        // p is not a byte address
unsigned char *q = (unsigned char*)&x; // q is a byte address
\end{lstlisting}

\begin{definition}
An address \(\Addr o \tau {\vec r} i \sigma {\ptype\sigma}\) is \emph{a byte
address} if \(\TType\sigma \neq \ptype\sigma\).
\end{definition}

To express that memory operations commute (see for example
Lemma~\ref{lemma:cmap_alter_commute}), we need to express that addresses are
\emph{disjoint}, meaning they do not overlap.
Addresses do not overlap if they belong to different objects or take a
different branch at an array or struct.
Let us consider an example:

\begin{lstlisting}
union { struct { int x, y; } s; int z; } u1, u2;
\end{lstlisting}

The pointers \lstinline|&u1| and \lstinline|&u2| are disjoint because they
point to separate memory objects.
Writing to one does not affect the value of the other and \viceversa.
Likewise, \lstinline|&u1.s.x| and \lstinline|&u1.s.y| are disjoint because they
point to different fields of the same struct, and as such do not affect each
other.
The pointers \lstinline|&u1.s.x| and \lstinline|&u1.z| are disjoint because they
point to overlapping objects and thus do affect each other.

\begin{definition}
\emph{Disjointness of references \(\vec r_1\) and \(\vec r_2\)}, notation
\(\refdisjoint {\vec r_1} {\vec r_2}\), is inductively defined as:
\begin{gather*}
\AXC{\(\freeze {\vec r_1} = \freeze {\vec r_2}\)}
\AXC{\(i \neq j\)}
\BIC{\(\refdisjoint {\app {\vec r_1 \RArray i \sigma n} {\vec r_3}}
	{\app {\vec r_2 \RArray j \sigma n} {\vec r_4}}\)}
\DP\qquad
\AXC{\(\freeze {\vec r_1} = \freeze {\vec r_2}\)}
\AXC{\(i \neq j\)}
\BIC{\(\refdisjoint {\app {\vec r_1 \RStruct i t} {\vec r_3}}
	{\app {\vec r_2 \RStruct j t} {\vec r_4}}\)}
\DP
\end{gather*}
\end{definition}

Note that we do not require a special case for
\(\freeze {\vec r_1} \neq \freeze {\vec r_2}\).
Such a case is implicit because disjointness is defined in terms of prefixes.

\begin{definition}
\emph{Disjointness of addresses \(a_1\) and \(a_2\)}, notation
\(\addrdisjoint \Gamma {a_1} {a_2}\), is inductively defined as:
\begin{gather*}
\AXC{\(\addrindex {a_1} \neq \addrindex {a_2}\)}
\UIC{\(\addrdisjoint \Gamma {a_1} {a_2}\)}
\DP\qquad
\AXC{\(\addrindex {a_1} = \addrindex {a_2}\)}
\AXC{\(\refdisjoint {\addrref \Gamma {a_1}} {\addrref \Gamma {a_2}}\)}
\BIC{\(\addrdisjoint \Gamma {a_1} {a_2}\)}
\DP\\[0.3em]
\AXC{both \(a_1\) and \(a_2\) are byte addresses}
\noLine\UIC{\(\addrindex {a_1} = \addrindex {a_2}\) \quad
 \(\freeze {\addrref \Gamma {a_1}} = \freeze {\addrref \Gamma {a_2}}\) \quad
 \(\addrrefbyte \Gamma {a_1} \neq \addrrefbyte \Gamma {a_2}\)}
\UIC{\(\addrdisjoint \Gamma {a_1} {a_2}\)}
\DP
\end{gather*}
\end{definition}

The first inference rule accounts for addresses whose object identifiers are
different, the second rule accounts for addresses whose references
are disjoint, and the third rule accounts for addresses that point to different
bytes of the same subobject.

\begin{definition}
\label{definition:addr_object_offset}
The \emph{reference bit-offset
\(\refobjectoffsetSym \Gamma : \rrefseg \to \nat\)} is defined as:
\begin{flalign*}
\refobjectoffset \Gamma {(\RArray i \tau n)} \defined{}&
	i \cdot \bitsizeof \Gamma \tau \\
\refobjectoffset \Gamma {(\RUnion i t q)} \defined{}& 0 \\
\refobjectoffset \Gamma {(\RStruct i t)} \defined{}&
	\bitoffsetof \Gamma {\vec \tau} i
	\quad
	\textnormal{where \(\elookup t \Gamma = \Some {\vec \tau}\)}
\end{flalign*}
Moreover, we let \(\addrobjectoffset \Gamma a \defined
\Sigma_i\, (\refobjectoffset \Gamma {(\addrref \Gamma a)_i}) +
\addrrefbyte \Gamma a \cdot \charbits\).
\end{definition}

Disjointness implies non-overlapping bit-offsets, but the reverse implication
does not always hold because references to different variants of unions are not
disjoint.
For example, given the declaration
\lstinline|union { struct { int x, y; } s; int z; } u|, the pointers
corresponding to \lstinline|&u.s.y| and \lstinline|&u.z| are not disjoint.

\begin{lemma}
If \(\addrtyped \Gamma \Delta {a_1} {\sigma_1}\),
\(\addrtyped \Gamma \Delta {a_2} {\sigma_2}\),
\(\addrstrict \Gamma {\{a_1,a_2\}}\),
\(\addrdisjoint \Gamma {a_1} {a_2}\), and
\(\addrindex {a_1} \ne \addrindex {a_2}\), then either:
\begin{enumerate}
\item \(\addrobjectoffset \Gamma {a_1} + \bitsizeof \Gamma {\sigma_1} \le
	\addrobjectoffset \Gamma {a_2}\), or
\item \(\addrobjectoffset \Gamma {a_2} + \bitsizeof \Gamma {\sigma_2} \le
	\addrobjectoffset \Gamma {a_1}\).
\end{enumerate}
\end{lemma}

\begin{lemma}[Well-typed addresses are properly aligned]
\label{lemma:align_of_addr_object_offset}
If \(\addrtyped \Gamma \Delta a \sigma\), then
\((\alignof \Gamma \sigma \cdot \charbits) \divides
\addrobjectoffset \Gamma a\).
\end{lemma}

\subsection{Representation of bits}
\label{section:bits}

As shown in Section~\ref{section:challenge:byte_level}, each object in \C{} can
be interpreted as an \lstinline|unsigned char| array called the \emph{object
representation}.
On actual computing architectures, the object representation consists of a
sequence of concrete bits (zeros and ones).
However, so as to accurately describe all undefined behaviors, we need a
special treatment for the object representations of pointers and indeterminate
memory in the formal memory model.
To that end, \CHtwoO{} represents the bits belonging to the object
representations of pointers and indeterminate memory symbolically.

\begin{definition}
\label{definition:bits}
\emph{Bits} are inductively defined as:
\begin{flalign*}
\syntax{bit}.
\end{flalign*}
\end{definition}

\begin{definition}
\label{definition:bit_typed}
The judgment \(\bitvalid \Gamma \Delta b\) describes that a bit \(b\) is
\emph{valid}.
It is inductively defined as:
\begin{gather*}
\AXC{\strut}
\UIC{\(\bitvalid \Gamma \Delta \BIndet\)}
\DP\quad
\AXC{\strut\(\beta \in \{ 0, 1 \}\)}
\UIC{\(\bitvalid \Gamma \Delta {\BBit \beta}\)}
\DP\quad
\AXC{\strut\(\ptrtyped \Gamma \Delta p {\ptype\sigma}\)}
\AXC{\(\frozen p\)}
\AXC{\(i < \bitsizeof \Gamma {(\TBase {\TPtr {\ptype\sigma}})}\)}
\TIC{\(\bitvalid \Gamma \Delta {\BPtrSeg {\PtrSeg p i}}\)}
\DP
\end{gather*}
\end{definition}

A bit is either a concrete bit \(\BBit 0\) or \(\BBit 1\), the
\(i\)th fragment bit \(\BPtrSeg {\PtrSeg p i}\) of a pointer \(p\), or the
indeterminate bit \(\BIndet\).
Integers are represented using concrete sequences of bits, and pointers as
sequences of fragment bits.
Assuming \(\bitsizeof {} {(\TBase {\TPtr {\TBase {\TInt {\IntType
\Signed \intrank}}}})} = 32\), a pointer \(p : \TBase {\TInt {\IntType
\Signed \intrank}}\) will be represented as the bit sequence
\(\BPtrSeg {\PtrSeg p 0} \ldots \BPtrSeg {\PtrSeg p {31}}\), and assuming
\(\bitsizeof {} {(\TBase {\TInt {\IntType \Signed \intrank}})} = 32\) on a
little-endian architecture, the integer \(33 : \TBase {\TInt
{\IntType \Signed \intrank}}\) will be represented as the bit sequence
\(\mathtt{1000010000000000}\).

The approach using a combination of symbolic and concrete bits is similar to
Leroy \etal~\cite{ler:app:bla:ste:12} and has the following advantages:

\begin{itemize}
\item Symbolic bit representations for pointers avoid the need to clutter the
	memory model with subtle, implementation defined, and run-time dependent
	operations to decode and encode pointers as concrete bit sequences.
\item We can precisely keep track of memory areas that are uninitialized.
	Since these memory areas consist of arbitrary concrete bit sequences on
	actual machines, most operations on them have undefined behavior.
\item While reasoning about program transformations one has to relate the
	memory states during the execution of the source program to those during the
	execution of the target program.
	Program transformations can, among other things, make more memory defined
	(that is, transform some indeterminate \(\BIndet\) bits into determinate
	bits)	and relabel the memory.
	Symbolic bit representations make it easy to deal with such transformations
	(see Section~\ref{section:refinements}).
\item It vastly decreases the amount of non-determinism, making it possible to
	evaluate the memory model as part of an executable
	semantics~\cite{kre:wie:15,kre:15:phd}.
\item The use of concrete bit representations for integers still gives a
	semantics to many low-level operations on integer representations.
\end{itemize}

A small difference with Leroy \etal~\cite{ler:app:bla:ste:12} is that the
granularity of our memory model is on the level of bits rather than bytes.
Currently we do not make explicit use of this granularity, but it allows us to
support bit-fields more faithfully with respect to the \Celeven{} standard in
future work.

Objects in our memory model are annotated with permissions.
We use permission annotations on the level of individual bits, rather than on
the level of bytes or entire objects, to obtain the most precise way of
permission handling.

\begin{definition}
\label{definition:pbit}
\emph{Permission annotated bits} are defined as:
\begin{flalign*}
	\syntax{pbit} = \perm \times \bit.
\end{flalign*}
\end{definition}

In the above definition, \(\taggedSym\) is the tagged separation algebra
that has been defined in Definition~\ref{definition:sep_tagged}.
We have spelled out its definition for brevity's sake.

\begin{definition}
The judgment \(\bitvalid \Gamma \Delta {\ppbit b}\) describes that a
permission annotated bit \(\ppbit b\) is \emph{valid}.
It is inductively defined as:
\begin{gather*}
\AXC{\(\bitvalid \Gamma \Delta b\)}
\AXC{\(\sepvalid \gamma\)}
\AXC{\(b = \BIndet\) in case \(\sepunmapped \gamma\)}
\TIC{\(\bitvalid \Gamma \Delta {\pair \gamma b}\)}
\DP
\end{gather*}
\end{definition}

\subsection{Representation of the memory}
\label{section:ctrees}

\emph{Memory trees} are abstract trees whose structure corresponds to the shape
of data types in \C.
They are used to describe individual objects (base values, arrays,
structs, and unions) in memory.
The memory is a forest of memory trees.

\begin{definition}
\label{definition:ctrees}
\emph{Memory trees} are inductively defined as:
\begin{flalign*}
\syntax{mtree}.
\end{flalign*}
\end{definition}

The structure of memory trees is close to the structure of types
(Definition~\ref{definition:types}) and thus reflects the expected semantics of
types: arrays are lists, structs are tuples, and unions are sums.
Let us consider the following example:

\begin{lstlisting}
struct S {
  union U { signed char x[2]; int y; } u; void *p;
} s = { .u = { .x = {33,34} }, .p = s.u.x + 2 };
\end{lstlisting}

The memory tree representing the object \lstinline|s| with object identifier
\(o_{\mathtt s}\) may be as follows (permissions are omitted for brevity's
sake, and integer encoding and padding are subject to implementation defined
behavior):
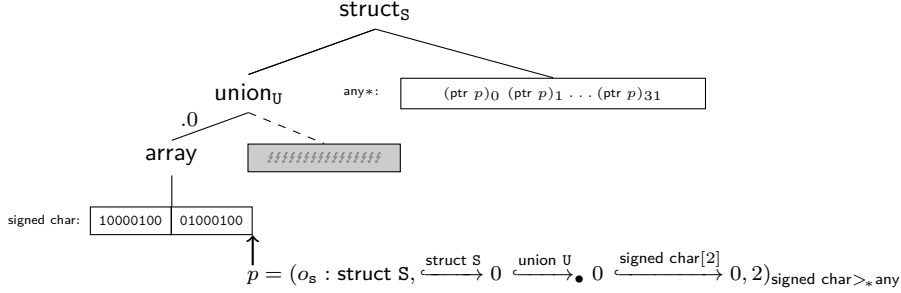
\begin{equation*}
\begin{tikzpicture}[
	node distance=-\pgflinewidth,
	every node/.style={shape=rectangle,minimum height=1.1em,minimum width=3em}
]
\node (S) {\(\MStructSym {\mathtt S}\)};
\node[below=1.8em of S,xshift=-5em,minimum width=12em] (U)
	{\(\MUnionSym {\mathtt U}\)};
\node[below=1.8em of U,anchor=east,minimum width=6em] (A)
	{\(\MArraySym {}\)};
\node[draw,below=1.8em of A,anchor=east,
	label=left:\tiny\(\TInt {\IntType \Signed \charrank}\):]	(x1)
	{\scalebox{.6}{\texttt{10000100}}};
\node[draw,right=of x1] (x2) {\scalebox{.6}{\texttt{01000100}}};
\node[draw,fill=gray,fill opacity=0.4,,text opacity=1,
	minimum width=6em,right=of A] (pad)	{\scalebox{.6}{\ntimes {16} \BIndet}};
\node[draw,anchor=west,minimum width=12em,right=of U,label=left:\tiny \(\TPtr \TAny\):]
	(p) {\tiny
	\(\BPtrSeg {\PtrSeg p 0}\ \BPtrSeg {\PtrSeg p 1}
	\dotsc \BPtrSeg {\PtrSeg p {31}}\)};
\draw (S.south) -- (U.north);
\draw (S.south) -- (U.north);
\draw (U.south) -- node[left,xshift=0.7em,yshift=0.3em] {\small .0} (A.north);
\draw (A.south) -- (x1.north east);
\draw[dashed] (U.south) -- (pad.north);
\draw (S.south) -- (p.north);
\draw[<-,thick] (x2.south east) --
	node [below,pos=1,right,yshift=-0.1em,xshift=-0.6em]
	{\small \(p = \Ptr {\Addr {o_{\mathtt s}}
	{\TStruct {\mathtt S}}
	{\RStruct 0 {\mathtt S}\,
		\RUnion 0 {\mathtt U} \RUnfrozen\,
		\RArray 0 {\TBase {\TInt {\IntType \Signed \charrank}}} 2}
	2 {\TBase {\TInt {\IntType \Signed \charrank}}}	\TAny}\)} +(0,-0.4);
\end{tikzpicture}
\end{equation*}

The representation of unions requires some explanation.
We considered two kinds of memory trees for unions:

\begin{itemize}
\item The memory tree \(\MUnion t i w {\vec {\ppbit b}}\) represents a union
	whose variant is \(i\).
	Unions of variant \(i\) can only be accessed through a pointer to variant
	\(i\).
	This is essential for effective types.
	The list \(\vec {\ppbit b}\) represents the padding after the element \(w\).
\item The memory tree \(\MUnionAll t {\vec {\ppbit b}}\) represents a union
	whose variant is yet unspecified.
	Whenever the union is accessed through a pointer to variant \(i\), the
	list \(\vec {\ppbit b}\) will be interpreted as a memory tree of the
	type belonging to the \(i\)th variant.
\end{itemize}

The reason that we consider unions \(\MUnionAll t {\vec {\ppbit b}}\) with
unspecific variant at all is that in some cases the variant cannot be known.
Unions that have not been initialized do not have a variant yet.
Also, when a union object is constructed byte-wise through its object
representation, the variant cannot be known.
	
Although unions are tagged in the formal memory, actual compilers implement
untagged unions.
Information about variants should thus be internal to the formal
memory model.
In Section~\ref{section:refinements} we prove that this is indeed the case.

The additional structure of memory trees, namely type annotations, variants
of unions, and structured information about padding, can be erased by
flattening.
Flattening just appends the bytes on the leaves of the tree.

\begin{definition}
\label{definition:ctree_flatten}
The \emph{flatten operation \(\ctreeflattenSym : \mtree \to \lst \pbit\)} is
defined as:
\begin{gather*}
\ctreeflatten{\MBase {\btype \tau} {\vec {\ppbit b}}} \defined
	\vec {\ppbit b} \qquad
\ctreeflatten{\MArray \tau {\vec w}} \defined
	\ctreeflatten{w_0} \dotsc \ctreeflatten{w_{\length{\vec w} - 1}} \\
\ctreeflatten{\MStruct t {\vv{w\vec {\ppbit b}}}} \defined
	(\ctreeflatten{w_0}\,\vec {\ppbit b}_0) \dotsc
	(\ctreeflatten{w_{\length{\vec w} - 1}}\,
	 \vec {\ppbit b}_{\length{\vec w} - 1}) \qquad
\ctreeflatten{\MUnion t j w {\vec {\ppbit b}}} \defined
	\ctreeflatten w\, \vec {\ppbit b} \qquad
\ctreeflatten{\MUnionAll t {\vec {\ppbit b}}} \defined \vec {\ppbit b}
\end{gather*}
\end{definition}

The flattened version of the memory tree representing the object \lstinline|s|
in the previous example is as follows:
\[
	10000100\ 01000100\ \ntimes 8 \BIndet\ \ntimes 8 \BIndet
	\ \BPtrSeg {\PtrSeg p 0}\ \BPtrSeg {\PtrSeg p 1}
		\dotsc \BPtrSeg {\PtrSeg p {31}}
\]

\begin{definition}
\label{definition:ctree_typed}
The judgment \(\ctreetyped \Gamma \Delta w \tau\) describes that the
\emph{memory tree \(w\) has type \(\tau\)}.
It is inductively defined as:
\begin{gather*}
\AXC{\(\basetypevalid \Gamma {\btype\tau}\)}
\AXC{\(\bitvalid \Gamma \Delta {\vec{\ppbit b}}\)}
\AXC{\(\length {\vec{\ppbit b}} = \bitsizeof \Gamma {\TBase {\btype\tau}}\)}
\TIC{\(\ctreetyped \Gamma \Delta {\MBase {\btype\tau} {\vec{\ppbit b}}}
	{\TBase {\btype\tau}}\)}
\DP\qquad
\AXC{\(\ctreetyped \Gamma \Delta {\vec w} \tau\)}
\AXC{\(\length {\vec w} = n \ne 0\)}
\BIC{\(\ctreetyped \Gamma \Delta {\MArray \tau {\vec w}}
	{\TArray \tau n}\)}
\DP\\[0.4em]
\AXC{\(\begin{array}{c}
\elookup t \Gamma = \vec\tau \qquad
\ctreetyped \Gamma \Delta {\vec w} {\vec\tau} \\
\forall i \wsdot (\bitvalid \Gamma \Delta {{\vec {\ppbit b}}_i}\qquad
	{\vec {\ppbit b}}_i\textnormal{ all }\BIndet \qquad
	\length {{\vec {\ppbit b}}_i} =
	(\fieldbitsizes \Gamma {\vec \tau})_i - \bitsizeof \Gamma {\tau_i})
\end{array}\)}
\UIC{\(\ctreetyped \Gamma \Delta
	{\MStruct t {\vv{w\vec {\ppbit b}}}} {\TStruct t}\)}
\DP\\[0.4em]
\AXC{\(\begin{array}{c}
\elookup t \Gamma = \vec\tau \qquad
i < \length {\vec\tau} \qquad
\ctreetyped \Gamma \Delta w {\tau_i} \qquad
\bitvalid \Gamma \Delta {\vec{\ppbit b}} \qquad
\vec{\ppbit b}\textnormal{ all }\BIndet \\
\bitsizeof \Gamma {(\TUnion t)} = \bitsizeof \Gamma {\tau_i} +
	\length {\vec{\ppbit b}}\qquad
\neg\sepunmapped {(\ctreeflatten w\,\vec{\ppbit b})}
\end{array}\)}
\UIC{\(\ctreetyped \Gamma \Delta
	{\MUnion t i w {\vec {\ppbit b}}} {\TUnion t}\)}
\DP\\[0.4em]
\AXC{\(\elookup t \Gamma = \Some {\vec\tau}\)}
\AXC{\(\bitvalid \Gamma \Delta {\vec{\ppbit b}}\)}
\AXC{\(\length {\vec{\ppbit b}} = \bitsizeof \Gamma {(\TUnion t)}\)}
\TIC{\(\ctreetyped \Gamma \Delta {\MUnionAll t {\vec{\ppbit b}}}
	{\TUnion t}\)}
\DP
\end{gather*}
\end{definition}

Although padding bits should be kept indeterminate (see
Section~\ref{section:challenge:byte_level}),
padding bits are explicitly stored in memory trees for uniformity's sake.
The typing judgment ensures that the value of each padding bit is
\(\BIndet\) and that the padding thus only have a permission.
Storing a value in padding is a no-op (see Definition~\ref{definition:cmap_alter}).

The side-condition \(\neg\sepunmapped {(\ctreeflatten w\,\vec{\ppbit b})}\) in the
typing rule for a union \(\MUnion t i w {\vec {\ppbit b}}\) of a specified
variant ensures canonicity.
Unions whose permissions are unmapped cannot be accessed and should therefore be
in an unspecified variant.
This condition is essential for the separation algebra structure, see
Section~\ref{section:memory_separation}. 

\begin{definition}
\label{definition:memory}
\emph{Memories} are defined as:
\begin{flalign*}
\syntax{mem}.
\end{flalign*}
\end{definition}

Each object \(\pair w \mu\) in memory is annotated with a Boolean \(\mu\)
to describe whether it has been allocated using
\lstinline|malloc| (in case \(\mu = \true\)) or as a block scope local, static,
or global variable (in case \(\mu = \false\)).
The types of deallocated objects are kept to ensure that dangling pointers
(which may remain to exist in memory, but cannot be used) have a unique type.

\begin{definition}
\label{definition:cmap_valid}
The judgment \(\memvalid \Gamma \Delta m\) describes that \emph{the memory \(m\)
is valid}.
It is defined as the conjunction of:
\begin{enumerate}
\item For each \(o\) and \(\tau\) with \(\mlookup o m = \Some \tau\) we have:
	\\
	\begin{inparaenum}
	\item \(\indextyped \Delta o \tau\),
	\item \(\notindexalive \Delta o\), and
	\item \(\typevalid \Gamma \tau\).
	\end{inparaenum}
\item	For each \(o\), \(w\) and \(\mu\) with
	\(\mlookup o m = \Some {\pair w \mu}\) we have:
	\\
	\begin{inparaenum}
	\item \(\indextyped \Delta o \tau\),
	\item \(\indexalive \Delta o\),
	\item \(\ctreetyped \Gamma \Delta w \tau\), and
	\item not \(\ctreeflatten w\) all \(\pair \sepempty \BIndet\).
	\end{inparaenum}
\end{enumerate}
The judgment \(\indexalive \Delta o\) on object identifiers is defined in
Definition~\ref{definition:index_alive}.
\end{definition}

\begin{definition}
The \emph{minimal memory typing environment} \(\memenvof m \in \memenv\) of
a memory \(m\) is defined as:
\begin{gather*}
\memenvof m \defined \lambda o \wsdot \begin{cases}
	\pair \tau \true & \textnormal{if \(\mlookup o m = \Some \tau\)} \\
	\pair {\typeof w} \false &
	\textnormal{if \(\mlookup o m = \Some {\pair w \mu}\)}
\end{cases}
\end{gather*}
\end{definition}

\begin{notation}
We let \(\memvalidShort \Gamma m\) denote \(\memvalid \Gamma {\memenvof m} m\).
\end{notation}

Many of the conditions of the judgment \(\memvalid \Gamma \Delta m\)
ensure that the types of \(m\) match up with the types in the
memory environment \(\Delta\) (see Definition~\ref{definition:index_typed}).
One may of course wonder why do we not define the judgment
\(\memvalidShort \Gamma m\) directly, and even consider typing of a memory
in an arbitrary memory environment.
Consider:

\begin{lstlisting}
int x = 10, *p = &x;
\end{lstlisting}

Using an assertion of separation logic we can describe the memory
induced by the above program as \(\axSingletonShort {\mathtt x} {} 10 \axSep
  \axSingletonShort {\mathtt p} {} {\mathtt{\&x}}\).
The separation conjunction \(\axSep\) describes that the memory can be
subdivided into two parts, a part for \lstinline|x| and another part for 
\lstinline|p|.
When considering \(\axSingletonShort {\mathtt p} {} {\mathtt{\&x}}\) in
isolation, which is common in separation logic, we have a pointer that refers
outside the part itself.
This isolated part is thus not typeable by \(\memvalidShort \Gamma m\),
but it is typeable in the context of a the memory environment corresponding
to the whole memory.
See also Lemma~\ref{lemma:cmap_union_valid}.

In the remaining part of this section we will define various auxiliary
operations that will be used to define the memory operations in
Section~\ref{section:memory_operations}.
We give a summary of the most important auxiliary operations:
\begin{align*}
\ctreenewSym \Gamma \gamma :{}& \type \to \mtree
	&&\textnormal{ for }\quad \gamma : \perm \\
\cmaplookupSym \Gamma :{}& \mem \to \addr \to \option\mtree && \\
\cmapalterSym \Gamma f :{}& \mem \to \addr \to \mem
	&&\textnormal{ for }\quad f : \mtree \to \mtree
\end{align*}

Intuitively these are just basic tree operations, but unions make their
actual definitions more complicated.
The indeterminate memory tree \(\ctreenew \Gamma \gamma \tau\) consists of
indeterminate bits with permission \(\gamma\),
the lookup operation \(\cmaplookup \Gamma a m\) yields the memory tree
at address \(a\) in \(m\), and the alter operation
\(\cmapalter \Gamma f a m\) applies the function \(f\) to the memory tree at
address \(a\) in \(m\).

The main delicacy of all of these operations is that we sometimes have to
interpret bits as memory trees, or reinterpret memory trees as memory trees of
a different type.
Most notably, reinterpretation is needed when type-punning is performed:

\begin{lstlisting}
union int_or_short { int x; short y; } u = { .x = 3 };
short z = u.y;
\end{lstlisting}

This code will reinterpret the bit representation of a memory tree representing
an \lstinline|int| value \lstinline|3| as a memory tree of type
\lstinline|short|.
Likewise:

\begin{lstlisting}
union int_or_short { int x; short y; } u;
((unsigned char*)&u)[0] = 3;
((unsigned char*)&u)[1] = 0;
short z = u.y;
\end{lstlisting}

Here, we poke some bytes into the object representation of \lstinline|u|,
and interpret these as a memory tree of type \lstinline|short|.

We have defined the flatten operation \(\ctreeflatten w\) that takes a memory
tree \(w\) and yields its bit representation already in
Definition~\ref{definition:ctree_flatten}.
We now define the operation which goes in opposite direction, called the
\emph{unflatten operation}.

\begin{definition}
\label{definition:ctree_unflatten}
The \emph{unflatten operation}
\(\ctreeunflattenSym \Gamma \tau : \lst \pbit \to \mtree\) is defined as:
\begin{samepage}
\begin{flalign*}
\ctreeunflatten \Gamma {\TBase {\btype\tau}} {\vec {\ppbit b}} \defined{}&
	\MBase {\btype\tau} {\vec {\ppbit b}} \\
\ctreeunflatten \Gamma {\TArray \tau n} {\vec {\ppbit b}} \defined{}&
	\MArray \tau {(
	\ctreeunflatten \Gamma \tau {\sublist 0 s {\vec{\ppbit b}}} \dotsc 
	\ctreeunflatten \Gamma \tau
		{\sublist {(n-1)s} {ns} {\vec{\ppbit b}}})}
	\ \textnormal{where \(s \defined \bitsizeof \Gamma \tau\)} \\
\ctreeunflatten \Gamma {\TStruct t} {\vec {\ppbit b}} \defined{}&
	\MStruct t {\begin{pmatrix*}[l]
	\ctreeunflatten \Gamma {\tau_0} {\sublist 0 {s_0} {\vec{\ppbit b}}}\,
	\pbitindetify{\sublist {s_0} {z_0} {\vec{\ppbit b}}}
	\\ \dotsc \\
	\ctreeunflatten \Gamma {\tau_{n-1}}
		{\sublist {z_{n-1}} {z_{n-1} + s_{n-1}} {\vec{\ppbit b}}}\,
	\pbitindetify{\sublist {z_{n-1} + s_{n-1}} {z_n} {\vec{\ppbit b}}} \\
	\end{pmatrix*}} \\
	& \textnormal{where \(\elookup t \Gamma = \Some {\vec \tau}\),
		\(n \defined \length{\vec \tau}\),
		\(s_i \defined \bitsizeof \Gamma {\tau_i}\) and
		\(z_i \defined \bitoffsetof \Gamma {\vec\tau} i\)} \\
\ctreeunflatten \Gamma {\TUnion t} {\vec {\ppbit b}} \defined{}&
	\MUnionAll t {\vec {\ppbit b}}
\end{flalign*}
Here, the operation \(\pbitindetifySym : \pbit \to \pbit\) is defined
as \(\pbitindetify {\pair x b} \defined \pair x \BIndet\).
\end{samepage}
\end{definition}

Now, the need for \(\MUnionAll t {\vec{\ppbit b}}\) 
memory trees becomes clear.
While unflattening a bit sequence as a union, there is no way of knowing
which variant of the union the bits constitute.
The operations \(\ctreeflattenSym\) and
\(\ctreeunflatten \Gamma \tau {\_}\) are neither left nor right inverses:

\begin{itemize}
\item We do not have \(\ctreeunflatten \Gamma \tau {\ctreeflatten w} = w\) for
	each \(w\) with \(\ctreetyped \Gamma \Delta w \tau\).
	Variants of unions are destroyed by flattening \(w\).
\item We do not have
	\(\ctreeflatten {\ctreeunflatten \Gamma \tau {\vec{\ppbit b}}} =
	\vec {\ppbit b}\) for each \(\vec {\ppbit b}\) with
	\(\length {\vec {\ppbit b}} = \bitsizeof \Gamma \tau\) either.
	Padding bits become indeterminate due to \(\pbitindetifySym\) by
	unflattening.
\end{itemize}

In Section~\ref{section:refinements} we prove weaker variants of these
cancellation properties that are sufficient for proofs about program
transformations.

\begin{definition}
\label{definition:ctree_new}
Given a permission \(\gamma \in \perm\), the operation
\(\ctreenewSym \Gamma \gamma : \type \to \mtree\) that yields the
indeterminate memory tree is defined as:
\[
\ctreenew \Gamma \gamma \tau \defined \ctreeunflatten \Gamma \tau
	{\replicate {\bitsizeof \Gamma \tau} {\pair \gamma \BIndet}}.
\]
\end{definition}

The memory tree \(\ctreenew \Gamma \gamma \tau\) that consists of
indeterminate bits with permission \(\gamma\) is used for objects with
indeterminate value.
We have defined \(\ctreenew \Gamma \gamma \tau\) in terms of the unflattening
operation for simplicity's sake.
This definition enjoys desirable structural
properties such as \(\ctreenew \Gamma \gamma {(\TArray \tau n)} =
\replicate n {(\ctreenew \Gamma \gamma \tau)}\).

We will now define the lookup operation \(\cmaplookup \Gamma a m\) that yields
the subtree at address \(a\) in the memory \(m\).
The lookup function is partial, it will fail in case \(a\) is end-of-array or
violates effective types.
We first define the counterpart of lookup on memory trees and then lift it
to memories.

\begin{definition}
\label{definition:ctree_lookup}
The \emph{lookup operation on memory trees}
\(\ctreelookupSym \Gamma : \mtree \to \rref \to \option\mtree\) is defined as:
\begin{flalign*}
\ctreelookup \Gamma \nil w \defined{}& w \\
\ctreelookup \Gamma {\app {(\RArray i \tau n)} {\vec r}} {(\MArray \tau {\vec w})}
	\defined{}& \ctreelookup \Gamma {\vec r} {w_i} \\
\ctreelookup \Gamma {\app {(\RStruct i t)} {\vec r}}
	{(\MStruct t {\vv{w\vec {\ppbit b}}})}
	\defined{}& \ctreelookup \Gamma {\vec r} {w_i} \\
\ctreelookup \Gamma {\app {(\RUnion i t q)} {\vec r}}
	{(\MUnion t j w {\vec {\ppbit b}})}
	\defined{}& \!\!\begin{cases}
	\ctreelookup \Gamma {\vec r} w & \textnormal{if \(i = j\)} \\
	\ctreelookup \Gamma {\vec r} {\ctreeunflatten \Gamma {\tau_i}
		{\sublist 0 s
		{(\ctreeflatten w\,\vec {\ppbit b})}}}
	& \textnormal{if \(i\neq j\), \(q = \RUnfrozen\),
		\(\sepunshared {(\ctreeflatten w\,\vec {\ppbit b})}\)} \\
	\None	& \textnormal{if \(i \neq j\), \(q = \RFrozen\)}
	\end{cases} \\
	& \textnormal{where \(\elookup t \Gamma = \Some {\vec \tau}\) and
		\(s = {\bitsizeof \Gamma {\tau_i}}\)} \\
\ctreelookup \Gamma {\app {(\RUnion i t q)} {\vec r}}
	{(\MUnionAll t {\vec {\ppbit b}})}
	\defined{}&
	\ctreelookup \Gamma {\vec r} {\ctreeunflatten \Gamma {\tau_i}
		{\sublist 0 {\bitsizeof \Gamma {\tau_i}} {\vec {\ppbit b}}}}
	\quad\textnormal{if \(\elookup t \Gamma = \Some {\vec \tau}\),
		 \(\sepunshared {\vec {\ppbit b}}\)}
\end{flalign*}
\end{definition}

The lookup operation uses the annotations \(q \in \{\RFrozen,\RUnfrozen\}\) on
\(\RUnion i s q\) to give a formal semantics to the
\emph{strict-aliasing restrictions}~\cite[6.5.2.3]{iso:12}.

\begin{itemize}
\item The annotation \(q = \RUnfrozen\) allows a union to be accessed
	via a reference whose variant is unequal to the current one.
	This is called type-punning.
\item The annotation \(q = \RFrozen\) allows a union to be accessed only
	via a reference whose variant is equal to the current one.
	This means, it rules out type-punning.
\end{itemize}

Failure of type-punning is captured by partiality of the lookup operation.
The behavior of type-punning of \(\MUnion t j w {\vec {\ppbit b}}\) via a
reference to variant \(i\) is described by the conversion
\(\ctreeunflatten \Gamma {\tau_i} {\sublist 0
{\bitsizeof \Gamma {\tau_i}} {(\ctreeflatten w\,\vec {\ppbit b})}}\).
The memory tree \(w\) is converted into bits and reinterpreted
as a memory tree of type \(\tau_i\).

\begin{definition}
\label{definition:cmap_lookup}
The \emph{lookup operation on memories}
\(\cmaplookupSym \Gamma : \mem \to \addr \to \option\mtree\) is
defined as:
\begin{flalign*}
\cmaplookup \Gamma a m \defined{}&
	\begin{cases}
	\ctreeunflatten \Gamma {\IntType \Unsigned \charrank}
		{\sublist i j {(\ctreeflatten{\ctreelookup \Gamma
			{\addrref \Gamma a} w})}}
		&\textnormal{if \(a\) is a byte address} \\
	\ctreelookup \Gamma {\addrref \Gamma a} w
		&\textnormal{if \(a\) is not a byte address}
	\end{cases}
\end{flalign*}
provided that \(\mlookup {(\addrindex a)} m = \Some {\pair w \mu}\).
In omitted cases the result is \(\None\).
In this definition we let \(i \defined \addrrefbyte \Gamma a \cdot \charbits\)
and \(j \defined (\addrrefbyte \Gamma a+1) \cdot \charbits\).
\end{definition}

We have to take special care of addresses that refer to individual bytes rather
than whole objects.
Consider:

\begin{lstlisting}
struct S { int x; int y; } s = { .x = 1, .y = 2 };
unsigned char z = ((unsigned char*)&s)[0];
\end{lstlisting}

In this code, we obtain the first byte \lstinline|((unsigned char*)&s)[0]| of
the struct \lstinline|s|.
This is formalized by flattening the entire memory tree of the struct
\lstinline|s|, and selecting the appropriate byte.

The \Celeven{} standard's description of effective types~\cite[6.5p6-7]{iso:12}
states that an access (which is either a
read or store) affects the effective type of the accessed object.
This means that although reading from memory does not affect the memory
contents, it may still affect the effective types.
Let us consider an example where it is indeed the case that effective types are
affected by a read:

\begin{lstlisting}
short g(int *p, short *q) {
  short z = *q; *p = 10; return z;
}
int main() {
  union int_or_short { int x; short y; } u;
  // initialize u with zeros, the variant of u remains unspecified
  for (size_t i = 0; i < sizeof(u); i++) ((unsigned char*)&u)[i] = 0;
  return g(&u.x, &u.y);
}
\end{lstlisting}

In this code, the variant of the union \lstinline|u| is initially unspecified.
The read \lstinline|*q| in \lstinline|g| \emph{forces} its variant to
\lstinline|x|, making the assignment \lstinline|*p| to variant \lstinline|y|
undefined.
Note that it is important that we also assign undefined behavior to this
example, a compiler may assume \lstinline|p| and \lstinline|q| to not alias
regardless of how \lstinline|g| is called.

We factor these side-effects out using a function
\(\memforceSym \Gamma : \addr \to \mem \to \mem\) that updates the effective
types (that is the variants of unions) after a successful lookup.
The \(\memforceSym \Gamma\) function, as defined in
Definition~\ref{section:memory_operations}, can be described in terms of the
alter operation \(\cmapalter \Gamma f a m\) that applies the function \(f :
\mtree \to \mtree\) to the object at address \(a\) in the memory \(m\) and
update variants of unions accordingly to \(a\).
To define \(\memforceSym \Gamma\) we let \(f\) be the identify.

\begin{definition}
Given a function \(f : \mtree \to \mtree\), the \emph{alter operation on
memory trees} \(\ctreealterSym \Gamma f : \mtree \to \rref \to \mtree\) is
defined as:
\begin{flalign*}
\ctreealter \Gamma f \nil w \defined{}& f\,w \\
\ctreealter \Gamma f {\app {(\RArray i \tau n)} {\vec r}}
	{(\MArray \tau {\vec w})}
	\defined{}&
	\MArray \tau {(\minsert i {\ctreealter \Gamma f {\vec r} {w_i}} {\vec w})} \\
\ctreealter \Gamma f {\app {(\RStruct i t)} {\vec r}}
	{(\MStruct t {\vv{w\vec {\ppbit b}}})}
	\defined{}&
	\MStruct t {(\minsert i {\ctreealter \Gamma f {\vec r} {w_i} \vec {\ppbit b}_i}
		{(\vv{w\vec {\ppbit b}})})} \\
\ctreealter \Gamma f {\app {(\RUnion j t q)} {\vec r}}
	{(\MUnion t i w {\vec {\ppbit b}})}
	\defined{}&
	\!\!\begin{cases}
	\MUnion t i {\ctreealter \Gamma f {\vec r} w} {\vec {\ppbit b}}
	& \textnormal{if \(i = j\)} \\
	\MUnion t i {\ctreealter \Gamma f {\vec r}
		{(\ctreeunflatten \Gamma {\tau_i}
			{\sublist 0 s {(\ctreeflatten w\vec {\ppbit b})}})}}
		{\sublist s z {\pbitindetify {(\ctreeflatten w\vec {\ppbit b})}}}
	& \textnormal{if \(i \neq j\)}
	\end{cases}\\
\ctreealter \Gamma f {\app {(\RUnion i t q)} {\vec r}}
	{(\MUnionAll t {\vec {\ppbit b}})}
	\defined{}&
	\MUnion t i {\ctreealter \Gamma f {\vec r}
		{(\ctreeunflatten \Gamma {\tau_i}
			{\sublist 0 s {\vec {\ppbit b}}})}}
		{\sublist s z {\pbitindetify {\vec {\ppbit b}}}}
\end{flalign*}
\noindent
In the last two cases we have \(\elookup t \Gamma = \vec\tau\),
\(s \defined \bitsizeof \Gamma {\tau_i}\) and
\(z \defined \bitsizeof \Gamma {(\TUnion t)}\).
The result of \(\ctreealter \Gamma f {\vec r} w\) is only well-defined in case
\(\ctreelookup \Gamma {\vec r} w \neq \None\).
\end{definition}

\begin{definition}
\label{definition:cmap_alter}
Given a function \(f : \mtree \to \mtree\), the \emph{alter operation
on memories} \(\cmapalterSym \Gamma f : \mem \to \addr \to \mem\) is defined as:
\begin{flalign*}
\cmapalter \Gamma f a m \defined{}&
	\begin{cases}
	\minsert {(\addrindex a)}
		{\pair {\ctreealter \Gamma {\overline f} {\addrref \Gamma a} w} \mu} m
		&\textnormal{if \(a\) is a byte address} \\
	\minsert {(\addrindex a)}
		{\pair {\ctreealter \Gamma f {\addrref \Gamma a} w} \mu} m
		&\textnormal{if \(a\) is not a byte address}
	\end{cases}
\end{flalign*}
provided that \(\mlookup {(\addrindex a)} m = \Some {\pair w \mu}\).
In this definition we let:
\[\overline f\,w \defined \ctreeunflatten \Gamma {\typeof w}
	{\sublist 0 i {\ctreeflatten w}\;
	{\ctreeflatten {f\;{\ctreeunflatten \Gamma {\IntType \Unsigned \charrank}
		{\sublist i j {\ctreeflatten w}}}}}\;
	{\sublist j {\bitsizeof \Gamma {(\typeof w)}}
		{\ctreeflatten w}}}
\] where
\(i \defined \addrrefbyte \Gamma a \cdot \charbits\) and
\(j \defined (\addrrefbyte \Gamma a+1) \cdot \charbits\).
\end{definition}

The lookup and alter operation enjoy various properties; they preserve
typing and satisfy laws about their interaction.
We list some for illustration.

\begin{lemma}[Alter commutes]
\label{lemma:cmap_alter_commute}
If \(\memvalid \Gamma \Delta m\), \(\addrdisjoint \Gamma {a_1} {a_2}\)
with:
\begin{itemize}
\item \(\addrtyped \Gamma \Delta {a_1} {\tau_1}\),
	\(\cmaplookup \Gamma {a_1} m = \Some {w_1}\), and
	\(\ctreetyped \Gamma \Delta {f_1\,w_1} {\tau_1}\), and
\item \(\addrtyped \Gamma \Delta {a_2} {\tau_2}\),
	\(\cmaplookup \Gamma {a_2} m = \Some {w_2}\), and
	\(\ctreetyped \Gamma \Delta {f_2\,w_2} {\tau_2}\),
\end{itemize}
then we have:
\[
	\cmapalter \Gamma {f_1} {a_1} {\cmapalter \Gamma {f_2} {a_2} m} =
	\cmapalter \Gamma {f_2} {a_2} {\cmapalter \Gamma {f_1} {a_1} m}.
\]
\end{lemma}

\begin{lemma}
\label{lemma:cmap_lookup_alter}
If \(\memvalid \Gamma \Delta m\),
\(\cmaplookup \Gamma a m = \Some w\), 
and \(a\) is not a byte address, then:
\[\cmaplookup \Gamma a {(\cmapalter \Gamma f a m)} = \Some {f\,w}.\]
\end{lemma}

A variant of Lemma~\ref{lemma:cmap_lookup_alter} for byte addresses is more
subtle because a byte address can be used to modify padding.
Since modifications of padding are masked, a successive lookup may yield
a memory tree with more indeterminate bits.
In Section~\ref{section:refinements} we present an alternative lemma that
covers this situation.

We conclude this section with a useful helper function that \emph{zips} a
memory tree and a list.
It is used in for example Definitions~\ref{definition:memory_operations}
and~\ref{definition:ctree_union}.

\begin{definition}
\label{definition:ctree_merge}
Given a function \(f : \pbit \to B \to \pbit\), the operation that zips the
leaves \(\ctreemergeSym f : \mtree \to \lst B \to \mtree\) is defined as:
\begin{flalign*}
\ctreemerge f {(\MBase {\btype \tau} {\vec {\ppbit b}})} {\vec y} \defined{}&
	\MBase {\btype \tau} {(f\,\vec {\ppbit b}\;\vec y)} \\
\ctreemerge f {(\MArray \tau {\vec w})} {\vec y} \defined{}&
	\MArray \tau {(\ctreemerge f {w_0}
		{\sublist 0 {s_0} {\vec y}}
	\dotsc \ctreemerge f {w_{n-1}} {\sublist {s_{n-1}} {s_n} {\vec y}})} \\
	& \textnormal{where \(n \defined \length{\vec w}\) and
	\(s_i \defined \Sigma_{j < i} \length {\ctreeflatten {w_j}}\)} \\
\ctreemerge f {(\MStruct t {\vv{w\vec {\ppbit b}}})} {\vec y} \defined{}&
	\MStruct t {\begin{pmatrix*}[l]
		\ctreemerge f {w_0} {\sublist 0 {s_0} {\vec y}}\
		f\,\vec {\ppbit b}_0\;\sublist{s_0} {z_0} {\vec y}
		\\ \dotsc \\
		\ctreemerge f {w_{n - 1}}
			{\sublist {z_{n-1}} {z_{n-1} + s_{n-1}} {\vec y}} \
		f\,\vec {\ppbit b}_{n-1}\,\sublist{z_{n-1} + s_{n-1}} {z_n} {\vec y}
	\end{pmatrix*}} \\
	& \textnormal{where \(n \defined \length{\vec w}\),
	\(s_i \defined \length {\ctreeflatten {w_i}}\), and
	\(z_i \defined \Sigma_{j < i} \length {\ctreeflatten {w_i}\,\vec {\ppbit b}_i}\)} \\
\ctreemerge f {(\MUnion t i w {\vec {\ppbit b}})} {\vec y} \defined{}&
	\MUnion t i
		{\ctreemerge f w {\sublist 0 {\length {\ctreeflatten w}} {\vec y}}}
		{f\,\vec {\ppbit b}\;\sublist {\length {\ctreeflatten w}}
			{\length {\ctreeflatten w \,\vec {\ppbit b}}} {\vec y}} \\
\ctreemerge f {(\MUnionAll t {\vec {\ppbit b}})} {\vec y} \defined{}&
	\MUnionAll t {(f\,\vec {\ppbit b}\;\vec y)}
\end{flalign*}
\end{definition}

\subsection{Representation of values}
\label{section:values}

Memory trees (Definition~\ref{definition:ctrees}) are still rather low-level
and expose permissions and implementation specific properties such as bit
representations.
In this section we define \emph{abstract values}, which are like memory trees
but have mathematical integers and pointers instead of bit representations as
leaves.
Abstract values are used in the external interface of the memory model.

\begin{definition}
\label{definition:base_values}
\emph{Base values} are inductively defined as:
\begin{flalign*}
\syntax{baseval}.
\end{flalign*}
\end{definition}

While performing byte-wise operations (for example, byte-wise copying a struct
containing pointer values), abstraction is broken, and pointer fragment bits have to reside outside
of memory.
The value \(\VByte {\vec b}\) is used for this purpose.

\begin{definition}
\label{definition:base_val_typed}
The judgment \(\basevaltyped \Gamma \Delta {\bval v} {\btype \tau}\)
describes that \emph{the base value \(\bval v\) has base type \(\btype \tau\)}.
It is inductively defined as:
\begin{gather*}
\renewcommand{\ScoreOverhang}{0.1em}
\AXC{\strut\(\basetypevalid \Gamma {\btype \tau}\)}
\AXC{\(\btype\tau \ne \TVoid\)}
\BIC{\(\basevaltyped \Gamma \Delta {\VIndet {\btype\tau}} {\btype\tau}\)}
\DP\qquad
\AXC{\strut}
\UIC{\(\basevaltyped \Gamma \Delta \VVoid \TVoid\)}
\DP\qquad
\AXC{\strut\(\inttyped x {\itype \tau}\)}
\UIC{\(\basevaltyped \Gamma \Delta {\VInt {\itype\tau} x}
	{\TInt {\itype\tau}}\)}
\DP
\\[0.5em]
\AXC{\strut\(\ptrtyped \Gamma \Delta p {\ptype\sigma}\)}
\UIC{\(\basevaltyped \Gamma \Delta {\VPtr p} {\TPtr {\ptype\sigma}}\)}
\DP\quad
\renewcommand{\defaultHypSeparation}{\hskip.3em}
\AXC{\strut\(\bitvalid \Gamma \Delta {\vec b}\)}
\AXC{\(\length {\vec b} = \charbits\)}
\AXC{Not \(\vec b\) all in \(\{\BBit 0, \BBit 1\}\)}
\AXC{Not \(\vec b\) all \(\BIndet\)}
\QIC{\(\basevaltyped \Gamma \Delta {\VByte {\vec b}}
	{\TInt {\IntType \Unsigned \charrank}}\)}
\DP
\end{gather*}
\end{definition}

The side-conditions of the typing rule for \(\VByte {\vec b}\) ensure
canonicity of representations of base values.
It ensures that the construct \(\VByte {\vec b}\) is only used if \(\vec b\)
cannot be represented as
an integer \(\VInt {\IntType \Unsigned \charrank} x\) 
or \(\VIndet {(\TInt {\IntType \Unsigned \charrank})}\).

In Definition~\ref{definition:values} we define abstract values by
extending base values with constructs for arrays, structs and unions.
In order to define the operations to look up and store values in memory, we
define conversion operations between abstract values and memory trees.
Recall that the leaves of memory trees, which represent base values, are just
sequences of bits.
We therefore first define operations that convert base values to and from bits.
These operations are called flatten and unflatten.

\begin{definition}
\label{definition:base_val_flatten}
The \emph{flatten operation
\(\basevalflattenSym \Gamma : \baseval \to \lst\bit\)} is
defined as:
\begin{flalign*}
\basevalflatten \Gamma {\VIndet {\btype \tau}} \defined{}&
	\replicate {\bitsizeof \Gamma {\btype \tau}} \BIndet \\
\basevalflatten \Gamma \VVoid \defined{}&
	\replicate {\bitsizeof \Gamma \TVoid} \BIndet \\
\basevalflatten \Gamma {\VInt {\itype \tau} x} \defined{}&
	\inttobits {\itype\tau} x \\
\basevalflatten \Gamma {\VPtr p} \defined{}&
	\BPtrSeg {\PtrSeg {\freeze p} 0} \ldots
	\BPtrSeg {\PtrSeg {\freeze p}
		{\bitsizeof \Gamma {(\TPtr {\typeof p}}) - 1}} \\
\basevalflatten \Gamma {\VByte {\vec b}} \defined{}& \vec b
\end{flalign*}
The operation \(\inttobitsSym {\itype\tau} : \Z \to \lst\bool\) is defined
in Definition~\ref{definition:int_encode}.
\end{definition}

\begin{definition}
The \emph{unflatten operation
\(\basevalunflattenSym {\btype\tau} \Gamma : \lst\bit \to \baseval\)} is
defined as:
\begin{flalign*}
\basevalunflatten \Gamma \TVoid {\vec b} \defined{}& \VVoid \\
\basevalunflatten \Gamma {\TInt {\itype\tau}} {\vec b} \defined{}&
	\begin{cases}
	\VInt {\itype\tau} {\intofbits {\itype\tau} {\vec\beta}} &
		\textnormal{if \(\vec b\) is a \(\{\BBit 0,\BBit 1\}\) sequence
		\(\vec\beta\)}\\
	\VByte {\vec b} &
		\textnormal{if \(\itype\tau = \IntType \Unsigned \charrank\),
		not \(\vec b\) all in \(\{\BBit 0, \BBit 1\}\), and
		not \(\vec b\) all \(\BIndet\)} \\
	\VIndet {\TInt {\itype\tau}} & \textnormal{otherwise}
	\end{cases} \\
\basevalunflatten \Gamma {\TPtr {\ptype\sigma}} {\vec b} \defined{}&
	\begin{cases}
	\VPtr p &
		\textnormal{if \(\vec b = \BPtrSeg {\PtrSeg p 0} \ldots
		\BPtrSeg {\PtrSeg p {\bitsizeof \Gamma {(\TPtr {\ptype\sigma}}) - 1}}\)
		and \(\typeof p = \ptype\sigma\)} \\
	\VIndet {(\TPtr {\ptype\sigma})} & \textnormal{otherwise}
	\end{cases}
\end{flalign*}
The operation \(\intofbitsSym {\itype\tau} : \lst\bool \to \Z\) is defined
in Definition~\ref{definition:int_encode}.
\end{definition}

The encoding of pointers is an important aspect of the flatten operation
related to our treatment of effective types.
Pointers are encoded as sequences of \emph{frozen} pointer fragment bits
\(\PtrSeg {\freeze p} i\) (see Definition~\ref{definition:ptr_freeze} for
the definition of frozen pointers).
Recall that the flatten operation is used to store base values in memory,
whereas the unflatten operation is used to retrieve them.
This means that whenever a pointer \(p\) is stored and read back, the frozen
variant \(\freeze p\) is obtained.

\begin{lemma}
For each \(\basevaltyped \Gamma \Delta {\bval v} {\btype\tau}\) we have
\(\basevalunflatten \Gamma {\btype\tau} {\basevalflatten \Gamma {\bval v}}
= \freeze {\bval v}\).
\end{lemma}

Freezing formally describes the situations in which type-punning is allowed
since a frozen pointer cannot be used to access a union of another variant
than its current one (Definition~\ref{definition:ctree_lookup}).
Let us consider an example:

\begin{lstlisting}
union U { int x; short y; } u = { .x = 3 };
short *p = &u.y;  // a frozen version of the pointer &u.y is stored
printf("%d", *p); // type-punning via a frozen pointer -> undefined
\end{lstlisting}

Here, an attempt to type-punning is performed via the frozen pointer
\lstinline|p|, which is formally represented as:
\[
\Addr {o_u}
	{\TUnion {\mathtt U}}
	{\RUnion 1 {\mathtt U} \RFrozen}
	0
	{\TBase {\TInt {\IntType \Signed \shortrank}}}
	{\TBase {\TInt {\IntType \Signed \shortrank}}}.
\]

The lookup operation on memory trees (which will be used to obtain the value
of \lstinline|*p| from memory, see Definitions~\ref{definition:ctree_lookup}
and~\ref{definition:memory_operations}) will fail.
The annotation \(\RFrozen\) prevents a union from being
accessed through an address to another variant than its current one.
In the example below type-punning is allowed:

\begin{lstlisting}
union U { int x; short y; } u = { .x = 3 };
printf("%d", u.y);
\end{lstlisting}

Here, type-punning is allowed because it is performed directly via
\lstinline|u.y|, which has not been stored in memory, and thus has not
been frozen.

\begin{definition}
\label{definition:values}
\emph{Abstract values} are inductively defined as:
\begin{flalign*}
\syntax{val}.
\end{flalign*}
\end{definition}

The abstract value \(\VUnionAll t {\vec v}\) represents a union whose variant
is unspecified.
The values \(\vec v\) correspond to interpretations of \emph{all} variants of
\(\TUnion t\). Consider:

\begin{lstlisting}
union U { int x; short y; int *p; } u;
for (size_t i = 0; i < sizeof(u); i++) ((unsigned char*)&u)[i] = 0;
\end{lstlisting}

Here, the object representation of \lstinline|u| is initialized with zeros, and
its variant thus remains unspecified.
The abstract value of \lstinline|u| is\footnote{Note that the \Celeven{} standard
does not guarantee that the \cNULL{} pointer is represented as
zeros, thus
\lstinline|u.p| is not necessarily \cNULL{}.}:
\[
\VUnionAll {\mathtt U} {\listlit {
	\VBase {\VInt {\TInt \TSignedInt} 0},\
	\VBase {\VInt {\TInt {\IntType \Signed \shortrank}} 0},\
	\VBase {\VIndet {(\TPtr {\TType {\TBase {\TInt \TSignedInt}}})}}
}}
\]

Recall that the variants of a union occupy a single memory area, so the
sequence \(\vec v\) of a union value \(\VUnionAll t {\vec v}\) cannot be
arbitrary.
There should be a common bit sequence representing it.
This is not the case in:
\[
\VUnionAll {\mathtt U} {\listlit {
	\VBase {\VInt {\TInt \TSignedInt} 0},\
	\VBase {\VInt {\TInt {\IntType \Signed \shortrank}} 1},\
	\VBase {\VIndet {(\TPtr {\TType {\TBase {\TInt \TSignedInt}}})}}
}}
\]

The typing judgment for abstract values guarantees that \(\vec v\) can be
represented by a common bit sequence.
In order to express this property, we first define the unflatten operation that
converts a bit sequence into an abstract value.

\begin{definition}
\label{definition:val_unflatten}
The \emph{unflatten operation}
\(\valunflattenSym \Gamma \tau : \lst \bit \to \val\) is defined as:
\begin{flalign*}
\valunflatten \Gamma {\TBase {\btype\tau}} {\vec b} \defined{}&
	\VBase {\basevalunflatten \Gamma {\btype\tau} {\vec b}} 
	\tag{the right hand side is Definition~\ref{definition:base_val_flatten}
	on base values} \\
\valunflatten \Gamma {\TArray \tau n} {\vec b} \defined{}&
	\VArray \tau {(
	\ctreeunflatten \Gamma \tau {\sublist 0 s {\vec b}} \dotsc 
	\ctreeunflatten \Gamma \tau
		{\sublist {(n-1)s} {ns} {\vec b}})}
	\ \textnormal{where \(s \defined \bitsizeof \Gamma \tau\)} \\
\valunflatten \Gamma {\TStruct t} {\vec b} \defined{}&
	\MStruct t {(
	\ctreeunflatten \Gamma {\tau_0} {\sublist 0 {s_0} {\vec b}} \dotsc
	\ctreeunflatten \Gamma {\tau_{n-1}}
		{\sublist {z_{n-1}} {z_{n-1} + s_{n-1}} {\vec b}})} \\
	& \textnormal{where \(\elookup t \Gamma = \Some {\vec \tau}\),
		\(n \defined \length {\vec\tau}\),
		\(s_i \defined \bitsizeof \Gamma {\tau_i}\) and
	\(z_i \defined \bitoffsetof \Gamma {\vec\tau} i\)} \\
\valunflatten \Gamma {\TUnion t} {\vec b} \defined{}&
	\MUnionAll t {(
	\ctreeunflatten \Gamma {\tau_0} {\sublist 0 {s_0} {\vec b}} \dotsc
	\ctreeunflatten \Gamma {\tau_{n-1}} {\sublist 0 {s_{n-1}} {\vec b}})} \\
	& \textnormal{where \(\elookup t \Gamma = \Some {\vec \tau}\),
		\(n \defined \length {\vec\tau}\) and
		\(s_i \defined \bitsizeof \Gamma {\tau_i}\)}
\end{flalign*}
\end{definition}

\begin{definition}
\label{definition:val_typed}
The judgment \(\valtyped \Gamma \Delta v \tau\) describes that \emph{the value
\(v\) has type \(\tau\)}.
It is inductively defined as:
\allowdisplaybreaks[4]
\begin{gather*}
\renewcommand{\ScoreOverhang}{0.1em}
\AXC{\(\basevaltyped \Gamma \Delta {\bval v} {\btype\tau}\)}
\UIC{\(\valtyped \Gamma \Delta
	{\VBase {\bval v}} {\TBase {\btype\tau}}\)}
\DP\qquad
\AXC{\(\valtyped \Gamma \Delta {\vec v} \tau\)}
\AXC{\(\length {\vec v} = n \ne 0\)}
\BIC{\(\valtyped \Gamma \Delta {\VArray \tau {\vec v}}
	{\TArray \tau n}\)}
\DP\\[0.5em]
\AXC{\(\elookup t \Gamma = \vec\tau\)}
\AXC{\(\valtyped \Gamma \Delta {\vec v} {\vec\tau}\)}
\BIC{\(\valtyped \Gamma \Delta {\VStruct t {\vec v}} {\TStruct t}\)}
\DP\qquad
\AXC{\(\elookup t \Gamma = \vec\tau\)}
\AXC{\(i < \length {\vec\tau}\)}
\AXC{\(\ctreetyped \Gamma \Delta v {\tau_i}\)}
\TIC{\(\ctreetyped \Gamma \Delta {\VUnion t i v} {\TUnion t}\)}
\DP\\[0.5em]
\AXC{\(\elookup t \Gamma = \vec\tau\)}
\AXC{\(\valtyped \Gamma \Delta {\vec v} {\vec\tau}\)}
\AXC{\(\bitvalid \Gamma \Delta {\vec b}\)}
\AXC{\(\forall i \wsdot (v_i = \valunflatten \Gamma {\tau_i}
	{\sublist 0 {\bitsizeof \Gamma {\tau_i}} {\vec b}})\)}
\QIC{\(\valtyped \Gamma \Delta {\VUnionAll t {\vec v}} {\TUnion t}\)}
\DP
\end{gather*}
\end{definition}

The flatten operation \(\valflattenSym \Gamma : \val \to \lst\bit\), which
converts an abstract value \(v\) into a bit representation
\(\valflatten \Gamma v\), is more difficult to define (we need this operation to
define the conversion operation from abstract values into memory trees, see
Definition~\ref{definition:of_val}).
Since padding bits are not present in abstract values, we have to insert these.
Also, in order to obtain the bit representation of an unspecified
\(\VUnionAll t {\vec v}\) value, we have to \emph{construct} the common bit sequence
\(\vec b\) representing \(\vec v\).
The typing judgment guarantees that such a sequence exists, but since it is
not explicit in the value \(\VUnionAll t {\vec v}\), we have to reconstruct
it from \(\vec v\).
Consider:

\begin{lstlisting}
union U { struct S { short y; void *p; } x1; int x2; };
\end{lstlisting}

Assuming \(\sizeof \Gamma {(\TBase {\TInt {\IntType \Signed \intrank}})}
= \sizeof \Gamma {(\TBase {\TPtr \TAny})} = 4\) and
\(\sizeof \Gamma {(\TBase {\TInt {\IntType \Signed \shortrank}})} = 2\),
a well-typed \(\TUnion {\mathtt U}\) value of an unspecified variant may
be:
\[v = \VUnionAll {\mathtt U} {\listlit {
	\VStruct {\mathtt S} {\listlit {
		\VBase {\VInt {\IntType \Signed \shortrank} 0},
		\VBase {\VPtr p}
	}}, \VBase {\VInt {\IntType \Signed \intrank} 0}
}}.\]

\noindent The flattened versions of the variants of \(v\) are:
\begin{equation*}
\begin{aligned}
\valflatten \Gamma {\VStruct {\mathtt S} {\listlit {
	\VBase {\VInt {\IntType \Signed \shortrank} 0},
	\VBase {\VPtr p}}}}={}&
	0 \dotsc 0 \; 0 \dotsc 0 \;
	\BIndet \dotsc \BIndet \; \BIndet \dotsc \BIndet \;
	\BPtrSeg {\PtrSeg p 0} \dotsc \BPtrSeg {\PtrSeg p {31}} \\
\valflatten \Gamma {\VBase {\VInt {\IntType \Signed \intrank} 0}} ={}&
	0 \dotsc 0 \; 0 \dotsc 0 \; 0 \dotsc 0 \; 0 \dotsc 0 \\ \midrule
\valflatten \Gamma v ={}&
	0 \dotsc 0 \; 0 \dotsc 0 \; 0 \dotsc 0 \; 0 \dotsc 0 \;
	\BPtrSeg {\PtrSeg p 0} \dotsc \BPtrSeg {\PtrSeg p {31}} 
\end{aligned}
\end{equation*}

This example already illustrates that so as to obtain the common bit sequence
\(\valflatten \Gamma v\) of \(v\) we have to insert padding bits and
``join'' the padded bit representations.

\begin{definition}
The \emph{join operation on bits \(\bitjoinSym : \bit\to\bit\to\bit\)} is
defined as:
\[
	\bitjoin \BIndet b \defined b \qquad
	\bitjoin b \BIndet \defined b \qquad
	\bitjoin b b \defined b.
\]
\end{definition}

\begin{definition}
The \emph{flatten operation \(\valflattenSym \Gamma : \val \to \lst\bit\)} is
defined as:
\begin{flalign*}
\valflatten \Gamma {\VBase {\bval v}} \defined{}&
	\basevalflatten \Gamma {\bval v} \\
\valflatten \Gamma {\VArray \tau {\vec v}} \defined{}&
	\valflatten \Gamma {v_0} \dotsc
	\valflatten \Gamma {v_{\length {\vec v} - 1}} \\
\valflatten \Gamma {\VStruct t {\vec v}} \defined{}&
	\sublistpad 0 {z_0} \BIndet {\valflatten \Gamma {v_0}} \dotsc
	\sublistpad 0 {z_{n-1}} \BIndet {\valflatten \Gamma {v_{n-1}}} \\
	& \textnormal{where \(\elookup t \Gamma = \Some {\vec \tau}\),
	\(n \defined \length{\vec \tau}\), and
	\(z_i \defined \bitoffsetof \Gamma {\vec\tau} i\)} \\
\valflatten \Gamma {\VUnion t i v} \defined{}&
	\sublistpad 0 {\bitsizeof \Gamma {(\TUnion t)}} \BIndet
		{\valflatten \Gamma v} \\
\valflatten \Gamma {\VUnionAll t {\vec v}} \defined{}&
	\textstyle\bigsqcup_{i=0}^{\length{\vec v}-1}
	\sublistpad 0 {\bitsizeof \Gamma {(\TUnion t)}} \BIndet
		{\valflatten \Gamma {v_i}}
\end{flalign*}
\end{definition}

The operation \(\ofvalSym \Gamma : \lst\perm \to \val \to \mtree\), which
converts a value \(v\) of type \(\tau\) into a memory tree
\(\ofval \Gamma {\vec\gamma} v\), is albeit technical fairly straightforward.
In principle it is just a recursive definition that uses the flatten 
operation \(\basevalflatten \Gamma {\bval v}\) for base values
\(\VBase {\bval v}\) and the flatten operation
\(\valflatten \Gamma {\VUnionAll t {\vec v}}\) for unions
\(\VUnionAll t {\vec v}\) of an unspecified variant.

The technicality is that abstract values do not contain permissions, so we
have to merge the given value with permissions.
The sequence \(\vec\gamma\) with \(\length {\vec\gamma} = \bitsizeof \Gamma \tau\)
represents a flattened sequence of permissions.
In the definition of the memory store \(\meminsert \Gamma a v m\) (see
Definition~\ref{definition:memory_operations}), we convert \(v\) into the
stored memory tree \(\ofval \Gamma {\vec\gamma} v\) where \(\gamma\) constitutes
the old permissions of the object at address \(a\).

\begin{definition}
\label{definition:of_val}
The operation \(\ofvalSym \Gamma : \lst\perm \to \val \to \mtree\) is defined
as:
\begin{flalign*}
\ofval \Gamma {\vec\gamma} {(\VBase {\bval v})} \defined{}&
	\MBase {\typeof {\bval v}} {\vv{\gamma b}}
	\quad\textnormal{where
		\(\vec b \defined \basevalflatten \Gamma {\bval v}\)}\\
\ofval \Gamma {\vec\gamma} {(\VArray \tau {\vec v})} \defined{}&
	\MArray \tau
		{(\ofval \Gamma {\sublist 0 s {\vec\gamma}} {v_0} \dotsc
		\ofval \Gamma {\sublist {(n - 1)s} {ns} {\vec\gamma}} {v_{n - 1}})} \\
	& \textnormal{where \(s \defined \bitsizeof \Gamma \tau\)
	and \(n \defined \length {\vec v}\)} \\
\ofval \Gamma {\vec\gamma} {(\VStruct t {\vec v})} \defined{}&
	\MStruct t {\begin{pmatrix*}[l]
		\ofval \Gamma {\sublist 0 {s_0} {\vec\gamma}} {v_0}\
		\zipBIndet {\sublist{s_0} {z_0} {\vec\gamma}}
		\\[-0.2em] \dotsc \\
		\ofval \Gamma {\sublist {z_{n-1}} {z_{n-1} + s_{n-1}} {\vec\gamma}}
			{v_{n - 1}}\
		\zipBIndet {\sublist{z_{n-1} + s_{n-1}} {z_n} {\vec\gamma}}
	\end{pmatrix*}} \\
	& \textnormal{where \(\elookup t \Gamma = \Some {\vec \tau}\),
		\(n \defined \length{\vec \tau}\),
		\(s_i \defined \bitsizeof \Gamma {\tau_i}\)}\\[-0.2em]
	& \textnormal{and \(z_i \defined \bitoffsetof \Gamma {\vec\tau} i\)} \\
\ofval \Gamma {\vec\gamma} {(\VUnion t i v)} \defined{}&
	\MUnion t i {\ofval \Gamma {\sublist 0 s {\vec\gamma}} v}
		{\zipBIndet {\sublist s {\bitsizeof \Gamma {(\TUnion t)}} {\vec\gamma}}} \\
	& \textnormal{where \(s \defined \bitsizeof \Gamma {(\typeof v)}\)} \\
\ofval \Gamma {\vec\gamma} {(\VUnionAll t {\vec v})} \defined{}&
	\MUnionAll t {\vv{\gamma b}}
	\quad\textnormal{where
	\(\vec b \defined \valflatten \Gamma {\VUnionAll t {\vec v}}\)}
\end{flalign*}
\end{definition}

Converting a memory tree into a value is as expected: permissions are removed
and unions are interpreted as values corresponding to each variant.

\begin{definition}
The operation \(\tovalSym \Gamma : \mtree \to \val\) is defined as:
\begin{flalign*}
\toval \Gamma {(\MBase {\btype \tau} {\vv {\gamma b}})} \defined{}&
	\VBase {\basevalunflatten \Gamma {\btype\tau} {\vec b}} \\
\toval \Gamma {(\MArray \tau {\vec w})} \defined{}&
	\VArray \tau
		{(\toval \Gamma {w_0} \dotsc \toval \Gamma {w_{\length {\vec w} - 1}})} \\
\toval \Gamma {(\MStruct t {\vv{w\vec {\ppbit b}}})} \defined{}&
	\VStruct t
		{(\toval \Gamma {w_0} \dotsc \toval \Gamma {w_{\length {\vec w} - 1}})} \\
\toval \Gamma {(\MUnion t i w {\vec {\ppbit b}})} \defined{}&
	\VUnion t i {\toval \Gamma w} \\
\toval \Gamma {(\MUnionAll t {\vv {\gamma b}})} \defined{}&
	\valunflatten \Gamma {\TUnion t} {\vec b}
\end{flalign*}
\end{definition}

The function \(\tovalSym \Gamma\) is an inverse of \(\ofvalSym \Gamma\) up to
freezing of pointers.
Freezing is intended, it makes indirect type-punning illegal.

\begin{lemma}
\label{lemma:to_of_val}
Given \(\valtyped \Gamma \Delta v \tau\), and let \(\vec\gamma\) be a flattened
sequence of permissions with \(\length {\vec\gamma} = \bitsizeof \Gamma \tau\),
then we have:
\[
	\toval \Gamma {(\ofval \Gamma {\vec\gamma} v)} = \freeze v.
\]
\end{lemma}

The other direction does not hold because invalid bit representations will
become indeterminate values.

\begin{lstlisting}
struct S { int *p; } s;
for (size_t i = 0; i < sizeof(s); i++) ((unsigned char*)&s)[i] = i;
// s has some bit representation that does not constitute a pointer
struct S s2 = s;
// After reading s, and storing it, there are no guarantees about s2,
// whose object representation thus consists of $\BIndet$s
\end{lstlisting}

We finish this section by defining the indeterminate abstract value
\(\valnew \Gamma \tau\), which consists of indeterminate base values.
The definition is similar to its counterpart on memory trees
(Definition~\ref{definition:ctree_new}).

\begin{definition}
\label{definition:val_new}
The operation \(\valnew \Gamma : \type \to \val\) that yields the
indeterminate value is defined as:
\[
\valnew \Gamma \tau \defined \valunflatten \Gamma \tau
	{\replicate {\bitsizeof \Gamma \tau} \BIndet}.
\]
\end{definition}

\begin{lemma}
If \(\typevalid \Gamma \tau\), then:
\[
	\toval \Gamma {(\ctreenew \Gamma \gamma \tau)} = \valnew \Gamma \tau
	\quad\textnormal{and}\quad
	\ofval \Gamma {(\replicate {\bitsizeof \Gamma \tau} \gamma)}
		{(\valnew \Gamma \tau)} = \ctreenew \Gamma \gamma \tau.
\]
\end{lemma}

\subsection{Memory operations}
\label{section:memory_operations}

Now that we have all primitive definitions in place, we can compose these to
implement the actual memory operations as described in the beginning of this
section.
The last part that is missing is a data structure to keep track of objects that
have been locked.
Intuitively, this data structure should represent a set of addresses, but up
to overlapping addresses.

\begin{definition}
\label{definition:lockset}
\emph{Locksets} are defined as:
\begin{flalign*}
\syntax{lockset}.
\end{flalign*}
\end{definition}

Elements of locksets are pairs \(\pair o i\) where \(o \in \memindex\) describes
the object identifier and \(i \in \nat\) a bit-offset in the object described
by \(o\).
We introduce a typing judgment to describe that the structure of locksets
matches up with the memory layout.

\begin{definition}
The judgment \(\locksetvalid \Gamma \Delta \Omega\) describes that \emph{the
lockset \(\Omega\) is valid}.
It is inductively defined as:
\begin{equation*}
\AXC{for each\quad \(\pair o i \in \Omega
	\quad\textnormal{there is a }\tau\textnormal{ with}\quad
	\indextyped \Delta o \tau \textnormal{ and }
	i < \bitsizeof \Gamma \tau\)}
\UIC{\(\locksetvalid \Gamma \Delta \Omega\)}
\DP
\end{equation*}
\end{definition}

\begin{definition}
The \emph{singleton lockset} \(\locksingletonSym \Gamma : \addr \to \lockset\)
is defined as:
\begin{equation*}
\locksingleton \Gamma a \defined
\{ (\addrindex a, i) \separator \addrobjectoffset \Gamma a \le i <
	\addrobjectoffset \Gamma a + \bitsizeof \Gamma {(\typeof a)} \}.
\end{equation*}
\end{definition}

\begin{lemma}
If \(\addrtyped \Gamma \Delta {a_1} {\sigma_1}\) and
\(\addrtyped \Gamma \Delta {a_2} {\sigma_2}\) and
\(\addrstrict \Gamma {\{a_1,a_2\}}\), then:
\begin{equation*}
	\addrdisjoint \Gamma {a_1} {a_2} \quad\textnormal{implies}\quad
	\locksingleton \Gamma {a_1} \cap \locksingleton \Gamma {a_2} = \emptyset.
\end{equation*}
\end{lemma}

\begin{definition}
\label{definition:memory_operations}
The \emph{memory operations} are defined as:
\allowdisplaybreaks[4]
\begin{flalign*}
\memlookup \Gamma a m \defined{}&
	\toval \Gamma w
	\quad \textnormal{ if \(\cmaplookup \Gamma a m = \Some w\) and
	\(\forall i \wsdot \Some\Readable \subseteq \permkind (\ctreeflatten w)_i\)}\\
\memforce \Gamma a m \defined{}&
	\minsert {(\addrindex a)}
		{\pair {\ctreealter \Gamma {\lambda w' \wsdot w'}
		{\addrref \Gamma a} w} \mu} m
		\quad \textnormal{if
		\(\mlookup {(\addrindex a)} m = \Some {\pair w \mu}\)} \\
\meminsert \Gamma a v m \defined{}&
	\cmapalter \Gamma {\lambda w \wsdot
		\ofval \Gamma {(\fst {\ctreeflatten w})} v} a m \\
\memwritable \Gamma a m \defined{} &
	\exists w \wsdot \cmaplookup \Gamma a m = \Some w
	\textnormal{ and \(\forall i \wsdot
	\Some\Writable \subseteq \permkind (\ctreeflatten w)_i\)}\\
\memlock \Gamma a m \defined{}&
	\cmapalter \Gamma {\lambda w \wsdot
		\textnormal{ apply \(\permlockSym\) to all permissions of \(w\)}} a m \\
\memunlock \Omega m \defined{}&
	\{ \pair o {\pair {\ctreemerge f w {\vec y}} \mu} \separator
		\mlookup o m = \Some {\pair w \mu} \}
	\cup
	\{ \pair o \tau \separator \mlookup o m = \Some \tau \} \\
	& \textnormal{where \(\begin{array}[t]{@{}l}
		f\,\pair \gamma b\,\true \defined \pair {\permunlock \gamma} b \\
		f\,\pair \gamma b\,\false \defined \pair \gamma b,
		\end{array}\)}\\[-0.2em]
	& \textnormal{and \(\vec y \defined (\pair o 0 \in \Omega) \dotsc 
		(\pair o {\length{\bitsizeof \Gamma {(\typeof w)}} - 1} \in \Omega)\)} \\
\memalloc \Gamma o v \mu m \defined{}&
	\minsert o {\pair
		{\ofval \Gamma	{(
			\replicate {\bitsizeof \Gamma {(\typeof v)}} {\LUnlocked {\pair 0 1}}
		)} v} \mu} m\\
\memfreeable a m \defined{}&
	\exists o\,\tau\,\sigma\,n\,w\wsdot
	\begin{array}[t]{@{}l}
	a = \Addr o \tau {\RArray 0 \tau n} 0 \tau \sigma,\
	\mlookup o m = \Some {\pair w \true} \\
	\textnormal{and all } \ctreeflatten w 
	\textnormal{ have the permission } \LUnlocked {\pair 0 1}
	\end{array}	\\
\memfree o m \defined{}&
	\minsert o {\typeof w}	 m
	\quad \textnormal{ if \(\mlookup o m = \Some {\pair w \mu}\)}
\end{flalign*}
\end{definition}

The lookup operation \(\memlookup \Gamma a m\) uses the lookup operation
\(\cmaplookup \Gamma a m\) that yields a memory tree \(w\)
(Definition~\ref{definition:cmap_lookup}), and then converts \(w\) into the
value \(\toval \Gamma w\).
The operation \(\cmaplookup \Gamma a m\) already yields \(\None\) in case
effective types are violated or \(a\) is an end-of-array address.
The additional condition of \(\memlookup \Gamma a m\) ensures that the
permissions allow for a read access.
Performing a lookup affects the effective types of the object at address \(a\).
This is factored out by the operation \(\memforce \Gamma a m\) which applies
the identity function to the subobject at address \(a\) in the memory \(m\).
Importantly, this does not change the memory contents, but merely
changes the variants of the involved unions.

The store operation \(\meminsert \Gamma a v m\) uses the alter operation
\(\cmapalter \Gamma {\lambda w \wsdot
\ofval \Gamma {(\fst {\ctreeflatten w})} v} a m\) on memories
(Definition~\ref{definition:cmap_alter}) to apply
\(\lambda w \wsdot \ofval \Gamma {(\fst {\ctreeflatten w})} v\) to
the subobject at address \(a\).
The stored value \(v\) is converted into a memory tree while retaining the
permissions \(\fst {\ctreeflatten w}\) of the previously stored memory tree
\(w\) at address \(a\).

The definition of \(\memlock \Gamma a m\) is straightforward.
In the \Coq{} development we use a map operation on memory trees to apply the
function \(\permlockSym\) (Definition~\ref{definition:permissions}) to the
permission of each bit of the memory tree at address \(a\).

The operation \(\memunlock \Omega m\) unlocks a whole lockset \(\Omega\),
rather than an individual address, in memory \(m\).
For each memory tree \(w\) at object identifier \(o\), it converts \(\Omega\)
to a Boolean vector \(\vec y = (\pair o 0 \in \Omega) \dotsc 
(\pair o {\length{\bitsizeof \Gamma {(\typeof w)}} - 1} \in \Omega)\) and merges
\(w\) with \(\vec y\) (using Definition~\ref{definition:ctree_merge}) to apply
\(\permunlockSym\) (Definition~\ref{definition:permissions}) to the permissions
of bits that should be unlocked in \(w\).
We show some lemmas to illustrate that the operations for locking and unlocking
enjoy the intended behavior:

\begin{lemma}
If \(\memvalid \Gamma \Delta m\) and
\(\addrtyped \Gamma \Delta a {\TType \tau}\) and
\(\memwritable \Gamma a m\), then we have:
\[
	\memlocks {(\memlock \Gamma a m)} = \memlocks m \cup \locksingleton \Gamma a.
\]
\end{lemma}

\begin{lemma}
If \(\Omega \subseteq \memlocks m\), then
\(\memlocks {(\memunlock \Omega m)} = \memlocks m \setminus \Omega\).
\end{lemma}
	
Provided \(o \notin \mdom m\), allocation \(\memalloc \Gamma o v \mu m\)
extends the memory with a new object holding the value \(v\) and \emph{full}
permissions \(\LUnlocked {\pair 0 1}\).
Typically we use \(v = \valnew \Gamma \tau\) for some \(\tau\), but global
and static variables are allocated with a specific value \(v\).

The operation \(\memfree o m\) deallocates the object \(o\) in \(m\), and keeps
track of the type of the deallocated object.
In order to deallocate dynamically obtained memory via \lstinline|free|, the
side-condition \(\memfreeable a m\) describes that the permissions are
sufficient for deallocation, and that \(a\) points to the first element of an
\lstinline|malloc|ed array.

All operations preserve typing and satisfy the expected laws about their
interaction.
We list some for illustration.

\begin{fact}
\label{fact:mem_writable_lookup}
If \(\memwritable \Gamma a m\), then there exists a value \(v\) with
\(\memlookup \Gamma m a = \Some v\).
\end{fact}

\begin{lemma}[Stores commute]
If \(\memvalid \Gamma \Delta m\) and
\(\addrdisjoint \Gamma {a_1} {a_2}\) with:
\begin{itemize}
\item \(\addrtyped \Gamma \Delta {a_1} {\tau_1}\),
	\(\memwritable \Gamma {a_1} m\), and
	\(\ctreetyped \Gamma \Delta {v_1} {\tau_1}\), and
\item \(\addrtyped \Gamma \Delta {a_2} {\tau_2}\),
	\(\memwritable \Gamma {a_2} m\), and
	\(\ctreetyped \Gamma \Delta {v_2} {\tau_2}\),
\end{itemize}
then we have:
\[
	\meminsert \Gamma {a_1} {v_1} {\meminsert \Gamma {a_2} {v_2} m} =
	\meminsert \Gamma {a_2} {v_2} {\meminsert \Gamma {a_1} {v_1} m}.
\]
\end{lemma}

\begin{lemma}[Looking up after storing]
\label{lemma:mem_lookup_insert}
If \(\memvalid \Gamma \Delta m\) and
\(\addrtyped \Gamma \Delta a \tau\) and
\(\valtyped \Gamma \Delta v \tau\) and
\(\memwritable \Gamma a m\)
and \(a\) is not a byte address, then we have: 
\[
	\memlookup \Gamma a {(\meminsert \Gamma a v m)} = \Some {\freeze v}.
\]
\end{lemma}

Storing a value \(v\) in memory and then retrieving it, does not necessarily
yield the same value \(v\).
It intentionally yields the value \(\freeze v\) whose pointers have been
frozen.
Note that the above result does not hold for byte addresses, which may store
a value in a padding byte, in which case the resulting value is indeterminate.

\begin{lemma}[Stores and look ups commute]
If \(\memvalid \Gamma \Delta m\) and \(\addrdisjoint \Gamma {a_1} {a_2}\) and
\(\addrtyped \Gamma \Delta {a_2} {\tau_2}\) and
\(\memwritable \Gamma {a_2} m\) and
\(\valtyped \Gamma \Delta {v_2} {\tau_2}\), then we have:
\[
\memlookup \Gamma {a_1} m = \Some {v_1}
\quad\textnormal{implies}\quad
\memlookup \Gamma {a_1} {(\meminsert \Gamma {a_2} {v_2} m)} = \Some {v_1}.
\]
\end{lemma}
 
\noindent
These results follow from Lemma~\ref{lemma:cmap_alter_commute},
\ref{lemma:cmap_lookup_alter} and~\ref{lemma:to_of_val}.

\section{Formal proofs}
\label{section:formal_proofs}

\subsection{Type-based alias analysis}
\label{section:aliasing}

The purpose of \Celeven's notion of effective types~\cite[6.5p6-7]{iso:12} is
to make it possible for compilers to perform typed-based alias analysis.
Consider:

\begin{lstlisting}
short g(int *p, short *q) {
  short x = *q; *p = 10; return x;
}
\end{lstlisting}

Here, a compiler should be able to assume that \lstinline|p| and \lstinline|q|
are not aliased because they point to objects with different types
(although the integer types \(\IntType \Signed \shortrank\) and
\(\IntType \Signed \intrank\) may have the same representation, they have
different integer ranks, see Definition~\ref{definition:int_coding_spec}, and
are thus different types).
If \lstinline|g| is called with aliased pointers, execution of the function
body should have undefined behavior in order to allow a compiler to soundly
assume that \lstinline|p| and \lstinline|q| are not aliased.

From the \Celeven{} standard's description of effective types it is not
immediate that calling \lstinline|g| with aliased pointers results
in undefined behavior.
We prove an abstract property of our memory model that shows that this is
indeed a consequence, and that indicates a compiler can perform type-based alias
analysis.
This also shows that our interpretation of effective types of the \Celeven{}
standard, in line with the interpretation from the \GCC{}
documentation~\cite{gnu:11}, is sensible.

\begin{definition}
\label{definition:subtype}
A type \(\tau\) is a \emph{subobject type of} \(\sigma\), notation
\(\subtype \Gamma \tau \sigma\), if there exists some reference \(\vec r\) with
\(\reftyped \Gamma {\vec r} \sigma \tau\).
\end{definition}

For example, \lstinline|int[2]| is a subobject type of
\lstinline|struct S { int x[2]; int y[3]; }| and \lstinline|int[2][2]|, but not of
\lstinline|struct S { short x[2]; }|, nor of \lstinline|int(*)[2]|.

\begin{theorem}[Strict-aliasing]
\label{theorem:aliasing}
Given \(\memvalid \Gamma \Delta m\), frozen addresses \(a_1\)
and \(a_2\) with \(\addrtyped \Delta m {a_1} {\sigma_1}\) and
\(\addrtyped \Delta m {a_2} {\sigma_2}\) and
\(\sigma_1, \sigma_2 \neq \TBase {\TInt {\IntType \Unsigned \charrank}}\), then
either:
\begin{enumerate}
\item We have \(\subtype \Gamma {\sigma_1} {\sigma_2}\) or
	\(\subtype \Gamma {\sigma_2} {\sigma_1}\).
\item We have \(\addrdisjoint \Gamma {a_1} {a_2}\).
\item Accessing \(a_1\) after accessing \(a_2\) and \viceversa{} fails.
	That means:
	\begin{enumerate}
	\item \(\memlookup \Gamma {a_1} {(\memforce \Gamma {a_2} m)} = \None\) and
	\(\memlookup \Gamma {a_2} {(\memforce \Gamma {a_1} m)} = \None\), and
	\item
	\(\memlookup \Gamma {a_1} {\meminsert \Gamma {a_2} {v_1} m} = \None\) and
	\(\memlookup \Gamma {a_2} {\meminsert \Gamma {a_1} {v_2} m} = \None\)
	for all stored values \(v_1\) and \(v_2\).
	\end{enumerate}
\end{enumerate}
\end{theorem}

This theorem implies that accesses to addresses of disjoint type are either
non-overlapping or have undefined behavior.
Fact~\ref{fact:mem_writable_lookup} accounts for a store after a lookup.
Using this theorem, a compiler can optimize the generated code in the
example based on the assumption that \lstinline|p| and \lstinline|q| are
not aliased.
Reconsider:

\begin{lstlisting}
short g(int *p, short *q) { short x = *q; *p = 10; return x; }
\end{lstlisting}

If \lstinline|p| and \lstinline|q| are aliased, then calling \lstinline|g|
yields undefined behavior because the assignment \lstinline|*p = 10| violates
effective types.
Let \(m\) be the initial memory while executing
\lstinline|g|, and let \(a_{\mathtt p}\) and
\(a_{\mathtt q}\) be the addresses corresponding to \lstinline|p| and
\lstinline|q|, then the condition \(\memwritable \Gamma {a_{\mathtt p}}
{(\memforce \Gamma {a_{\mathtt q}} m)}\) does not hold by
Theorem~\ref{theorem:aliasing} and Fact~\ref{fact:mem_writable_lookup}.

\subsection{Memory refinements}
\label{section:refinements}

This section defines the notion of \emph{memory refinements} that allows us to
relate memory states.
The author's PhD thesis~\cite{kre:15:phd} shows that the \CHtwoO{}
operational semantics is invariant under this notion.
Memory refinements form a general way to validate many common-sense properties
of the memory model in a formal way.
For example, they show that the memory is invariant under relabeling.
More interestingly, they show that symbolic information (such as variants of
unions) cannot be observed.

Memory refinements also open to door to reason about program transformations.	
We demonstrate their usage by proving soundness of constant propagation and by
verifying an abstract version of \lstinline|memcpy|. 

Memory refinements are a variant of Leroy and Blazy's
notion of memory extensions and injections~\cite{ler:bla:08}.
A memory refinement is a relation \(\memrefineShort \Gamma f {m_1} {m_2}\)
between a source memory state \(m_1\) and target memory state \(m_2\), where:

\begin{enumerate}
\item The function \(f : \memindex \to \option {(\memindex \times \rref)}\) is
	used to rename object identifiers and to coalesce multiple objects into
	subobjects of a compound object.
\item Deallocated objects in \(m_1\) may be replaced by arbitrary objects in
	\(m_2\).
\item Indeterminate bits \(\BIndet\) in \(m_1\) may be replaced by
	arbitrary bits in \(m_2\).
\item Pointer fragment bits \(\BPtrSeg {\PtrSeg p i}\) that belong to
	deallocated	pointers in \(m_1\) may be replaced by	arbitrary bits in \(m_2\).
\item Effective types may be weakened.
	That means, unions with a specific variant in \(m_1\) may be replaced
	by unions with an unspecified variant in \(m_2\), and pointers with frozen
	union annotations \(\RFrozen\) in \(m_1\) may be replaced by
	pointers	with unfrozen union annotations \(\RUnfrozen\) in \(m_2\).
\end{enumerate}

The key property of a memory refinement
\(\memrefineShort \Gamma f {m_1} {m_2}\), as
well as of Leroy and Blazy's memory extensions and injections, is that memory
operations are more defined on the target memory \(m_2\) than on the
source memory \(m_1\).
For example, if a lookup succeeds on \(m_1\), it also succeed on \(m_2\)
and yield a related value.

The main judgment \(\memrefine \Gamma f {\Delta_1} {\Delta_2} {m_1} {m_2}\) of
memory refinements will be built using a series of refinement relations on the
structures out of which the memory consists (addresses, pointers, bits, memory
trees, values).
All of these judgments should satisfy some basic properties, which are captured
by the judgment \(\memrefineShort \Delta f {\Delta_1} {\Delta_2}\).

\begin{definition}
A \emph{renaming function} \(f : \memindex \to
\option {(\memindex \times \rref)}\) is a \emph{refinement}, notation
\(\memrefineShort \Delta f {\Delta_1} {\Delta_2}\), if the following
conditions hold:
\begin{enumerate}
\item If	\(f\,o_1 = \Some {\pair o {\vec r_1}}\) and
	\(f\,o_2 = \Some {\pair o {\vec r_2}}\), then
	\(o_1 = o_2\) or \(\refdisjoint {\vec r_1} {\vec r_2}\)
	(\emph{injectivity}).
\item If	\(f\,o_1 = \Some {\pair {o_2} {\vec r}}\), then \(\frozen {\vec r}\).
\item If	\(f\,o_1 = \Some {\pair {o_2} {\vec r}}\) and
	\(\indextyped {\Delta_1} {o_1} \sigma\), then 
	\(\indextyped {\Delta_2} {o_2} \tau\) and
	\(\reftyped \Gamma {\vec r} \tau \sigma\) for a \(\tau\).
\item If	\(f\,o_1 = \Some {\pair {o_2} {\vec r}}\) and
	\(\indextyped {\Delta_2} {o_2} \tau\), then 
	\(\indextyped {\Delta_1} {o_1} \sigma\) and
	\(\reftyped \Gamma {\vec r} \tau \sigma\) for a \(\sigma\).	
\item If	\(f\,o_1 = \Some {\pair {o_2} {\vec r}}\) and
	\(\indexalive {\Delta_1} {o_1}\), then \(\indexalive {\Delta_2} {o_2}\).	
\end{enumerate}
\end{definition}

The renaming function \(f : \memindex \to \option {(\memindex \times \rref)}\)
is the core of all refinement judgments.
It is used to rename object identifiers and to coalesce multiple source objects
into subobjects of a single compound target object.

Consider a renaming \(f\) with
\(f\,o_1 = \Some {\pair {o_1} {\RStruct 0 t}}\)
and \(f\,o_2 = \Some {\pair {o_1} {\RStruct 1 t}}\), and an environment
\(\Gamma\) with \(\elookup t \Gamma = \listlit {\tau_1,\tau_2}\).
This gives rise to following refinement:
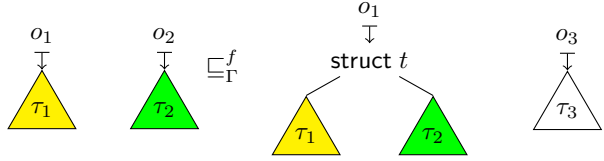
\begin{equation*}
\tikzset{
	level distance=3em,
	baseline=(current bounding box.center),
	sibling distance=5em,
	child anchor=north,
	leaf/.style={draw=black,regular polygon,regular polygon sides=3,
		inner sep=0.1em},
}
\begin{tikzpicture}
\node[leaf,fill=yellow] (tau1) {\(\tau_1\)};
\draw[<-|] (tau1.north) -- node[above,pos=1] {\(o_1\)} +(0,0.7em);
\node[leaf,fill=green,right=3em of tau1] (tau2) {\(\tau_2\)};
\draw[<-|] (tau2.north) -- node[above,pos=1] {\(o_2\)} +(0,0.7em);
\end{tikzpicture}
\memrefineShortSym \Gamma f \quad
\begin{tikzpicture}
\node (t) {\(\TStruct t\)}
	child[-] { node[leaf,fill=yellow] {\(\tau_1\)}}
	child[-] { node[leaf,fill=green] {\(\tau_2\)}};
\draw[<-|] (t.north) -- node[above,pos=1] {\(o_1\)} +(0,0.7em);
\node[leaf, yshift=-2em,right=5em of t] (tau3) {\(\tau_3\)};
\draw[<-|] (tau3.north) -- node[above,pos=1] {\(o_3\)} +(0,0.7em);
\end{tikzpicture}
\end{equation*}

Injectivity of renaming functions guarantees that distinct source objects are
coalesced into disjoint target subobjects.
In the case of Blazy and Leroy, the renaming functions have type \(\memindex
\to \option {(\memindex \times \nat)}\), but we replaced the natural number by a
reference since our memory model is structured using trees.

Since memory refinements rearrange the memory layout, addresses should be
rearranged accordingly.
The judgment
\(\addrrefine \Gamma f {\Delta_1} {\Delta_2} {a_1} {a_2} {\ptype\tau}\)
describes how \(a_2\) is obtained by renaming \(a_1\) according to the renaming
\(f\), and moreover allows frozen union annotations \(\RFrozen\) in \(a_1\) to
be changed into unfrozen ones \(\RUnfrozen\) in \(a_2\).
The index \(\ptype\tau\) in the judgment \(\addrrefine \Gamma f {\Delta_1}
{\Delta_2} {a_1} {a_2} {\ptype\tau}\) corresponds to the type of \(a_1\) and
\(a_2\).

The judgment for addresses is lifted to the judgment for pointers in the
obvious way.
The judgment for bits is inductively defined as:
\begin{gather*}
\AXC{\strut\(\beta \in \{ 0, 1 \}\)}
\UIC{\(\bitrefine \Gamma f {\Delta_1} {\Delta_2} {\BBit \beta} {\BBit \beta}\)}
\DP\qquad
\AXC{\strut
	\(\ptrrefine \Gamma f {\Delta_1} {\Delta_2} {p_1} {p_2} {\ptype\sigma}\)}
\AXC{\(\frozen {p_2}\)}
\AXC{\(i < \bitsizeof \Gamma {(\TBase {\TPtr {\ptype\sigma}})}\)}
\TIC{\(\bitrefine \Gamma f {\Delta_1} {\Delta_2}
	{\BPtrSeg {\PtrSeg {p_1} i}} {\BPtrSeg {\PtrSeg {p_2} i}}\)}
\DP\\[0.5em]
\AXC{\(\bitvalid \Gamma {\Delta_2} b\)}
\UIC{\(\bitrefine \Gamma f {\Delta_1} {\Delta_2} \BIndet b\)}
\DP\qquad
\AXC{\strut\(\addrtyped \Gamma {\Delta_1} a \sigma\)}
\AXC{\(\notaddralive {\Delta_1} a\)}
\AXC{\(\bitvalid \Gamma {\Delta_2} b\)}
\TIC{\(\bitrefine \Gamma f {\Delta_1} {\Delta_2}
	{\BPtrSeg {\PtrSeg {\Ptr a} i}} b\)}
\DP
\end{gather*}

The last two rules allow indeterminate bits \(\BIndet\), as well as pointer
fragment bits \(\BPtrSeg {\PtrSeg {\Ptr a} i}\) belonging to deallocated
storage, to be replaced by arbitrary bits \(b\).

The judgment is lifted to memory trees following the tree structure and
using the following additional rule:
\begin{gather*}
\AXC{\(\elookup t \Gamma = \vec\tau\)}
\AXC{\(\ctreetyped \Gamma \Delta {w_1} {\tau_i}\)}
\AXC{\(\bitrefine \Gamma f {\Delta_1} {\Delta_2}
	{\ctreeflatten {w_1}\,\vec {\ppbit b}_1} {\vec {\ppbit b}_2}\)}
\AXC{\(\vec{\ppbit b}_1\) all \(\BIndet\)}
\def\extraVskip{0pt}
\noLine\QIC{\(\bitsizeof \Gamma {(\TUnion t)} = \bitsizeof \Gamma {\tau_i} +
	\length {\vec{\ppbit b}_1}\)\qquad
	\(\neg\sepunmapped {(\ctreeflatten w_1\,\vec{\ppbit b}_1)}\)}
\def\extraVskip{2pt}
\UIC{\(\ctreerefine \Gamma f {\Delta_1} {\Delta_2}
	{\MUnion t i {w_1} {\vec {\ppbit b}_1}}
	{\MUnionAll t {\vec {\ppbit b}_2}} {\TUnion t}\)}
\DP
\end{gather*}

This rule allows a union that has a specific variant in the source to be
replaced by a union with an unspecified variant in the target.
The direction seems counter intuitive, but keep in mind that
unions with an unspecified variant allow more behaviors.

\begin{lemma}
If \(\ctreerefine \Gamma f {\Delta_1} {\Delta_2} {w_1} {w_2} \tau\), then
\(\ctreetyped \Gamma {\Delta_1} {w_1} \tau\) and
\(\ctreetyped \Gamma {\Delta_2} {w_2} \tau\).
\end{lemma}

This lemma is useful because it removes the need for simultaneous
inductions on both typing and refinement judgments.

We define \(\memrefineShort \Gamma f {m_1} {m_2}\) as
\(\memrefine \Gamma f {\memenvof {m_1}} {\memenvof {m_2}} {m_1} {m_2}\), where
the judgment \(\memrefine \Gamma f {\Delta_1} {\Delta_2} {m_1} {m_2}\) is
defined such that if \(f\,o_1 = \Some {\pair {o_2} {\vec r}}\), then:
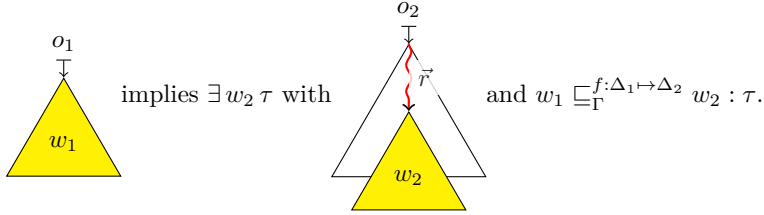
\begin{equation*}
\tikzset{baseline=(current bounding box.center)}
\begin{tikzpicture}
\node[draw,regular polygon,regular polygon sides=3,minimum height=4em,fill=yellow]
	(small) {\(w_1\)};
\draw[<-|] (small.north) -- node[above,pos=1] {\(o_1\)} +(0,0.7em);
\end{tikzpicture}
\textnormal{implies \(\exists\, w_2\,\tau\) with}
\begin{tikzpicture}
\node[draw,regular polygon,regular polygon sides=3,minimum height=7em] (big) {};
\node[draw,regular polygon,regular polygon sides=3,minimum height=4em,fill=yellow]
	(small) at (big.south) {\(w_2\)};
\draw[thick,->,decorate, decoration={snake,amplitude=0.1em}, draw=red]
	(big.north) -- node[fill=white,fill opacity=0.8,text opacity=1,right]
	{\(\vec r\)} (small.north);
\draw[<-|] (big.north) -- node[above,pos=1] {\(o_2\)} +(0,0.7em);
\end{tikzpicture}
\textnormal{and
	\(\ctreerefine \Gamma f {\Delta_1} {\Delta_2} {w_1} {w_2} \tau\).}
\end{equation*}

The above definition makes sure that objects are renamed, and possibly coalesced
into subobjects of a compound object, as described by the renaming function \(f\).

In order to reason about program transformations modularly, we show
that memory refinements can be composed.

\begin{lemma}
\label{lemma:refine_id}
Memory refinements are reflexive for valid memories, that means, if
\(\memvalid \Gamma \Delta m\), then
\(\memrefine \Gamma \meminjid \Delta \Delta m m\) where
\(\meminjid\,o \defined \Some {\pair o \nil}\).
\end{lemma}

\begin{lemma}
\label{lemma:refine_compose}
Memory refinements compose, that means, if
\(\memrefine \Gamma f {\Delta_1} {\Delta_2} {m_1} {m_2}\) and
\(\memrefine \Gamma {f'} {\Delta_2} {\Delta_3} {m_2} {m_3}\), then
\(\memrefine \Gamma {\meminjcompose {f'} f} {\Delta_1} {\Delta_3} {m_1} {m_3}\)
where:
\[(\meminjcompose {f'} f)\,o_1 \defined \begin{cases}
\Some {\pair {o_3} {\vec {r_2}\,\vec r_3}} &
	\textnormal{if}\, f\,o_1 = \Some {\pair {o_2} {\vec r_2}}
	\,\textnormal{and}\, f'\,o_2 = \Some {\pair {o_3} {\vec r_3}} \\
\None & \textnormal{otherwise}
\end{cases}\]
\end{lemma}

All memory operations are preserved by memory refinements.
This property is not only useful for reasoning about program transformations,
but also indicates that the memory interface does not expose internal details
(such as variants of unions) that are unavailable in the memory of a
(concrete) machine.

\begin{lemma}
If \(\memrefine \Gamma f {\Delta_1} {\Delta_2} {m_1} {m_2}\) and
\(\addrrefine \Gamma f {\Delta_1} {\Delta_2} {a_1} {a_2} \tau\) and
\(\memlookup \Gamma {a_1} {m_1} = \Some {v_1}\), then there exists a value
\(v_2\) with
\(\memlookup \Gamma {a_2} {m_2} = \Some {v_2}\) and
\(\valrefine \Gamma f {\Delta_1} {\Delta_2} {v_1} {v_2} \tau\).
\end{lemma}

\begin{lemma}
If \(\memrefine \Gamma f {\Delta_1} {\Delta_2} {m_1} {m_2}\) and
\(\addrrefine \Gamma f {\Delta_1} {\Delta_2} {a_1} {a_2} \tau\) and
\(\valrefine \Gamma f {\Delta_1} {\Delta_2} {v_1} {v_2} \tau\) and
\(\memwritable \Gamma {m_1} {a_1}\), then:
\begin{enumerate}
\item We have \(\memwritable \Gamma {m_2} {a_2}\).
\item We have \(\memrefine \Gamma f {\Delta_1} {\Delta_2}
 {\meminsert \Gamma {a_1} {v_1} {m_1}} {\meminsert \Gamma {a_2} {v_2} {m_2}}\).
\end{enumerate}
\end{lemma}

As shown in Lemma~\ref{lemma:mem_lookup_insert}, storing a value \(v\) in
memory and then retrieving it, does not necessarily yield the same value \(v\).
In case of a byte address, the value may have been stored in padding and
therefore have become indeterminate.
Secondly, it intentionally yields the value \(\freeze v\) in which all pointers
are frozen.
However, the widely used compiler optimization of constant propagation,
which substitutes values of known constants at compile time, is still valid in
our memory model.

\begin{lemma}
If \(\valtyped \Gamma \Delta v \tau\), then
\(\valrefineID \Gamma \Delta {\freeze v} v \tau\).
\end{lemma}

\begin{theorem}[Constant propagation]
\label{theorem:const_prop}
If \(\memvalid \Gamma \Delta m\) and
\(\addrtyped \Gamma \Delta a \tau\) and
\(\valtyped \Gamma \Delta v \tau\) and
\(\memwritable \Gamma a m\), then there exists a value \(v'\)
with:
\[
	\memlookup \Gamma a {\meminsert \Gamma a v m} = \Some v'
	\quad\textnormal{and}\quad
	\valrefineID \Gamma \Delta {v'} v \tau.
\]
\end{theorem}

Copying an object \(w\) by an assignment results in it being converted
to a value \(\toval \Gamma w\) and back.
This conversion makes invalid representations of base values indeterminate.
Copying an object \(w\) byte-wise results in it being converted to bits
\(\ctreeflatten w\) and back.
This conversion makes all variants of unions unspecified.
The following theorem shows that a copy by assignment can be
transformed into a byte-wise copy.

\begin{theorem}[Memcpy]
If \(\ctreetyped \Gamma \Delta w \tau\), then:
\[
	\ctreerefineID \Gamma \Delta {\ofval \Gamma
		{(\fst {\ctreeflatten w})} {(\toval \Gamma w)}}
	{w \ctreerefineIDSym \Gamma \Delta
	\ctreeunflatten \Gamma \tau {\ctreeflatten w}} \tau.
\]
\end{theorem}

Unused reads cannot be removed unconditionally in the \CHtwoO{} memory model
because these have side-effects in the form of uses of the
\(\memforceSym \Gamma\) operation that updates effective types.
We show that uses of \(\memforceSym \Gamma\) can be removed for frozen
addresses.

\begin{theorem}
If \(\memvalid \Gamma \Delta m\) and \(\memlookup \Gamma a m \neq \None\) and
\(\frozen a\), then \(\memrefineID \Gamma \Delta {\memforce \Gamma a m} m\).
\end{theorem}

\subsection{Reasoning about disjointness}
\label{section:disjointness}

In order to prove soundness of the \CHtwoO{} axiomatic semantics, we were
often in need to to reason about preservation of disjointness under memory
operations~\cite{kre:15:phd}.
This section describes some machinery to ease reasoning about disjointness.
We show that our machinery, as originally developed
in~\cite{kre:14:popl}, extends to any separation algebra.

\begin{definition}
\emph{Disjointness of a list \(\vec x\)}, notation
\(\sepdisjointlist {\vec x}\), is defined as:
\begin{enumerate}
\item \(\sepdisjointlist \nil\)
\item If \(\sepdisjointlist {\vec x}\) and
	\(\sepdisjoint x {\sepunionlist {\vec x}}\),
	then \(\sepdisjointlist {(\cons x {\vec x})}\)
\end{enumerate}
\end{definition}

Notice that \(\sepdisjointlist {\vec x}\) is stronger than having
\(\sepdisjoint {x_i} {x_j}\) for each \(i \neq j\).
For example, using fractional permissions, we do not have
\(\sepdisjointlist {\listlit {0.5,\,0.5,\,0.5}}\) whereas
\(\sepdisjoint {0.5} {0.5}\) clearly holds.
Using disjointness of lists we can for example state the associativity
law (law~\ref{item:sep_associative} of
Definition~\ref{definition:separation_algebra}) in a symmetric way:

\begin{fact}
If \(\sepdisjointlist {(x\;y\;z)}\), then
\(\sepunion x {(\sepunion y z)} = \sepunion {(\sepunion x y)} z\).
\end{fact}

We define a relation \(\sepdisjointequiv {\vec x_1} {\vec x_2}\) that 
expresses that \(\vec x_1\) and \(\vec x_2\) behave equivalently with
respect to disjointness.

\begin{definition}
\label{definition:sep_disjoint_equiv}
\emph{Equivalence of lists \(\vec x_1\) and \(\vec x_2\) with respect to
disjointness}, notation \(\sepdisjointequiv {\vec x_1} {\vec x_2}\), is defined
as:
\begin{flalign*}
\sepdisjointle {\vec x_1} {\vec x_2} \defined{}&
	\forall x \wsdot \sepdisjointlist{(\cons x {\vec x_1})} \impl
	\sepdisjointlist{(\cons x {\vec x_2})} \\
\sepdisjointequiv {\vec x_1} {\vec x_2} \defined{}&
	\sepdisjointle {\vec x_1} {\vec x_2} \land
	\sepdisjointle {\vec x_2} {\vec x_1}
\end{flalign*}
\end{definition}

It is straightforward to show that \(\sepdisjointleSym\) is reflexive and
transitive, is respected by concatenation of lists, and is preserved by
list containment.
Hence, \(\sepdisjointequivSym\) is an equivalence relation, a congruence with
respect to concatenation of lists, and is preserved by permutations.
The following results (on arbitrary separation algebras) allow us to reason
algebraically about disjointness.

\begin{fact}
\label{fact:disjoint_le_disjoint}
If \(\sepdisjointle {\vec x_1} {\vec x_2}\), then
\(\sepdisjointlist {\vec x_1}\) implies \(\sepdisjointlist {\vec x_2}\).
\end{fact}

\begin{fact}
If \(\sepdisjointequiv {\vec x_1} {\vec x_2}\), then
\(\sepdisjointlist {\vec x_1}\) iff \(\sepdisjointlist {\vec x_2}\).
\end{fact}

\begin{theorem}
\label{theorem:sep_disjoint}
We have the following algebraic properties:
\begin{equation*}
\begin{aligned}
\sepempty \sepdisjointequivSym{}& \nil \\
\sepunion {x_1} {x_2} \sepdisjointequivSym{}&
	x_1\, x_2
	&&\textnormal{provided that \(\sepdisjoint {x_1} {x_2}\)} \\
\sepunionlist {\vec x} \sepdisjointequivSym{}&
	\vec x
	&&\textnormal{provided that \(\sepdisjointlist {\vec x}\)} \\
x_2 \sepdisjointequivSym{}&
	x_1\,(\sepdifference {x_2} {x_1})
	&&\textnormal{provided that \(\sepsubseteq {x_1} {x_2}\)}
\end{aligned}
\end{equation*}
\end{theorem}

In Section~\ref{section:memory_separation} we show that we have similar
properties as the above for the specific operations of our memory model.

\subsection{The memory as a separation algebra}
\label{section:memory_separation}

We show that the \CHtwoO{} memory model is a separation algebra,
and that the separation algebra operations interact
appropriately with the memory operations that we have defined in
Section~\ref{section:memory}.

In order to define the separation algebra relations and operations on memories,
we first define these on memory trees.
Memory trees do not form a separation algebra themselves due to the
absence of a unique \(\sepempty\) element (memory trees have a distinct
identity element \(\ctreenew \Gamma \tau\) for each type \(\tau\), see
Definition~\ref{definition:ctree_new}).
The separation algebra of memories is then defined by lifting the definitions
on memory trees to memories (which are basically finite functions to
memory trees).

\begin{definition}
\label{definition:ctree_valid}
The predicate \(\sepvalidSym : \mtree \to \Prop\) is inductively defined as:
\renewcommand{\ScoreOverhang}{0.1em}
\begin{gather*}
\AXC{\(\sepvalid {\vec {\ppbit b}}\)}
\UIC{\(\sepvalid {(\MBase {\btype \tau} {\vec {\ppbit b}})}\)}
\DP\qquad
\AXC{\(\sepvalid {\vec w}\)}
\UIC{\(\sepvalid {(\MArray \tau {\vec w})}\)}
\DP\qquad
\AXC{\(\sepvalid {\vec w}\)}
\AXC{\(\sepvalid {\vv {\vec {\ppbit b}}}\)}
\BIC{\(\sepvalid {(\MStruct t {\vv{w\vec {\ppbit b}}})}\)}
\DP \\[0.5em]
\AXC{\strut\(\sepvalid w\)}
\AXC{\(\sepvalid {\vec {\ppbit b}}\)}
\AXC{\(\neg\sepunmapped {(\ctreeflatten w\,\vec {\ppbit b}})\)}
\TIC{\(\sepvalid {(\MUnion t i w {\vec {\ppbit b}})}\)}
\DP\qquad
\AXC{\strut\(\sepvalid {\vec {\ppbit b}}\)}
\UIC{\(\sepvalid {(\MUnionAll t {\vec {\ppbit b}})}\)}
\DP
\end{gather*}
\end{definition}

\begin{fact}
If \(\ctreetyped \Gamma \Delta w \tau\), then \(\sepvalid w\).
\end{fact}

The \(\sepvalidSym\) predicate specifies the subset of memory trees on which
the separation algebra structure is defined.
The definition basically lifts the \(\sepvalidSym\) predicate from the
leaves to the trees.
The side-condition \(\neg\sepunmapped {(\ctreeflatten w\,\vec{\ppbit b})}\) on
\(\MUnion t i w {\vec {\ppbit b}}\) memory trees ensures canonicity, unions
whose permissions are unmapped cannot be accessed and are thus kept in
unspecified variant.
Unmapped unions \(\MUnionAll t {\vec {\ppbit b}}\) can be combined with
other unions using \(\sepunionSym\).
The rationale for doing so will become clear in the context of the
separation logic in the author's PhD thesis~\cite{kre:15:phd}.

\begin{definition}
The relation \({\sepdisjointSym} : \mtree \to \mtree \to \Prop\) is inductively
defined as:
\renewcommand{\ScoreOverhang}{0.2em}
\begin{gather*}
\AXC{\(\sepdisjoint {\vec {\ppbit b}_1} {\vec {\ppbit b}_2}\)}
\UIC{\(\sepdisjoint {\MBase {\btype \tau} {\vec {\ppbit b}_1}}
	{\MBase {\btype \tau} {\vec {\ppbit b}_2}}\)}
\DP \qquad
\AXC{\(\sepdisjoint {\vec w_1} {\vec w_2}\)}
\UIC{\(\sepdisjoint {\MArray \tau {\vec w_1}}
	{\MArray \tau {\vec w_2}}\)}
\DP \qquad
\AXC{\(\sepdisjoint {\vec w_1} {\vec w_2}\)}
\AXC{\(\sepdisjoint {\vv{\vec {\ppbit b}_1}} {\vv{\vec {\ppbit b}_2}}\)}
\BIC{\(\sepdisjoint {\MStruct t {\vv{w_1\vec {\ppbit b}_1}}}
	{\MStruct t {\vv{w_2\vec {\ppbit b}_2}}}\)}
\DP
\\[0.5em]
\AXC{\(\sepdisjoint {w_1} {w_2}\)}
\AXC{\(\sepdisjoint {\vec {\ppbit b}_1} {\vec {\ppbit b}_2}\)}
\AXC{\(\neg\sepunmapped {(\ctreeflatten{w_1}\;\vec {\ppbit b}_1)}\)}
\AXC{\(\neg\sepunmapped {(\ctreeflatten{w_2}\;\vec {\ppbit b}_2)}\)}
\QIC{\(\sepdisjoint {\MUnion t i {w_1} {\vec {\ppbit b}_1}}
	{\MUnion t i {w_2} {\vec {\ppbit b}_2}}\)}
\DP
\\[0.5em]
\AXC{\(\strut\sepdisjoint {\vec {\ppbit b}_1} {\vec {\ppbit b}_2}\)}
\UIC{\(\sepdisjoint {\MUnionAll t {\vec {\ppbit b}_1}}
	{\MUnionAll t {\vec {\ppbit b}_2}}\)}
\DP \qquad
\renewcommand{\defaultHypSeparation}{\hskip 0.3em}
\AXC{\(\strut\sepdisjoint {\ctreeflatten{w_1}\; \vec {\ppbit b}_1}
	{\vec {\ppbit b}_2}\)}
\AXC{\(\sepvalid {w_1}\)}
\AXC{\(\neg\sepunmapped {(\ctreeflatten{w_1}\;\vec {\ppbit b}_1)}\)}
\AXC{\(\sepunmapped {\vec {\ppbit b}_2}\)}
\QIC{\(\sepdisjoint {\MUnion t i {w_1} {\vec {\ppbit b}_1}}
	{\MUnionAll t {\vec {\ppbit b}_2}}\)}
\DP
\\[0.5em]
\AXC{\(\strut\sepdisjoint {\vec {\ppbit b}_1}
	{\ctreeflatten{w_2}\; \vec {\ppbit b}_2}\)}
\AXC{\(\sepvalid {w_2}\)}
\AXC{\(\sepunmapped {\vec {\ppbit b}_1}\)}
\AXC{\(\neg\sepunmapped {(\ctreeflatten{w_2}\;\vec {\ppbit b}_2)}\)}
\QIC{\(\sepdisjoint {\MUnionAll t {\vec {\ppbit b}_1}}
	{\MUnion t i {w_2} {\vec {\ppbit b}_2}}\)}
\DP
\end{gather*}
\end{definition}

\begin{definition}
\label{definition:ctree_union}
The operation \({\sepunionSym} : \mtree \to \mtree \to \mtree\) is defined as:
\begin{equation*}
\begin{aligned}
\sepunion {\MBase {\btype \tau} {\vec {\ppbit b}_1}&{}}
	{\MBase {\btype \tau} {\vec {\ppbit b}_2}} &\defined{}&
	\MBase {\btype \tau} {(\sepunion {\vec {\ppbit b}_1} {\vec {\ppbit b}_2})} \\
\sepunion {\MArray \tau {\vec w_1}&{}}
	{\MArray \tau {\vec w_2}} &\defined{}&
	\MArray \tau {(\sepunion {\vec w_1} {\vec w_2})} \\
\sepunion {\MStruct t {\vv{w_1\vec {\ppbit b}_1}}&{}}
	{\MStruct t {\vv{w_2\vec {\ppbit b}_2}}} &\defined{}&
	\MStruct t {(\sepunion {\vv{w_1\vec {\ppbit b}_1}}
		{\vv{w_2\vec {\ppbit b}_2}})} \\
\sepunion {\MUnion t i {w_1} {\vec {\ppbit b}_1}&{}}
	{\MUnion t i {w_2} {\vec {\ppbit b}_2}} &\defined{}&
	\MUnion t i {\sepunion {w_1} {w_2}}
		{\sepunion {\vec {\ppbit b}_1} {\vec {\ppbit b}_2}} \\
\sepunion {\MUnionAll t {\vec {\ppbit b}_1}&{}}
	{\MUnionAll t {\vec {\ppbit b}_1}} &\defined{}&
	\MUnionAll t {(\sepunion {\vec {\ppbit b}_1} {\vec {\ppbit b}_2})} \\
\sepunion {\MUnion t i {w_1} {\vec {\ppbit b}_1}&{}}
	{\MUnionAll t {\vec {\ppbit b}_2}} &\defined{}&
	\ctreemergeInfix \sepunionSym
		{\MUnion t i {w_1} {\vec {\ppbit b}_1}} {\vec {\ppbit b}_2} \\
\sepunion {\MUnionAll t {\vec {\ppbit b}_1}&{}} 
	{\MUnion t i {w_2} {\vec {\ppbit b}_2}} &\defined{}&
	\ctreemergeInfix \sepunionSym
		{\MUnion t i {w_2} {\vec {\ppbit b}_2}} {\vec {\ppbit b}_1}
\end{aligned}
\end{equation*}
In the last two clauses, \(\ctreemergeInfix \sepunionSym w \vec {\ppbit b}\)
is a modified version of the memory tree \(w\) in which the elements on the
leaves of \(w\) are zipped with \(\vec {\ppbit b}\) using the \(\sepunionSym\)
operation on permission annotated bits
(see Definitions~\ref{definition:ctree_merge}
and~\ref{definition:sep_tagged}).
\end{definition}

The definitions of \(\sepvalidSym\), \(\sepdisjointSym\) and
\(\sepunionSym\) on memory trees satisfy all laws of a separation algebra (see
Definition~\ref{definition:separation_algebra}) apart from those involving
\(\sepempty\).
We prove the cancellation law explicitly since it involves the aforementioned
side-conditions on unions.

\begin{lemma}
If \(\sepdisjoint {w_3} {w_1}\) and \(\sepdisjoint {w_3} {w_2}\) then:
\[
	\sepunion {w_3} {w_1} = \sepunion {w_1} {w_2}
	\quad\textnormal{implies}\quad
	w_1 = w_2.
\]
\end{lemma}

\begin{proof}
By induction on the derivations \(\sepdisjoint {w_3} {w_1}\) and
\(\sepdisjoint {w_3} {w_2}\).
We consider one case:
\begin{equation*}
\renewcommand{\defaultHypSeparation}{\hskip.4em}
\AXC{\(\sepdisjoint {\MUnion t i {w_3} {\vec {\ppbit b}_3}}
	{\MUnion t i {w_1} {\vec {\ppbit b}_1}}\)}
\AXC{\(\sepdisjoint {\MUnion t i {w_3} {\vec {\ppbit b}_3}}
	{\MUnionAll t {\vec {\ppbit b}_2}}\)}
\noLine
\BIC{\(\sepunion {\MUnion t i {w_3} {\vec {\ppbit b}_3}}
	{\MUnion t i {w_1} {\vec {\ppbit b}_1}}
	= \sepunion {\MUnion t i {w_3} {\vec {\ppbit b}_3}}
	{\MUnionAll t {\vec {\ppbit b}_2}}\)}
\UIC{\(\MUnion t i {w_1} {\vec {\ppbit b}_1}
	= \MUnionAll t {\vec {\ppbit b}_2}\)}
\DP
\end{equation*}
Here, we have \(\sepunion {\ctreeflatten {w_3}\;\vec {\ppbit b}_3}
{\ctreeflatten {w_1}\;\vec {\ppbit b}_1}
= \sepunion {\ctreeflatten {w_3}\;\vec {\ppbit b}_3} {\vec {\ppbit b}_2}\) by
assumption, and therefore
\(\ctreeflatten {w_1}\;\vec {\ppbit b}_1 = \vec {\ppbit b}_2\) by the
cancellation law of a separation algebra.
However, by assumption we also have
\(\neg\sepunmapped {(\ctreeflatten {w_1}\;\vec {\ppbit b}_1)}\)
and 
\(\sepunmapped {\vec {\ppbit b}_2}\), which contradicts
\(\ctreeflatten {w_1}\;\vec {\ppbit b}_1 = \vec {\ppbit b}_2\).
\end{proof}

\begin{definition}
\label{definition:cmap_separation}
The \emph{separation algebra of memories} is defined as:
\begin{flalign*}
\sepvalid m \defined{}& \forall o\,w\,\mu \wsdot
	\mlookup o m = \Some {\pair w \mu}
	\impl (\sepvalid w \textnormal{ and not } \ctreeflatten w
	\textnormal{ all } \pair \sepempty \BIndet) \\\
\sepdisjoint {m_1} {m_2} \defined{}& \forall o \wsdot P\,m_1\,m_2\,o \\ 
\sepunion {m_1} {m_2} \defined{}& \lambda o \wsdot f\,m_1\,m_2\,o
\end{flalign*}
\(P : \mem \to\mem\to\memindex\to\Prop\) and
\(f : \mem \to\mem\to\memindex\to\option\mtree\) are defined
by case analysis on \(\mlookup o {m_1}\) and \(\mlookup o {m_2}\):
\begin{equation*}
\begin{array}{l|l|l|l}
\mlookup o {m_1} & \mlookup o {m_2} &
	P\,m_1\,m_2\,o & f\,m_1\,m_2\,o \\ \hline \hline
\Some {\pair {w_1} \mu} & \Some {\pair {w_1} \mu}
	& \sepdisjoint {w_1} {w_2},
	\textnormal{ not } \ctreeflatten {w_1}
		\textnormal{ all } \pair \sepempty \BIndet
	\textnormal{ and not }
	\ctreeflatten {w_2} \textnormal{ all } \pair \sepempty \BIndet
	& \pair {\sepunion {w_1} {w_2}} \mu \\
\Some {\pair {w_1} \mu} & \None
	& \sepvalid {w_1} \textnormal{ and not } \ctreeflatten {w_1}
		\textnormal{ all } \pair \sepempty \BIndet
	& \Some {\pair {w_1} \mu} \\
\None & \Some {\pair {w_2} \mu}
	& \sepvalid {w_2} \textnormal{ and not } \ctreeflatten {w_2}
		\textnormal{ all } \pair \sepempty \BIndet
	& \Some {\pair {w_2} \mu} \\
\Some {\tau_1} & \None & \True & \Some {\tau_1} \\
\None & \Some{\tau_2} & \True & \Some {\tau_2} \\
\None & \None & \True & \None \\
\multicolumn{2}{c|}{\textnormal{otherwise}} & \False & \None
\end{array}
\end{equation*}
The definitions of the omitted relations and operations are as expected.
\end{definition}

The emptiness conditions ensure canonicity.
Objects that solely consist of indeterminate bits with \(\sepempty\) permission
are meaningless and should not be kept at all.
These conditions are needed for cancellativity.

\begin{fact}
If \(\memvalid \Gamma \Delta m\), then \(\sepvalid m\).
\end{fact}

\begin{lemma}
\label{lemma:cmap_union_valid}
If \(\sepdisjoint {m_1} {m_2}\), then:
\[
	\memvalid \Gamma \Delta {\sepunion {m_1} {m_2}}
	\qquad\textnormal{iff}\qquad
	\memvalid \Gamma \Delta {m_1}\textnormal{ and }
	\memvalid \Gamma \Delta {m_2}.
\]
\end{lemma}

Notice that the memory typing environment \(\Delta\) is not subdivided among
\(m_1\) and \(m_2\).
Consider the memory state corresponding to \lstinline|int x = 10, *p = &x|:
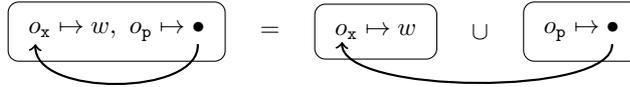
\begin{equation*}
\begin{tikzpicture}[
	start chain=going right,
	node distance=1em,
	box/.style={rounded corners,draw=black,inner sep=0.8em},
]
\node[box,on chain] (o12) {\(o_{\mathtt x} \mapsto w,\ o_{\mathtt p} \mapsto \bullet\)};
\node[on chain] {\(=\)};
\node[box,on chain] (o1) {\(o_{\mathtt x} \mapsto w\)};
\node[on chain] {\(\sepunionSym\)};
\node[box,on chain] (o2) {\(o_{\mathtt p} \mapsto \bullet\)};
\draw[thick,->] ($(o12.south east) + (-1.1em,0.7em)$) ..
	controls ($(o12.south east) + (-1em,-1.4em)$)
	and ($(o12.south west) + (1em,-1.4em)$) ..
	($(o12.south west) + (1.1em,0.7em)$);
\draw[thick,->] ($(o2.south east) + (-1.1em,0.7em)$) ..
	controls ($(o2.south east) + (-1em,-1.4em)$)
	and ($(o1.south west) + (1em,-1.4em)$) ..
	($(o1.south west) + (1.1em,0.7em)$);
\end{tikzpicture}
\end{equation*}

Here, \(w\) is the memory tree that represents the integer value 10.
The pointer on the right hand side is well-typed in the memory
environment 
\(\memenvof{\vphantom{X}o_{\mathtt x} \mapsto w,\ o_{\mathtt p} \mapsto \bullet}\)
of the whole memory, but not in
\(\memenvof{\vphantom{X}o_{\mathtt p} \mapsto \bullet}\).

We prove some essential properties about the interaction between
the separation algebra operations and the memory operations.
These properties will be used in the soundness proof of the separation logic
in the author's PhD thesis~\cite{kre:15:phd}.

\begin{lemma}[Preservation of lookups]
If \(\memvalid \Gamma \Delta {m_1}\) and \(\sepsubseteq {m_1} {m_2}\),
then:
\begin{align*}
\memlookup \Gamma a {m_1} = \Some v
	&\quad\textnormal{ implies }\quad
	\memlookup \Gamma a {m_2} = \Some v \\
\memwritable \Gamma a {m_1}
	&\quad\textnormal{ implies }\quad
	\memwritable \Gamma a {m_2}
\end{align*}
The relation \(\sepsubseteqSym\) is part of a separation algebra, see
Definition~\ref{definition:separation_algebra}.
We have \(\sepsubseteq {m_1} {m_2}\) iff there is an \(m_3\) with
\(\sepdisjoint {m_1} {m_3}\) and \(m_2 = \sepunion {m_1} {m_3}\).
\end{lemma}

\begin{lemma}[Preservation of disjointness]
If \(\memvalid \Gamma \Delta m\) then:
\begin{align*}
\sepdisjointle m {&\memforce \Gamma a m}
	&&\textnormal{if \(\addrtyped \Gamma \Delta a \tau\)
	and \(\memlookup \Gamma a m \neq \None\)}\\
\sepdisjointle m {&\meminsert \Gamma a v m}
	&&\textnormal{if \(\addrtyped \Gamma \Delta a \tau\)
	and \(\memwritable \Gamma a m\)}\\
\sepdisjointle m {&\memlock \Gamma a m}
	&&\textnormal{if \(\addrtyped \Gamma \Delta a \tau\)
	and \(\memwritable \Gamma a m\)}\\
\sepdisjointle m {&\memunlock \Omega m}
	&&\textnormal{if \(\Omega \subseteq \memlocks m\)}
\end{align*}
The relation \(\sepdisjointleSym\) is defined in
Definition~\ref{definition:sep_disjoint_equiv}.
If \(\sepdisjointle m {m'}\), then each memory that is disjoint to \(m\) is
also disjoint to \(m'\).
\end{lemma}

As a corollary of the above lemma and Fact~\ref{fact:disjoint_le_disjoint}
we obtain that \(\sepdisjoint {m_1} {m_2}\) implies disjointness of the
memory operations:
\begin{align*}
\sepdisjoint {\memforce \Gamma a {m_1}} {{}&m_2} &
\sepdisjoint {\meminsert \Gamma a v {m_1}} {{}&m_2} \\
\sepdisjoint {\memlock \Gamma a {m_1}} {{}&m_2} &
\sepdisjoint {\memunlock \Omega {m_1}} {{}&m_2}
\end{align*}

\begin{lemma}[Unions distribute]
\label{lemma:mem_insert_union}
If \(\memvalid \Gamma \Delta m\) and \(\sepdisjoint {m_1} {m_2}\) then:
\begin{align*}
\memforce \Gamma a {(\sepunion {m_1} {m_2})} ={}&
	\sepunion {\memforce \Gamma a {m_1}} {m_2}
	&&\textnormal{if \(\addrtyped \Gamma \Delta a \tau\)
	and \(\memlookup \Gamma a {m_1} \neq \None\)}\\
\meminsert \Gamma a v {(\sepunion {m_1} {m_2})} ={}&
	\sepunion {\meminsert \Gamma a v {m_1}} {m_2}
	&&\textnormal{if \(\addrtyped \Gamma \Delta {\{a,v\}} \tau\)
	and \(\memwritable \Gamma a {m_1}\)}\\
\memlock \Gamma a {(\sepunion {m_1} {m_2})} ={}&
	\sepunion {\memlock \Gamma a {m_1}} {m_2}
	&&\textnormal{if \(\addrtyped \Gamma \Delta a \tau\)
	and \(\memwritable \Gamma a {m_1}\)}\\
\memunlock \Omega {(\sepunion {m_1} {m_2})} ={}&
	\sepunion {\memunlock \Omega {m_1}} {m_2}
	&&\textnormal{if \(\Omega \subseteq \memlocks {m_1}\)}
\end{align*}
\end{lemma}

Memory trees and memories can be generalized to contain elements of an arbitrary
separation algebra as leaves instead of just permission annotated
bits~\cite{kre:14:vstte}.
These generalized memories form a functor that lifts the separation algebra
structure on the leaves to entire trees.
We have taken this approach in the \Coq{} development, but for brevity's sake,
we have refrained from doing so in this paper.

\section{Formalization in \Coq}
\label{section:coq}

Real-world programming language have a large number of features that require
large formal descriptions.
As this paper has shown, the \C{} programming language is not different in
this regard.
On top of that, the \C{} semantics is very subtle due to an abundance of
delicate corner cases.
Designing a semantics for \C{} and proving properties about such a semantics
therefore inevitably requires computer support.

For these reasons, we have used the \Coq{} proof assistant~\cite{coq:15} to
formalize all definitions and theorems in this paper.
Although \Coq{} does not guarantee the absence of mistakes in our definitions,
it provides a rigorous set of checks on our definitions.
Already \Coq's type checking of definitions provides an effective sanity check.
On top of that, we have used \Coq{} to prove all metatheoretical results
stated in this paper.
Last but not least, using \Coq's program extraction facility we have extracted
an exploration tool to test our memory model on small example
programs~\cite{kre:wie:15,kre:15:phd}.

\subsection{Overloaded typing judgments}

Type classes are used to overload notations for typing judgments (we have 25
different typing judgments).
The class \lstinline|Valid| is used for judgments without a type, such as
\(\envvalid \Gamma\) and \(\memvalid \Gamma \Delta m\).

\begin{lstlisting}[language=coq]
Class Valid (E A : Type) := valid: E → A → Prop.
Notation "✓{ Γ }" := (valid Γ).
Notation "✓{ Γ }*" := (Forall (✓{Γ})).
\end{lstlisting}

We use product types to represent judgments with multiple environments such as
\(\memvalid \Gamma \Delta m\).
The notation \lstinline|✓{Γ}*| is used to lift the judgment to lists.
The class \lstinline|Typed| is used for judgments such as
\(\valtyped \Gamma \Delta v \tau\) and
\(\exprtyped \Gamma \Delta {\vec\tau} e {\lr\tau}\).

\begin{lstlisting}[language=coq]
Class Typed (E T V : Type) := typed: E → V → T → Prop.
Notation "Γ ⊢ v : τ" := (typed Γ v τ).
Notation "Γ ⊢* vs :* τs" := (Forall2 (typed Γ) vs τs).
\end{lstlisting}

\subsection{Implementation defined behavior}

Type classes are used to parametrize the whole \Coq{} development by implementation
defined parameters such as integer sizes.
For example, Lemma~\ref{lemma:to_of_val} looks like:

\noindent
\begin{minipage}{\linewidth}
\begin{lstlisting}[language=coq]
Lemma to_of_val `{EnvSpec K} Γ Δ γs v τ :
  ✓ Γ → (Γ,Δ) ⊢ v : τ → length γs = bit_size_of Γ τ →
  to_val Γ (of_val Γ γs v) = freeze true v.
\end{lstlisting}
\end{minipage}

The parameter \lstinline|EnvSpec$\;K$| is a type class describing an
implementation environment with ranks \(K\)
(Definition~\ref{definition:env_spec}).
Just as in this paper, the type \(K\) of integer ranks is a parameter of the
inductive definition of types (see Definition~\ref{definition:int_types})
and is propagated through all syntax.

\begin{lstlisting}[language=coq]
Inductive signedness := Signed | Unsigned.
Inductive int_type (K: Set) := IntType {$\,$sign: signedness; rank: K$\,$}.
\end{lstlisting}

The definition of the type class \lstinline|EnvSpec| is based on the approach
of Spitters and van der Weegen approach~\cite{spi:wee:11}.
We have a separate class \lstinline|Env| for the operations that is an implicit
parameter of the whole class and all lemmas.

\begin{lstlisting}[language=coq]
Class Env (K: Set) := {
  env_type_env :> IntEnv K;
  size_of : env K → type K → nat;
  align_of : env K → type K → nat;
  field_sizes : env K → list (type K) → list nat
}.
Class EnvSpec (K: Set) `{Env K} := {
  int_env_spec :>> IntEnvSpec K;
  size_of_ptr_ne_0 Γ τp : size_of Γ (τp.*) ≠ 0;
  size_of_int Γ τi : size_of Γ (intT τi) = rank_size (rank τi);
  ...
}.
\end{lstlisting}

\subsection{Partial functions}
\label{section:coq_partial}

Although many operations in \CHtwoO{} are partial, we have formalized many such
operations as total functions that assign an appropriate default value.
We followed the approach presented in
Section~\ref{section:separation_algebra} where operations are combined with a
\emph{validity predicate} that describes in which case they may be used.
For example, part (2) of Lemma~\ref{lemma:mem_insert_union} is stated
in the \Coq{} development as follows:

\begin{lstlisting}[language=coq]
Lemma mem_insert_union `{EnvSpec K} Γ Δ m1 m2 a1 v1 τ1 :
  ✓ Γ → ✓{Γ,Δ} m1 → m1 ⊥ m2 →
  (Γ,Δ) ⊢ a1 : TType τ1 → mem_writable Γ a1 m1 → (Γ,Δ) ⊢ v1 : τ1 →
  <[a1:=v1]{Γ}>(m1 ∪ m2) = <[a1:=v1]{Γ}>m1 ∪ m2.
\end{lstlisting}

Here, \lstinline|m1 ⊥${}$ m2| is the side-condition of \lstinline|m1 ∪${}$ m2|,
and \lstinline|mem_writable Γ${}$ a1 m1| the side-condition of
\lstinline|<[a1:=v1]{Γ}>m1|.
Alternatives approaches include using the option monad or dependent types, but
our approach proved more convenient.
In particular, since most validy predicates are given by an inductive
definition, various proofs could be done by induction on the structure
of the validy predicate.
The cases one has to consider correspond exactly to the domain of
the partial function.

Admissible side-conditions, such as in the above example
\lstinline|<[a1:=v1]{Γ}>m1 ⊥${}$ m2|
and \lstinline|mem_writable Γ${}$ a1 (m1 ∪${}$ m2)|, do not have to be stated
explicitly and follow from the side-conditions that are already there.
By avoiding the need to state admissible side-conditions, we avoid a blow-up in
the number of side-conditions of many lemmas.
We thus reduce the proof effort needed to use such a lemma.

\subsection{Automation}

The proof style deployed in the \CHtwoO{} development combines interactive
proofs with automated proofs.
In this section we describe some tactics and forms of proof
automation deployed in the \CHtwoO{} development.

\paragraph{Small inversions.}
\Coq's \lstinline|inversion| tactic has two serious shortcomings on inductively
defined predicates with many constructors.
It is rather slow and its way of controlling of names for variables and
hypotheses is deficient.
Hence, we are often using the technique of small inversions by Monin and
Shi~\cite{mon:shi:13} that improves on both shortcomings.

\paragraph{Solving disjointness.}
We have used \Coq's setoid machinery~\cite{soz:10} to enable rewriting
using the relations \(\sepdisjointleSym\) and \(\sepdisjointequivSym\)
(Definition~\ref{definition:sep_disjoint_equiv}).
Using this machinery, we have implemented a tactic that automatically solves
entailments of the form:
\[
H_0 : \sepdisjointlist {\vec x_0},\ \dotsc,\
H_n : \sepdisjointlist {\vec x_{n-1}}
\quad \vdash\quad  \sepdisjointlist {\vec x}
\]

\noindent
where \(\vec x\) and \(\vec x_i\) (for \(i < n\)) are arbitrary \Coq{}
expressions built from \(\sepempty\), \(\sepunionSym\) and \(\sepunionlistSym\).
This tactic works roughly as follows:

\begin{enumerate}
\item Simplify hypotheses using Theorem~\ref{theorem:sep_disjoint}.
\item Solve side-conditions by simplification using
	Theorem~\ref{theorem:sep_disjoint} and a solver for list containment
	(implemented by reflection).
\item Repeat these steps until no further simplification is possible.
\item Finally,	solve the goal by simplification using
	Theorem~\ref{theorem:sep_disjoint} and list containment.
\end{enumerate}

This tactic is not implemented using reflection, but that is something we intend
to do in future work to improve its performance.

\paragraph{First-order logic.}
Many side-conditions we have encountered involve simple entailments
of first-order logic such as distributing logical quantifiers combined with
some propositional reasoning.
\Coq{} does not provide a solver for first-order logic apart from the
\lstinline|firstorder| tactic whose performance is already insufficient on
small goals.

We have used \Ltac{} to implemented an ad-hoc solver called
\lstinline|naive_solver|, which performs a simple breath-first search proof
search.
Although this tactic is inherently incomplete and suffers from some limitations,
it turned out to be sufficient to solve many uninteresting
side-conditions (without the need for classical axioms).

\subsection{Overview of the \Coq{} development}

The \Coq{} development of the memory model, which is
entirely constructive and axiom free, consists of the following parts:
\begin{center}
\begin{tabular}{l|l|r}
\textbf{Component} & \textbf{Sections} & \textbf{LOC} \\ \hline \hline
Support library (lists, finite sets, finite maps, \etc)
	& Section~\ref{section:notations} & 12\,524 \\
Types \& Integers & Section~\ref{section:types} & 1\,928 \\
Permissions \& separation algebras & Section~\ref{section:permissions} & 1\,811\\
Memory model & Section~\ref{section:memory} & 8\,736 \\
Refinements & Section~\ref{section:refinements} & 4\,046 \\
Memory as separation algebra
	& Section~\ref{section:memory_separation} & 3\,844 \\ \hline
\multicolumn{2}{l|}{\emph{Total}} & 32\,889
\end{tabular}
\end{center}

\section{Related work}
\label{section:related}

The idea of using a memory model based on trees instead of arrays of plain bits,
and the idea of using pointers based on paths instead of offsets, has already
been used for object oriented languages.
It goes back at least to Rossie and Friedman~\cite{ros:fri:95}, and has been
used by Ramananandro \etal~\cite{ram:dos:ler:11} for \Cplusplus{}.
Furthermore, many researchers have considered connections between unstructured and
structured views of data in \C~\cite{tuc:kle:nor:07,%
coh:mos:tob:sch:09,aff:sak:14,gre:lim:adr:kle:14} in the context of program
logics.

However, a memory model that combines an abstract tree based structure with
low-level object representations in terms of bytes has not been
explored before.
In this section we will describe other formalization of the \C{} memory model.

\paragraph{Norrish (1998)}
Norrish has formalized a significant fragment of the \Ceightynine{} standard
using the proof assistant \HOLfour~\cite{nor:98,nor:99}.
He was the first to describe non-determinism and sequence points formally.
Our treatment of these features has partly been based on his work.
Norrish's formalization of the \C{} type system has some similarities with
our type system: he has also omitted features that can be desugared and has
proven type preservation.

Contrary to our work, Norrish has used an unstructured memory model based on
sequences of bytes.
Since he has considered the \Ceightynine{} standard in which effective types
(and similar notions) were not introduced yet, his choice is appropriate.
For \Cninetynine{} and beyond, a more detailed memory model like ours is
needed, see also Section~\ref{section:challenges} and Defect
Report \#260 and \#451~\cite{iso:09}.

Another interesting difference is that Norrish represents abstract values
(integers, pointers and structs) as sequences of bytes instead of mathematical
values.
Due to this, padding bytes retain their value while structs are copied.
This is not faithful to the \Cninetynine{} standard and beyond.

\paragraph{Leroy \etal{} (2006)}
Leroy \etal{} have formalized a significant part of \C{} using the
\Coq{} proof assistant~\cite{ler:06,ler:09:cacm}.
Their part of \C, which is called \CompCertC{}, covers most major features of \C{}
and can be compiled into assembly (\PowerPC, \ARM{} and \xeightysix) using a
compiler written in \Coq.
Their compiler, called \CompCert{}, has been proven correct with respect to the
\CompCertC{} and assembly semantics.

The goal of \CompCert{} is essentially different from \CHtwoO's.
What can be proven with respect to the \CompCert{} semantics does not have
to hold for \emph{any} \Celeven{} compiler, it just
has to hold for the \CompCert{} compiler.
\CompCert{} is therefore in its semantics allowed to restrict implementation
defined behaviors to be very specific (for example, it uses 32 bits
\lstinline|int|s since it targets only 32-bits
computing architectures) and allowed to give a defined semantics to various
undefined behaviors (such as sequence point violations, violations of
effective types, and certain uses of dangling pointers).

The \CompCert{} memory model is used by all languages (from \C{} until
assembly) of the \CompCert{} compiler~\cite{ler:bla:08,ler:app:bla:ste:12}.
The \CompCert{} memory is a finite partial function from object
identifiers to objects.
Each local, global and static variable, and invocation of \lstinline|malloc| is
associated with a unique object identifier of a separate object in memory.
We have used the same approach in \CHtwoO, but there are some important
differences.
The paragraphs below discuss the relation of \CHtwoO{} with the first and
second version of the \CompCert{} memory model.

\paragraph{Leroy and Blazy (2008)}
In the first version of the \CompCert{} memory model~\cite{ler:bla:08}, objects
were represented as arrays of type-annotated fragments of base values.
Examples of bytes are thus ``the 2nd byte of the short 13'' or ``the 3rd byte of
the pointer \(\pair o i\)''.
Pointers were represented as pairs \(\pair o i\) where \(o\) is an
object identifier and \(i\) the byte offset into the object \(o\).

Since bytes are annotated with types and could only be retrieved from memory
using an expression of matching type, effective types on the level of
base types are implicitly described.
However, this does not match the \Celeven{} standard.
For example, Leroy and Blazy do assign the return value 11 to the following
program:

\begin{lstlisting}
struct S1 { int x; };
struct S2 { int y; };
int f(struct S1 *p, struct S2 *q) {
  p->x = 10;
  q->y = 11;
  return p->x;
}
int main() {
  union U { struct S1 s1; struct S2 s2; } u;
  printf("%d\n", f(&u.s1, &u.s2));
}
\end{lstlisting}

This code strongly resembles example~\cite[6.5.2.3p9]{iso:12} from the
\Celeven{} standard, which is stated to have undefined
behavior\footnote{We have modified the example from the
standard slightly in order to trigger optimizations by \GCC{} and \clang.}.
\GCC{} and \clang{} optimize this code to print 10, which 
differs from the value assigned by Leroy and Blazy.

Apart from assigning too much defined behavior, Leroy and Blazy's treatment of
effective types also prohibits any form of ``bit twiddling''.

Leroy and Blazy have introduced the notion of memory
injections in~\cite{ler:bla:08}.
This notion allows one to reason about memory transformations in an elegant way.
Our notion of memory refinements (Section~\ref{section:refinements})
generalize the approach of Leroy and Blazy to a tree based memory model.

\paragraph{Leroy \etal{} (2012)}
The second version of \CompCert{} memory model~\cite{ler:app:bla:ste:12} is
entirely untyped and is extended with permissions.
Symbolic bytes are only used for pointer values and indeterminate storage,
whereas integer and floating point values are represented as numerical bytes
(integers between 0 and \(2^8-1\)).

We have extended this approach by analogy to bit-representations, representing
indeterminate storage and pointer values using symbolic
bits, and integer values using concrete bits.
This choice is detailed in Section~\ref{section:bits}.

As an extension of \CompCert{}, Robert and Leroy have formally proven soundness
of an alias analysis~\cite{rob:ler:12}.
Their alias analysis is untyped and operates on the RTL intermediate language
of \CompCert{}.

Beringer \etal~\cite{ber:ste:doc:app:14} have developed an extension of
\CompCert's memory injections to reason about program transformations in the
case of separate compilation.
The issues of separate compilation are orthogonal to those that we consider.

\paragraph{Appel \etal{} (2014)}
The \VSTfull{} (\VST) by Appel \etal{} provides a higher-order separation logic
for Verifiable \C, which is a variant of \CompCert's intermediate language
\Clight~\cite{app:14}.

The \VST{} is intended to be used together with the \CompCert{} compiler.
It gives very strong guarantees when done so.
The soundness proof of the \VST{} in conjunction with the correctness proof of
the \CompCert{} compiler ensure that the proven properties also hold for the
generated assembly.

In case the verified program is compiled with a compiler different from
\CompCert{}, the trust in the program is still increased, but no full
guarantees can be given.
That is caused by the fact that \CompCert's intermediate language \Clight{}
uses a specific evaluation order and assigns defined behavior to many
undefined behaviors of the \Celeven{} standard.
For example, \Clight{} assigns defined behavior to violations of effective
types and sequence point violations.
The \VST{} inherits these defined behaviors from \CompCert{} and allows one to
use them in proofs.

Since the \VST{} is linked to \CompCert, it uses \CompCert's coarse
permission system on the level of the operational semantics.
Stewart and Appel~\cite[Chapter 42]{app:14} have introduced a way
to use a more fine grained permission system at the level of the separation
logic without having to modify the \Clight{} operational semantics.
Their approach shows its merits when used for concurrency, in which case the
memory model contains \emph{ghost} data related to the conditions of
locks~\cite{hob:08,hob:app:nar:08}.

\paragraph{Besson \etal{} (2014)}
Besson \etal{} have proposed an extension of the \CompCert{} memory
model that assigns a defined semantics to operations that rely on the numerical values
of uninitialized memory and pointers~\cite{bes:bla:wil:14}.

Objects in their memory model constitute of lazily evaluated values
described by symbolic expressions.
These symbolic expressions are used to delay the evaluation of operations on
uninitialized memory and pointer values.
Only when a concrete value is needed (for example in case of the controlling
expression of an \cifthenelse, \cfor, or \cwhile{} statement), the symbolic
expression is normalized.
Consider:

\begin{lstlisting}
int x, *p = &x;
int y = ((unsigned char*)p)[1] | 1;
// y has symbolic value "2nd pointer byte of p" | 1
if (y & 1) printf("one\n"); // unique normalization -> OK
if (y & 2) printf("two\n"); // no unique normalization -> bad
\end{lstlisting}

The value of \lstinline$((unsigned char*)p)[1] | 1$ is not evaluated eagerly.
Instead, the assignment to \lstinline|y| stores a symbolic expression denoting
this value.
During the execution of the first \cif{} statement, the actual value of
\lstinline|y & 1| is needed.
In this case, \lstinline$y & 1$ has the value 1 for any possible numerical
value of \lstinline|((unsigned char*)p)[1]|.
As a result, the string \lstinline|one| is printed.

The semantics of Besson \etal{} is deterministic by definition.
Normalization of symbolic expressions has defined behavior if and only if the
expression can be normalized to a unique value under any choice of numeral
values for pointer representations and uninitialized storage.
In the second \cif{} statement this is not the case.

The approach of Besson \etal{} gives a semantics to some programming techniques
that rely on the numerical representations of pointers and uninitialized memory.
For example, it gives an appropriate semantics to pointer tagging in which
unused bits of a pointer representation are used to store additional information.

However, as already observed by Kang \etal~\cite{kan:hur:man:gar:zda:vaf:15},
Besson \etal{} do not give a semantics to many other useful cases.
For example, printing the object representation of a struct, or computing
the hash of a pointer value, is inherently non-deterministic.
The approach of Besson \etal{} assigns undefined behavior to these use cases.

The goal of Besson \etal{} is inherently different from ours.
Our goal is to describe the \Celeven{} standard faithfully whereas Besson
\etal{} focus on \defacto{} versions of \C.
They intentionally assign defined behavior to many constructs involving
uninitialized memory that are clearly undefined according to the \Celeven{}
standard, but that are nonetheless faithfully compiled by specific compilers.

\paragraph{Ellison and Ro\c{s}u (2012)}
Ellison and Ro\c{s}u~\cite{ell:ros:12,ell:12} have
developed an executable semantics of the \Celeven{} standard using the
\Kframework\footnote{This work has been superseded by Hathhorn
\etal~\cite{hat:ell:ros:15}, which is described below.}.
Their semantics is very comprehensive and describes all features of a
freestanding \C{} implementation~\cite[4p6]{iso:12} including some parts of
the standard library.
It furthermore has been thoroughly tested against test suites (such as the
\GCC{} torture test suite), and has been used as an oracle for
compiler testing~\cite{reg:che:cuo:eid:ell:yan:12}.

Ellison and Ro\c{s}u support more \C{} features than we do, but
they do not have infrastructure for formal proofs, and thus have
not established any metatheoretical properties about their semantics.
Their semantics, despite being written in a formal framework, should more
be seen as a debugger, a state space search tool, or possibly, as a model
checker.
It is unlikely to be of practical use in proof assistants because it is defined
on top of a large \C{} abstract syntax and uses a rather ad-hoc execution state
that contains over 90 components.

Similar to our work, Ellison and Ro\c{s}u's goal is to \emph{exactly}
describe the \Celeven{} standard.
However, for some programs their semantics is less precise than ours, which
is mainly caused by their memory model, which is less principled than ours.
Their memory model is based on \CompCert's: it is essentially
a finite map of objects consisting of unstructured arrays of bytes.

\paragraph{Hathhorn \etal{} (2015)}
Hathhorn \etal~\cite{hat:ell:ros:15} have extended the work of Ellison and Ro\c{s}u to handle
more underspecification of \Celeven.
Most importantly, the memory model has been extended and support for the type
qualifiers \lstinline|const|, \lstinline|restrict| and \lstinline|volatile| has
been added.

Hathhorn \etal{} have extended the original memory model (which was based on
\CompCert's) with decorations to handle effective types, restrictions on
padding and the \lstinline|restrict| qualifier.
Effective types are modeled by a map that associates a type to each object.
Their approach is less fine-grained than ours and is unable to account for
active variants of unions.
It thus does not assign undefined behavior to important violations of effective
types and in turn does not allow compilers to perform optimizations based
on type-based alias analysis.
For example:

\begin{lstlisting}
// Undefined behavior in case f is called with aliased
// pointers due to effective types
int f(short *p, int *q) { *p = 10; *q = 11; return *p; }
int main() {
  union { short x; int y; } u = { .y = 0 };
  return f(&u.x, &u.y);
}
\end{lstlisting}

The above program has undefined behavior due to a violation of effective types.
This is captured by our tree based memory model, but Hathhorn \etal{} require
the program to return the value 11.
When compiled with \GCC{} or \clang{} with optimization level \lstinline|-O2|,
the compiled program returns the value 10.

Hathhorn \etal{} handle restrictions on padding bytes in the case of unions,
but not in the case of structs.
For example, the following program returns the value 1 according to their
semantics, whereas it has unspecified behavior according to
the \Celeven{} standard~\cite[6.2.6.1p6]{iso:12}
(see also Section~\ref{section:challenge:padding}):

\begin{lstlisting}
struct S { char a; int b; } s;
((unsigned char*)(&s))[1] = 1;
s.a = 10; // Makes the padding bytes of 's' indeterminate
return ((unsigned char*)(&s))[1];
\end{lstlisting}

The restrictions on paddings bytes are implicit in our memory model based on
structured trees, and thus handled correctly.
The above examples provide evidence that a structured approach,
especially combined with metatheoretical results, is more reliable than
depending on ad-hoc decorations.

\paragraph{Kang \etal{} (2015)}
Kang \etal{}~\cite{kan:hur:man:gar:zda:vaf:15} have proposed a memory
model that gives a semantics to pointer to integer casts.
Their memory model uses a combination of numerical and symbolic representations
of pointer values (whereas \CompCert{} and \CHtwoO{} always represent pointer
values symbolically).
Initially each pointer is represented symbolically, but whenever the numerical
representation of a pointer is needed (due to a pointer to integer cast), it is
non-deterministically \emph{realized}.

The memory model of Kang \etal{} gives a semantics to pointer to integer casts
while allowing common compiler optimizations that are invalid in a naive
low-level memory model.
They provide the following motivating example:

\begin{lstlisting}
void g(void) { ... }
int f(void) {
  int a = 0;
  g();
  return a;
}
\end{lstlisting}

In a concrete memory model, there is the possibility that the function
\lstinline|g| is able to \emph{guess} the numerical representation of
\lstinline|&a|, and thereby access or even modify \lstinline|a|.
This is undesirable, because it prevents the widely used optimization of
constant propagation, which optimizes the variable \lstinline|a| out.

In the \CompCert{} and \CHtwoO{} memory model, where pointers are represented
symbolically, it is guaranteed that \lstinline|f| has exclusive control over
\lstinline|a|.
Since \lstinline|&a| has not been leaked, \lstinline|g| can impossibly access
\lstinline|a|.
In the memory model of Kang \etal{} a pointer will only be given a numerical
representation when it is cast to an integer.
In the above code, no such casts appear, and \lstinline|g| cannot access
\lstinline|a|.

The goal of Kang \etal{} is to give a unambiguous mathematical model for
pointer to integer casts, but not necessarily to comply with \Celeven{}
or existing compilers.
Although we think that their model is a reasonable choice, it is unclear
whether it is faithful to the \Celeven{} standard in the context
of Defect Report \#260~\cite{iso:09}.
Consider:

\begin{lstlisting}
int x = 0, *p = 0;
for (uintptr_t i = 0; ; i++) {
  if (i == (uintptr_t)&x) {
    p = (int*)i; break;
  }
}
*p = 15;
printf("%d\n", x);
\end{lstlisting}

Here we loop through the range of integers of type \lstinline|uintptr_t| until
we have found the integer representation \lstinline|i| of \lstinline|&x|, which
we then assign to the pointer \lstinline|p|.

When compiled with \lstinline|gcc -O2| (version 4.9.2), the generated assembly
no longer contains a loop, and the pointers \lstinline|p| and \lstinline|q| are
assumed not to alias.
As a result, the program prints the old value of \lstinline|x|, namely
\lstinline|0|.
In the memory model of Kang \etal{} the pointer obtained via the cast
\lstinline|(int*)i| is exactly the same as \lstinline|&x|.
In their model the program thus has defined behavior and is required to
print \lstinline|15|.

We have reported this issue to the \GCC{} bug tracker\footnote{
See \url{https://gcc.gnu.org/bugzilla/show_bug.cgi?id=65752 }.}.
However it unclear whether the \GCC{} developers consider this a bug or not.
Some developers seem to believe that this program has undefined behavior and
that \GCC's optimizations are thus justified.
Note that the cast \lstinline|(intptr_t)&x| is already forbidden by the
type system of \CHtwoO.

\section{Conclusion}
\label{section:conclusion}

In this paper we have given a formal description of a significant part of the
non-concurrent \Celeven{} memory model.
This formal description has been used in~\cite{kre:wie:15,kre:15:phd} as part
of an an operational, executable and axiomatic semantics of \C.
On top of this formal description, we have provided a comprehensive collection
of metatheoretical results.
All of these results have been formalized using the Coq proof assistant.

It would be interesting to investigate whether our memory model can be
used to help the standard committee to improve future versions of the standard.
For example, whether it could help to improve the standard's prose description
of effective types.
As indicated on page~\pageref{page:effective_types} of
Section~\ref{section:introduction}, the standard's
description is not only ambiguous, but also does not cover its intent to enable
type-based alias analysis.
The description of our memory model is unambiguous and allows
one to express intended consequences formally.
We have formally proven soundness of an abstract version of type-based alias
analysis with respect to our memory model (Theorem~\ref{theorem:aliasing}).

An obvious direction for future work is to extend the memory model with
additional features.
We give an overview of some features of \Celeven{} that are
absent.

\begin{itemize}
\item \textbf{Floating point arithmetic.}
	Representations of floating point numbers and the behaviors of floating
	point arithmetic are subject to a considerable amount of implementation
	defined behavior~\cite[5.2.4.2.2]{iso:12}.
	
	First of all, one could restrict to IEEE-754 floating point arithmetic,
	which has a clear specification~\cite{iee:08}
	and a comprehensive formalization in \Coq~\cite{bol:mel:11}.
	Boldo \etal{} have taken this approach in the context of
	\CompCert~\cite{bol:jou:ler:mel:13} and we see no fundamental problems
	applying it to \CHtwoO{} as well.
	
	Alternatively, one could consider formalizing all implementation defined
	aspects of the description of floating arithmetic in the \Celeven{} standard.
\item \textbf{Bitfields.}
	Bitfields are fields of struct types that occupy individual
	bits~\cite[6.7.2.1p9]{iso:12}.
	We do not foresee fundamental problems adding bitfields to \CHtwoO{} as
	bits already constitute the smallest unit of storage in our memory model.
\item \textbf{Untyped malloc.}
	\CHtwoO{} supports dynamic memory allocation via an operator
	\(\EAlloc \tau e\) close to \Cplusplus's \lstinline|new| operator.
	The \(\EAlloc \tau e\) operator yields a \(\TPtr \tau\) pointer
	to storage for an \(\tau\)-array of length \(e\).
	This is different from \C's \lstinline|malloc| function that yields a
	\lstinline|void*| pointer to storage of unknown type~\cite[7.22.3.4]{iso:12}.
	
	Dynamic memory allocation via the untyped \lstinline|malloc| function 
	is closely related to unions and effective types.
	Only when dynamically allocated storage is actually used, it will
	receive an effective type.
	We expect one could treat \lstinline|malloc|ed objects as unions that
	range over all possible types that fit.
\item \textbf{Restrict qualifiers.}
	The \lstinline|restrict| qualifier can be applied to
	any pointer type to
	express that the pointers do not alias.
	Since the description in the \Celeven{} standard~\cite[6.7.3.1]{iso:12} is
	ambiguous (most notably, it is unclear how it interacts with nested
	pointers and data types), formalization and metatheoretical proofs may
	provide prospects for clarification.

\item \textbf{Volatile qualifiers.}
	The \lstinline|volatile| qualifier can be applied to any type to indicate
	that its value may be changed by an external process.
	It is meant to prevent compilers from optimizing away data accesses or
	reordering these~\cite[footnote 134]{iso:12}.
	Volatile accesses should thus be considered as a form of I/O.

\item \textbf{Concurrency and atomics.}
	Shared-memory concurrency and atomic operations are the main omission from
	the \Celeven{} standard in the \CHtwoO{} semantics.
	Although shared-memory concurrency is a relatively new addition to the \C{}
	and \Cplusplus{} standards, there is already a large body of ongoing work in
	this direction, see for example~\cite{sew:sar:owe:zap:myr:10,%
bat:owe:sar:sew:web:11,sev:vaf:zap:jag:sew:13,vaf:bal:cha:mor:zap:15,
bat:mem:nie:pic:sew:05}.
	These works have led to improvements of the standard text.

	There are still important open problems in the area of concurrent
	memory models for already small sublanguages of
	\C~\cite{bat:mem:nie:pic:sew:05}.
	Current memory models for these sublanguages involve just features specific
	to threads and atomic operations whereas we have focused on
	structs, unions, effective types and indeterminate memory.
	We hope that both directions are largely orthogonal and will
	eventually merge into a fully fledged \Celeven{} memory model and semantics.
\end{itemize}

\paragraph{Acknowledgments.}
I thank my supervisors Freek Wiedijk and Herman Geuvers for their helpful
suggestions.
I thank Xavier Leroy, Andrew Appel, Lennart Beringer and Gordon Stewart for
many discussions on the \CompCert{} memory model, and the anonymous reviewers
for their feedback.
This work is financed by the Netherlands Organisation for Scientific 
Research (NWO), project 612.001.014.


\begin{thebibliography}{10}

\bibitem{aff:mar:13}
R.~Affeldt and N.~Marti.
\newblock {Towards formal verification of TLS network packet processing written
  in C}.
\newblock In {\em PLPV}, pages 35--46, 2013.

\bibitem{aff:sak:14}
R.~Affeldt and K.~Sakaguchi.
\newblock {An Intrinsic Encoding of a Subset of C and its Application to TLS
  Network Packet Processing}.
\newblock {\em JFR}, 7(1), 2014.

\bibitem{app:14}
A.~W. Appel, editor.
\newblock {\em {Program Logics for Certified Compilers}}.
\newblock Cambridge University Press, 2014.

\bibitem{bat:mem:nie:pic:sew:05}
M.~Batty, K.~Memarian, K.~Nienhuis, J.~Pichon-Pharabod, and P.~Sewell.
\newblock {The Problem of Programming Language Concurrency Semantics}, 2015.
\newblock To appear in ESOP.

\bibitem{bat:owe:sar:sew:web:11}
M.~Batty, S.~Owens, S.~Sarkar, P.~Sewell, and T.~Weber.
\newblock {Mathematizing C++ concurrency}.
\newblock In {\em POPL}, pages 55--66, 2011.

\bibitem{ben:jen:cie:bir:11}
J.~Bengtson, J.~B. Jensen, F.~Sieczkowski, and L.~Birkedal.
\newblock {Verifying Object-Oriented Programs with Higher-Order Separation
  Logic in Coq}.
\newblock In {\em ITP}, volume 6898 of {\em LNCS}, pages 22--38, 2011.

\bibitem{ber:ste:doc:app:14}
L.~Beringer, G.~Stewart, R.~Dockins, and A.~W. Appel.
\newblock {Verified Compilation for Shared-Memory C}.
\newblock In {\em ESOP}, volume 8410 of {\em LNCS}, pages 107--127, 2014.

\bibitem{bes:bla:wil:14}
F.~Besson, S.~Blazy, and P.~Wilke.
\newblock {A Precise and Abstract Memory Model for {C} Using Symbolic Values}.
\newblock In {\em APLAS}, volume 8858 of {\em LNCS}, pages 449--468, 2014.

\bibitem{bol:jou:ler:mel:13}
S.~Boldo, J.-H. Jourdan, X.~Leroy, and G.~Melquiond.
\newblock {A Formally-Verified {C} Compiler Supporting Floating-Point
  Arithmetic}.
\newblock In {\em ARITH}, pages 107--115, 2013.

\bibitem{bol:mel:11}
S.~Boldo and G.~Melquiond.
\newblock {Flocq: {A} Unified Library for Proving Floating-Point Algorithms in
  Coq}.
\newblock In {\em ARITH}, pages 243--252, 2011.

\bibitem{bor:cal:hea:par:05}
R.~Bornat, C.~Calcagno, P.~W. O'Hearn, and M.~J. Parkinson.
\newblock {Permission Accounting in Separation Logic}.
\newblock In {\em POPL}, pages 259--270, 2005.

\bibitem{boy:03}
J.~Boyland.
\newblock {Checking Interference with Fractional Permissions}.
\newblock In {\em SAS}, volume 2694 of {\em LNCS}, pages 55--72, 2003.

\bibitem{cal:hea:yan:07}
C.~Calcagno, P.~W. O'Hearn, and H.~Yang.
\newblock {Local Action and Abstract Separation Logic}.
\newblock In {\em LICS}, pages 366--378, 2007.

\bibitem{coh:mos:tob:sch:09}
E.~Cohen, M.~Moskal, S.~Tobies, and W.~Schulte.
\newblock {A Precise Yet Efficient Memory Model For {C}}.
\newblock {\em ENTCS}, 254:85--103, 2009.

\bibitem{coq:15}
{Coq Development Team}.
\newblock {\em The Coq Proof Assistant Reference Manual}, 2015.
\newblock Available at \url{https://coq.inria.fr/doc/}.

\bibitem{dij:68:seq}
E.~W. Dijkstra.
\newblock {Cooperating sequential processes}.
\newblock In {\em Programming Languages: NATO Advanced Study Institute}, pages
  43--112. Academic Press, 1968.

\bibitem{doc:hob:app:09}
R.~Dockins, A.~Hobor, and A.~W. Appel.
\newblock {A Fresh Look at Separation Algebras and Share Accounting}.
\newblock In {\em APLAS}, volume 5904 of {\em LNCS}, pages 161--177, 2009.

\bibitem{ell:12}
C.~Ellison.
\newblock {\em {A Formal Semantics of {C} with Applications}}.
\newblock PhD thesis, University of Illinois, 2012.

\bibitem{ell:ros:12}
C.~Ellison and G.~Ro\c{s}u.
\newblock {An executable formal semantics of C with applications}.
\newblock In {\em POPL}, pages 533--544, 2012.

\bibitem{gnu:11}
GCC.
\newblock {The GNU Compiler Collection}.
\newblock Website, available at \url{http://gcc.gnu.org/}.

\bibitem{gre:lim:adr:kle:14}
D.~Greenaway, J.~Lim, J.~Andronick, and G.~Klein.
\newblock {Don't Sweat the Small Stuff: Formal Verification of {C} Code Without
  the Pain}.
\newblock In {\em PLDI}, pages 429--439, 2014.

\bibitem{hat:ell:ros:15}
C.~Hathhorn, C.~Ellison, and G.~Ro\c{s}u.
\newblock {Defining the Undefinedness of C}.
\newblock In {\em PLDI}, pages 336--345, 2015.

\bibitem{hob:08}
A.~Hobor.
\newblock {\em {Oracle Semantics}}.
\newblock PhD thesis, Princeton University, 2008.

\bibitem{hob:app:nar:08}
A.~Hobor, A.~W. Appel, and F.~Z. Nardelli.
\newblock {Oracle Semantics for Concurrent Separation Logic}.
\newblock In {\em ESOP}, volume 4960 of {\em LNCS}, pages 353--367, 2008.

\bibitem{iee:08}
{IEEE Computer Society}.
\newblock {\em {754-2008: IEEE Standard for Floating Point Arithmetic}}.
\newblock IEEE, 2008.

\bibitem{iso:09}
{ISO}.
\newblock {WG14 Defect Report Summary}.
\newblock Website, available at
  \url{http://www.open-std.org/jtc1/sc22/wg14/www/docs/}.

\bibitem{iso:12}
{ISO}.
\newblock {\em {ISO/IEC 9899-2011: Programming languages -- C}}.
\newblock ISO Working Group 14, 2012.

\bibitem{kan:hur:man:gar:zda:vaf:15}
J.~Kang, C.-K. Hur, W.~Mansky, D.~Garbuzov, S.~Zdancewic, and V.~Vafeiadis.
\newblock {A Formal C Memory Model Supporting Integer-Pointer Casts}.
\newblock In {\em PLDI}, 2015.

\bibitem{kle:kol:boy:12}
G.~Klein, R.~Kolanski, and A.~Boyton.
\newblock {Mechanised Separation Algebra}.
\newblock In {\em ITP}, volume 7406 of {\em LNCS}, pages 332--337, 2012.

\bibitem{kre:13:cpp}
R.~Krebbers.
\newblock {Aliasing restrictions of C11 formalized in Coq}.
\newblock In {\em CPP}, volume 8307 of {\em LNCS}, 2013.

\bibitem{kre:14:popl}
R.~Krebbers.
\newblock {An Operational and Axiomatic Semantics for Non-determinism and
  Sequence Points in C}.
\newblock In {\em POPL}, pages 101--112, 2014.

\bibitem{kre:14:vstte}
R.~Krebbers.
\newblock {Separation algebras for C verification in Coq}.
\newblock In {\em VSTTE}, volume 8471 of {\em LNCS}, pages 150--166, 2014.

\bibitem{kre:15:phd}
R.~Krebbers.
\newblock {\em {The C standard formalized in Coq}}.
\newblock PhD thesis, Radboud University, 2015.
\newblock Manuscript accepted by the committee, available online
  at~\url{http://robbertkrebbers.nl/research/thesis_draft.pdf}.

\bibitem{kre:ler:wie:14}
R.~Krebbers, X.~Leroy, and F.~Wiedijk.
\newblock {Formal C semantics: CompCert and the C standard}.
\newblock In {\em ITP}, volume 8558 of {\em LNCS}, pages 543--548, 2014.

\bibitem{kre:wie:11}
R.~Krebbers and F.~Wiedijk.
\newblock {A Formalization of the C99 Standard in HOL, Isabelle and Coq}.
\newblock In {\em CICM}, volume 6824 of {\em LNCS}, pages 297--299, 2011.

\bibitem{kre:wie:13}
R.~Krebbers and F.~Wiedijk.
\newblock {Separation Logic for Non-local Control Flow and Block Scope
  Variables}.
\newblock In {\em FoSSaCS}, volume 7794 of {\em LNCS}, pages 257--272, 2013.

\bibitem{kre:wie:15}
R.~Krebbers and F.~Wiedijk.
\newblock {A Typed C11 Semantics for Interactive Theorem Proving}.
\newblock In {\em CPP}, pages 15--27, 2015.

\bibitem{ler:06}
X.~Leroy.
\newblock Formal certification of a compiler back-end or: programming a
  compiler with a proof assistant.
\newblock In {\em POPL}, pages 42--54, 2006.

\bibitem{ler:09:cacm}
X.~Leroy.
\newblock Formal verification of a realistic compiler.
\newblock {\em CACM}, 52(7):107--115, 2009.

\bibitem{ler:app:bla:ste:12}
X.~Leroy, A.~W. Appel, S.~Blazy, and G.~Stewart.
\newblock {The CompCert Memory Model, Version 2}.
\newblock Research report RR-7987, INRIA, 2012.
\newblock Revised version available as Chapter 32 of~\cite{app:14}.

\bibitem{ler:bla:08}
X.~Leroy and S.~Blazy.
\newblock Formal verification of a {C}-like memory model and its uses for
  verifying program transformations.
\newblock {\em JAR}, 41(1):1--31, 2008.

\bibitem{mac:01}
N.~Maclaren.
\newblock {What is an Object in C Terms?}, 2001.
\newblock Mailing list message, available at
  \url{http://www.open-std.org/jtc1/sc22/wg14/9350}.

\bibitem{mon:shi:13}
J.~Monin and X.~Shi.
\newblock {Handcrafted Inversions Made Operational on Operational Semantics}.
\newblock In {\em ITP}, volume 7998 of {\em LNCS}, pages 338--353, 2013.

\bibitem{nor:98}
M.~Norrish.
\newblock {\em {C formalised in HOL}}.
\newblock PhD thesis, University of Cambridge, 1998.

\bibitem{nor:99}
M.~Norrish.
\newblock {Deterministic Expressions in C}.
\newblock In {\em ESOP}, volume 1576 of {\em LNCS}, pages 147--161, 1999.

\bibitem{hea:04}
P.~W. O'Hearn.
\newblock {Resources, Concurrency and Local Reasoning}.
\newblock In {\em CONCUR}, volume 3170 of {\em LNCS}, pages 49--67, 2004.

\bibitem{hea:rey:yan:01}
P.~W. O'Hearn, J.~C. Reynolds, and H.~Yang.
\newblock {Local Reasoning about Programs that Alter Data Structures}.
\newblock In {\em CSL}, volume 2142 of {\em LNCS}, pages 1--19, 2001.

\bibitem{ram:dos:ler:11}
T.~Ramananandro, G.~Dos~Reis, and X.~Leroy.
\newblock {Formal verification of object layout for {C++} multiple
  inheritance}.
\newblock In {\em POPL}, pages 67--80, 2011.

\bibitem{reg:che:cuo:eid:ell:yan:12}
J.~Regehr, Y.~Chen, P.~Cuoq, E.~Eide, C.~Ellison, and X.~Yang.
\newblock {Test-case reduction for C compiler bugs}.
\newblock In {\em PLDI}, pages 335--346, 2012.

\bibitem{rob:ler:12}
V.~Robert and X.~Leroy.
\newblock {A Formally-Verified Alias Analysis}.
\newblock In {\em CPP}, volume 7679 of {\em LNCS}, pages 11--26, 2012.

\bibitem{ros:fri:95}
J.~G. Rossie and D.~P. Friedman.
\newblock {An Algebraic Semantics of Subobjects}.
\newblock In {\em OOPSLA}, pages 187--199, 1995.

\bibitem{sev:vaf:zap:jag:sew:13}
J.~Sevc\'{\i}k, V.~Vafeiadis, F.~Z. Nardelli, S.~Jagannathan, and P.~Sewell.
\newblock {CompCertTSO: A Verified Compiler for Relaxed-Memory Concurrency}.
\newblock {\em JACM}, 60(3):22, 2013.

\bibitem{sew:sar:owe:zap:myr:10}
P.~Sewell, S.~Sarkar, S.~Owens, F.~Z. Nardelli, and M.~O. Myreen.
\newblock {x86-TSO: a rigorous and usable programmer's model for x86
  multiprocessors}.
\newblock {\em CACM}, 53(7):89--97, 2010.

\bibitem{soz:10}
M.~Sozeau.
\newblock {A New Look at Generalized Rewriting in Type Theory}.
\newblock {\em JFR}, 2(1), 2010.

\bibitem{spi:wee:11}
B.~Spitters and E.~van~der Weegen.
\newblock {Type Classes for Mathematics in Type Theory}.
\newblock {\em Mathematical Structures in Computer Science}, 21(4):795--825,
  2011.

\bibitem{tuc:kle:nor:07}
H.~Tuch, G.~Klein, and M.~Norrish.
\newblock Types, bytes, and separation logic.
\newblock In {\em POPL}, pages 97--108, 2007.

\bibitem{vaf:bal:cha:mor:zap:15}
V.~Vafeiadis, T.~Balabonski, S.~Chakraborty, R.~Morisset, and F.~Z. Nardelli.
\newblock {Common compiler optimisations are invalid in the C11 memory model
  and what we can do about it}.
\newblock In {\em POPL}, pages 209--220, 2015.

\end{thebibliography}
\end{document}